%% file: main.tex
\documentclass[11pt]{article}
\input{arxiv}

\input{commands}

\title{A Tutorial on the Non-Asymptotic Theory of System Identification}
\begin{document}
\maketitle

\input{abstract}

\tableofcontents

\input{sections/notation}
\input{sections/introduction}

\input{sections/concentrationineqs}

\input{sections/lowertail}
\input{sections/linearregression}

\input{sections/sysid}

\input{sections/basicineq}

\input{sections/extension}
\input{acks}

\bibliographystyle{plainnat}
\bibliography{main.bib}

\appendix
\onecolumn

\input{sections/proofhw}

\input{sections/twosided}

\input{sections/anticoncproofs}

\input{sections/selfnormproofs}
\input{sections/proof_sysid}

\input{sections/proofsec67}
\end{document}

%% file: arxiv.tex
\usepackage{authblk}
\author[1]{Ingvar Ziemann}
\author[2]{Anastasios Tsiamis}
\author[1]{Bruce Lee}
\author[3]{Yassir Jedra}
\author[1]{Nikolai Matni}
\author[1]{George J. Pappas}

\affil[1]{University of Pennsylvania}
\affil[2]{ETH Zürich}
\affil[3]{KTH Royal Institute of Technology}

\usepackage[utf8]{inputenc} 
\usepackage[T1]{fontenc} 

\usepackage{fullpage}

\usepackage{url}            
\usepackage{booktabs}       
\usepackage{amsfonts}       
\usepackage{nicefrac}       
\usepackage{microtype}      
\usepackage[dvipsnames]{xcolor}

\usepackage{natbib} 

\usepackage{enumitem}
\usepackage{amsthm}
\usepackage{amssymb}
\usepackage{amsmath}
\usepackage{mathrsfs}
\usepackage{hyperref}
\hypersetup{
    pdftoolbar=true,        
    pdfmenubar=true,        
    pdffitwindow=false,     
    pdfstartview={FitH},    
    pdfnewwindow=true,      
    colorlinks=true,       
    linkcolor=blue,          
    citecolor=ForestGreen,        
    filecolor=magenta,      
    urlcolor=red,           
    breaklinks=false,
}

\usepackage{cleveref}
\usepackage{thmtools}
\usepackage{thm-restate}

\numberwithin{equation}{section}

\date{}

%% file: commands.tex
\renewcommand\qedsymbol{$\blacksquare$}

\newcommand{\e}{\varepsilon}

\newcommand{\E}{\mathbf{E}}
\newcommand{\Prob}{\mathbf{P}}

\DeclareMathOperator*{\argmin}{argmin}

\DeclareMathOperator{\blkdiag}{blkdiag}

\DeclareMathOperator{\tr}{tr}
\DeclareMathOperator{\VEC}{\mathsf{vec}}

\DeclareMathOperator{\I}{\mathtt{I}}

\newcommand{\T}{\mathsf{T}}

\newcommand{\Z}{\mathbb{Z}}
\newcommand{\calA}{\mathcal{A}}
\newcommand{\calB}{\mathcal{B}}
\newcommand{\calN}{\mathcal{N}}
\newcommand{\calF}{\mathcal{F}}

\newcommand{\calE}{\mathcal{E}}

\newcommand{\calT}{\mathcal{T}}

\newcommand{\scrF}{\mathscr{F}}

\newcommand{\sfX}{\mathsf{X}}

\newcommand{\sfP}{\mathsf{P}}

\newcommand{\dx}{d_{\mathsf{X}}}
\newcommand{\dy}{d_{\mathsf{Y}}}

\newcommand{\du}{d_{\mathsf{U}}}
\newcommand{\dw}{d_{\mathsf{W}}}
\newcommand{\da}{d_{\mathsf{\eta}}}
\newcommand{\deta}{d_{\mathsf{\eta}}}

\newcommand{\R}{\mathbb{R}}
\newcommand{\N}{\mathbb{N}}
\renewcommand{\Pr}{\mathbf{P}}
\newcommand{\C}{\ensuremath{\mathcal{C}}}

\newcommand{\iid}{iid}
\newcommand{\set}[1]{\ensuremath{\left\{ #1\right\}}}

\newcommand{\matr}[1]{\ensuremath{\begin{bmatrix} #1 \end{bmatrix}}}

\newcommand{\paren}[1]{\ensuremath{\left( #1\right)}}
\newcommand{\curly}[1]{\mathopen{}\left\{#1\right\}\mathclose{}}
\newcommand{\brac}[1]{\mathopen{}\left[#1\right]\mathclose{}}
\newcommand{\norm}[1]{\lVert #1 \rVert}
\newcommand{\bignorm}[1]{\left\lVert #1 \right\rVert}

\newcommand{\opnorm}[1]{\| #1 \|_{\mathsf{op}}}
\newcommand{\bigopnorm}[1]{\left\| #1 \right\|_{\mathsf{op}}}

\newcommand{\As}{A^\star}
\newcommand{\Acs}{A^\star_\mathsf{cl}}
\newcommand{\Bs}{B^\star}
\newcommand{\Cs}{C^\star}
\newcommand{\Ks}{F^\star}
\newcommand{\Ps}{P^\star}
\newcommand{\RIC}{\mathsf{RIC}}
\newcommand{\Ms}{\theta^\star}
\newcommand{\Mls}{\widehat{\theta}}
\newcommand{\ECov}{\widehat{\Sigma}}
\newcommand{\CovX}{\Sigma}
\newcommand{\CovXState}{\Sigma}
\newcommand{\ECovZ}{\widehat{\Sigma}}
\newcommand{\CovZ}{\Sigma}
\newcommand{\bfL}{\mathbf{L}}
\newcommand{\bfP}{\mathbf{P}}
\newcommand{\Csys}{C_{\mathsf{sys}}}

\newcommand{\Tburn}{T_\mathsf{pe}}
\newcommand{\Tburnss}{T^{\mathsf{ss}
}_\mathsf{pe}}
\newcommand{\snr}{\mathsf{SNR}}
\newcommand{\Sigmaw}{\Sigma_W}
\newcommand{\Sigmav}{\Sigma_V}
\newcommand{\Sigmae}{\Sigma_E}

\newtheorem{definition}{Definition}[section] 
\newtheorem{theorem}{Theorem}[section] 
\newtheorem{proposition}{Proposition}[section] 
\newtheorem{corollary}{Corollary}[section] 
\newtheorem{assumption}{Assumption}[section] 
\newtheorem{problem}{Problem}[section] 
\newtheorem{remark}{Remark}[section]
\newtheorem{lemma}{Lemma}[section]
\newtheorem{example}{Example}[section]


%% file: abstract.tex
\begin{abstract}

This tutorial serves as an introduction to recently developed non-asymptotic methods in the theory of---mainly linear---system identification. We emphasize tools we deem particularly useful for a range of problems in this domain, such as the covering technique, the Hanson-Wright Inequality and the method of self-normalized martingales. We then employ these tools to give streamlined proofs of the performance of various least-squares based estimators for identifying the parameters in autoregressive models. We conclude by sketching out how the ideas presented herein can be extended to certain nonlinear identification problems. 
\end{abstract}

%% file: sections/notation.tex
\newpage

\section*{Notation}

Maxima (resp.\ minima) of two numbers $a,b\in \R$ are denoted by $a\vee b =\max(a,b)$ ($a\wedge b = \min(a,b)$). For two sequences $\{a_t\}_{t\in \Z}$ and $\{b_t\}_{t\in \Z}$ we introduce the shorthand $a_t \lesssim b_t$ if there exists a universal constant $C>0$ and an integer $t_0$ such that $a_t \leq C b_t$ for every $t \geq t_0$. If $a_t \lesssim b_t$ and $b_t \lesssim a_t$ we write $a_t \asymp b_t$.  Let $\mathsf{X} \subset \R^d$ and let $f,g \in \mathsf{X} \to R$. We write $f=O(g)$ if $\limsup_{x\to x_0} |f(x)/g(x)|<\infty$, where the limit point $x_0$ is typically understood from the context. We use $\tilde O$ to hide logarithmic factors and write  $f=o(g)$ if $\limsup_{x\to x_0} |f(x)/g(x)|=0$. We write  $f=\Omega(g)$ if $\limsup_{x\to x_0} |f(x)/g(x)|>0$. For an integer $N$, we also define the shorthand $[N] \triangleq \{1,\dots,N\}$.  

Expectation (resp.\ probability) with respect to all the randomness of the underlying
probability space is denoted by $\E$ (resp.\ $\Pr$).

The Euclidean norm on $\mathbb{R}^{d}$ is denoted $\|\cdot\|_2$,
and the unit sphere in $\R^d$ is denoted $\mathbb{S}^{d-1}$. The standard inner product on $\R^{d}$ is denoted $\langle\cdot,\cdot\rangle$. We embed matrices $M \in \R^{d_1\times d_2}$ in Euclidean space by vectorization: $\VEC M \in \mathbb{R}^{d_1 d_2}$, where $\VEC$ is the operator that vertically stacks the columns of $M$ (from left to right and from top to bottom). For a matrix $M$ the Euclidean norm is the Frobenius norm, i.e., $\|M\|_F\triangleq \|\VEC M\|_2$. We similarly define the inner product of two matrices $M,N$ by $\langle M, N \rangle \triangleq \langle \VEC M, \VEC N\rangle $. The transpose of a matrix $M$ is denoted by $M^\T$ and $\tr M $ denotes its trace. For a matrix $M \in \R^{d_1 \times d_2}$, we order its singular values $\sigma_{1}(M),\dots,\sigma_{d_1 \wedge d_2}(M)$ in descending order by magnitude. We also write $\opnorm{M}$ for its largest singular value: $\opnorm{M} \triangleq \sigma_1(M)$. To not carry dimensional notation, we will also use $\sigma_{\min}(M)$ for the smallest nonzero singular value.  For square matrices $M\in \R^{d\times d}$ with real eigenvalues, we similarly order the eigenvalues of $M$ in descending order as $\lambda_{1}(M),\dots,\lambda_{d}(M)$. In this case, $\lambda_{\min}(M)$ will also be used
to denote the minimum (possibly zero) eigenvalue of $M$. For two symmetric matrices $M, N$, we write $M \succ N$ ($M\succeq N)$ if $M-N$ is positive (semi-)definite. 

\newpage

%% file: sections/introduction.tex
\section{Introduction}

Machine learning methods are at an ever increasing pace being integrated into domains that have classically been within the purview of controls. There is a wide range of examples, including perception-based control, agile robotics, and autonomous driving and racing. As exciting as these developments may be, they have been most pronounced on the experimental and empirical sides. To deploy these systems safely, stably, and robustly into the real world, we argue that a principled and integrated theoretical understanding of a) fundamental limitations and b) statistical optimality is needed. Under the past few years, a host of new techniques have been introduced to our field. Unfortunately, existing results in this area are relatively inaccessible to a typical first or second year graduate student in control theory, as they require both sophisticated mathematical tools not typically included in a control theorist's training (e.g., high-dimensional statistics and learning theory). 

This tutorial seeks to provide a streamlined exposition of some of these recent advances that are most relevant to the non-asymptotic theory of linear system identification. Our aim is not to be encyclopedic but rather to give simple proofs of the main developments and to highlight and collect the key technical tools to arrive at these results. For a broader---and less technical---overview of the literature we point the reader to our recent survey \citep{tsiamis2022statistical}. It is also worth to point out that the classical literature on system identification has done a formidable job at---often very accurately---characterizing the asymptotic performance of identification algorithms \citep{lennart1999system}. Our aim is not to supplant this literature but rather to complement the asymptotic picture with finite sample guarantees by relaying recently developed technical tools drawn from high-dimensional probability, statistics and learning theory \citep{vershynin_2018,wainwright2019high}. 

\subsection{Problem Formulation}\label{subsec:probfo}

Let us now fix ideas. We are concerned with  linear time-series models of the form:
\begin{equation}\label{eq:regressionmodel}
\begin{aligned}
    Y_t &= \theta^\star X_t + V_t&& t=1,2,\dots, T
\end{aligned}
\end{equation}
where $Y_{1:T}$ is a sequence of outputs (or targets) assuming values in $\R^{\dy}$ and $X_{1:T}$ is a sequence of inputs (or covariates) assuming values in $\R^{\dx}$. The goal of the user (or learner) is to recover the a priori unknown linear map $\theta^\star \in \R^{\dy \times \dx}$ using only the observations $X_{1:T}$ and $Y_{1:T}$. The linear relationship in the regression model \eqref{eq:regressionmodel} is perturbed by a stochastic noise sequence $V_{1:T}$ assuming values in $\R^{\dy}$. We refer to the regression model \eqref{eq:regressionmodel} as a time-series to emphasize the fact that the observations $X_{1:T}$ and $Y_{1:T}$ may arrive sequentially and in  particular that past $X_t$ and $Y_t$ may influence future  $X_{t'}$ and $Y_{t'}$ (i.e. with $t'>t)$.

\paragraph{Example: Autoregressive  Models.}
For instance, a model class of particular interest to us which is subsumed by \eqref{eq:regressionmodel} are the (vector) autoregressive exogenous models of order $p$ and $q$ (briefly ARX$(p,q)$):
\begin{equation}\label{eq:arx}
    \begin{aligned}
        Y_{t} &= \sum_{i=1}^{p} A^\star_i Y_{t-i} +\sum_{j=1}^{q} B^\star_i U_{t-j}+W_t
    \end{aligned}
\end{equation}
where typically $U_{1:T-1}$ is a sequence of user specified inputs taking values in $\R^{\du}$ and $W_{1:T}$ is an \iid\ sequence of noise variables taking values in $\R^{\dw}$. If we are only interested in the parameters $\begin{bmatrix}
        A_{1:p}^\star & B_{1:q}^\star 
        \end{bmatrix}$, we obtain the model \eqref{eq:arx} by setting 
\begin{equation}
    \begin{aligned}
        X_t &= \begin{bmatrix}
        Y_{t-1:t-p}^\T & U_{t-1:t-q}^\T 
        \end{bmatrix}^\T; && \theta^\star =\begin{bmatrix}
        A_{1:p}^\star & B_{1:q}^\star 
        \end{bmatrix}; &&V_t = W_t.
    \end{aligned}
\end{equation}
We  point out that that the above discussion presupposes that the order of the model, $ (p,q)$, is known (there are ways around this).

In this tutorial we will provide the necessary tools to tackle the following problem.

\begin{problem}\label{prob:learningprob}
    Fix $\e>0$, $\delta \in (0,1)$, and a norm $\|\cdot\|$. Fix also a `reasonable' estimator $\widehat \theta$ of $\theta_\star$ using a sample $(X,Y)_{1:T}$ from \eqref{eq:regressionmodel}. We seek to establish finite sample guarantees of the form
\begin{equation}\label{eq:highprobaguarantee}
    \| \widehat \theta -\theta^\star \| \leq \e \qquad \textnormal{ with probability at least }1-\delta 
\end{equation}
    where $\e $ controls the accuracy (or rate) and the failure parameter $\delta$ controls the confidence.
\end{problem}

In the sequel, `reasonable' estimator will typically mean some form of least squares estimator \eqref{eq:LSEdef}. These are introduced in \Cref{subsec:LSE} below. A bound of the form \eqref{eq:highprobaguarantee} is typically thought of as follows. We fix a priori the failure parameter $\delta$ and then provide guarantees of the form $ \| \widehat \theta -\theta^\star \| \leq \e(T,\delta, \sfP_{XY})$ where $\sfP_{XY}$ is the joint distribution of $(X,Y)_{1:T}$. Hence, the sample size $T$, the failure probability $\delta$ and the distribution of the samples all impact the performance guarantee $\e$ we are able to establish. To be more specific, $\e$ will typically be of the form 
\begin{equation}\label{eq:informalsnr}
    \e \propto (\textnormal{Noise Scale}) \times \sqrt{\frac{\textnormal{problem dimension} + \log(1/\delta)}{\textnormal{sample size}}}.
\end{equation}
Thus in principle, the best possible choice of $\e^2$ can be thought of as a high probability version of the (inverse) signal-to-noise ratio of the problem at hand. The fact that the confidence parameter $\delta$ typically  affects \eqref{eq:informalsnr} additively in $\log (1/\delta)$ is consistent with classical asymptotic normality theory of estimators. One often expects the normalized difference $T^{-1/2}(\widehat \theta - \theta^\star)$ to converge in law to a normal distribution \citep{van2000asymptotic}. In this tutorial we will provide tools that allow us to match such classical asymptotics but with a finite sample twist. Let us also remark that there often is a minimal requirement on the sample size necessary for a bound of the form \eqref{eq:highprobaguarantee}-\eqref{eq:informalsnr} to hold. Such requirements are typically of the form 
\begin{equation}\label{eq:burnininformal}
\textnormal{sample size} \: \gtrsim \: \textnormal{problem dimension} + \log(1/\delta).
\end{equation}
Requirements such as \eqref{eq:burnininformal} are called burn-in times and are related to the notion of persistence of excitation. They correspond to the rather minimal requirement that the parameter identification problem is feasible in the complete absence of observation noise. 


\subsection{Least Squares Regression and the Path Ahead}\label{subsec:LSE}

Let us now return to the general setting of \eqref{eq:regressionmodel}. Fix a subset $\mathsf{M}$ of $ \R^{\dy \times \dx}$, called the model class. The estimator
\begin{equation}\label{eq:LSEdef}
    \widehat \theta  \in \argmin_{\theta \in \mathsf{M}} \frac{1}{T}\sum_{t=1}^{T} \| Y_t -\theta X_t\|_2^2  
\end{equation}
is the  least squares estimator (LSE) of $\theta^\star$ (with respect to $\mathsf{M}$). Often we simply set $\mathsf{M}= \R^{\dy \times \dx}$. In this case, equivalently:
\begin{equation}\label{eq:OLSdef}
    \widehat \theta  = \left(\sum_{t=1}^{T} Y_t X_t^\T \right)\left(\sum_{t=1}^{T} X_t X_t^\T \right)^\dagger
\end{equation}
and the LSE reduces to the (minimum norm) ordinary least squares (OLS) estimator \eqref{eq:OLSdef}.

For simplicity, let us further assume that the (normalized) empirical covariance matrix:
\begin{equation}\label{eq:empcov}
\widehat \Sigma \triangleq \frac{1}{T}\sum_{t=1}^{T} X_t X_t^\T    ;
\end{equation}
is full rank almost surely. 

\paragraph{The Path Ahead.}
Let us now briefly sketch the path ahead to solve \Cref{prob:learningprob}. If \eqref{eq:empcov} is full rank---as required above---the estimator \eqref{eq:OLSdef}  admits the convenient error representation:
\begin{equation}\label{eq:OLSerror}
    \widehat \theta  -\theta^\star = \left[\left(\sum_{t=1}^{T} V_t X_t^\T \right)\left(\sum_{t=1}^{T} X_t X_t^\T \right)^{-1/2 }\right]\left(\sum_{t=1}^{T} X_t X_t^\T \right)^{-1/2 }.
\end{equation}
The leftmost term of \eqref{eq:OLSerror} (in square brackets) can be shown to be (almost) time-scale invariant in many situations. For instance, if the noise $V_{1:T}$ is a sub-Gaussian martingale difference sequence with respect to the filtration generated by the covariates $X_{1:T}$, one  can invoke methods from the theory of self-normalized processes to show this \citep{pena2009self, abbasi2013online}. These methods are the topic of \Cref{sec:selfnorm}.

Whenever this is the case, the dominant term in the rate of convergence of the least squares estimator is $ \left(\sum_{t=1}^{T} X_t X_t^\T \right)^{-1/2 }$. In other words, providing control of the smallest eigenvalue of \eqref{eq:empcov} effectively yields control of the rate of convergence of the least squares estimator in many situations. Thus, to analyze the rate of convergence of \eqref{eq:LSEdef} when $\mathsf{M}= \R^{\dy \times \dx}$ it suffices to:
\begin{itemize}
    \item Analyze the smallest eigenvalue (or lower tail) of \eqref{eq:empcov}. We provide such analyses in \Cref{sec:onesided}
    \item Analyze the scale invariant term (in square brackets) of \eqref{eq:OLSerror}. This can in many situations be handled for instance by the self-normalized martingale method described in \Cref{sec:selfnorm}.
\end{itemize}

\subsection{Overview}

Before covering these more technical topics in \Cref{sec:onesided} and \Cref{sec:selfnorm}, we also briefly review some preliminaries from probability theory in \Cref{sec:prels}. We then demonstrate how to apply these ideas in the setting of identifying the parameters of an ARX$(p,q)$ model of the form \eqref{eq:arx} in \Cref{sec:sysid}. An alternative perspective not based on the decomposition \eqref{eq:OLSerror} for more general least squares algorithms is given in \Cref{sec:basicineq}. We conclude with a brief discussion on how the tools in \Cref{sec:basicineq} can be extended to study more general nonlinear phenomena in \Cref{sec:nonlinear}.

%% file: sections/concentrationineqs.tex
\section{Preliminaries: Concentration Inequalities, Packing and Covering}\label{sec:prels}
Before we proceed to tackle the more advanced question of analyzing the LSE \eqref{eq:LSEdef}, let us discuss a few preliminary inequalities that control the tail of a random variable. Our first inequality is Markov's. 

\begin{lemma}
Let $X$ be a nonnegative random variable. For every $s > 0 $ we have that
\begin{equation}\label{eq:markovs}
    \mathbf{P} (X \geq  s) \leq s^{-1}\E[X].
\end{equation}
\end{lemma}
\begin{proof}
We have that $\E [X] \geq \E [\mathbf{1}_{X\geq s} X] \geq s
 \E [\mathbf{1}_{X\geq s } ]$. Since $\E [\mathbf{1}_{X \geq s } ]= \mathbf{P} (X \geq s) $ the result follows by rearranging.
\end{proof}

Typically, Markov's inequality itself is insufficient for our goals: we seek deviation inequalities that taper of exponentially fast in $s$ and not as $s^{-1}$. Such scaling is for instance predicted asymptotically by the central limit theorem by the asymptotic normality of renormalized sums of square integrable \iid\ random variables; that is, sums of the form $S_n/\sqrt{n} =( X_1+X_2+\dots+X_n)/\sqrt{n}$ where the $X_i, i\in [n]$ are independent and square integrable. For random variables possessing a moment generating function, Markov's inequality can be "boosted" by the so-called "Chernoff trick". Namely, we apply Markovs inequality to the moment generating function of the random variable instead of applying it directly to the random variable itself.

\begin{corollary}[Chernoff]
Fix $s>0$ and suppose that $\E \exp\left(\lambda X \right)$ exists. Then
\begin{equation}\label{eq:chernofftrick}
    \mathbf{P}\left(X \geq s \right)\leq \min_{\lambda \geq 0} e^{-\lambda s }\E \exp \left(\lambda  X \right).
\end{equation}
\end{corollary}

\begin{proof}
Fix $\lambda \geq 0$. We have:
    \begin{equation*}
    \begin{aligned}
    \mathbf{P}\left(X  \geq s \right)&= \mathbf{P}\left(\exp\left(\lambda X\right)  \geq \exp \left(\lambda s\right) \right) && (\textnormal{monotonicity of } x\mapsto e^{\lambda x})\\
    &\leq e^{-\lambda s }\E \exp \left(\lambda X \right) && (\textnormal{Markov's inequality}).
    \end{aligned}
\end{equation*}
The result follows by optimizing.
\end{proof}
Recall that the function $\psi_X( \lambda)\triangleq \E \exp \left(\lambda  X \right)$ is the moment generating function of $X$. For instance, if $X$ has univariate Gaussian distribution with mean zero and variance $\sigma^2$, the moment generating function appearing in \eqref{eq:chernofftrick} is just $\E \exp \left(\lambda X \right)=\exp\left(\lambda^2\sigma^2/2\right)$. Hence the probability that said Gaussian exceeds $s$ is upper-bounded:
\begin{equation}\label{eq:gaussiandevianineq}
     \mathbf{P}\left(X>s \right)\leq \min_{\lambda \geq 0} e^{-\lambda s }\exp \left(\lambda^2\sigma^2/2\right) =\exp\left(\frac{-s^2}{2\sigma^2} \right)
\end{equation}
which (almost) exhibits the correct Gaussian tails as compared to \eqref{eq:markovs}.\footnote{We write almost because $\exp(-s^2/2\sigma^2) \approx \Pr(V > s)$ where $V\sim N(0,\sigma^2)$ but the expression is not exact.}  It should be pointed out that assumptions stronger than those of the Central Limit Theorem (finite variance) are indeed needed for a non-asymptotic theory with sub-Gaussian tails as in \eqref{eq:gaussiandevianineq}. An assumption of this kind which is relatively standard in the literature is introduced next.

\subsection{Sub-Gaussian Concentration and the Hanson-Wright Inequality}
In the sequel, we will not want to impose the Gaussian assumption. Instead, we define a class of random variables that admit reasoning analogous to \eqref{eq:gaussiandevianineq}.

\begin{definition}\label{def:sub_G}
We say that a random vector $W$ taking values in $\R^d$ is $\sigma^2$-sub-Gaussian ($\sigma^2$-subG) if for every $v \in \R^d$ we have that:
\begin{equation}\label{eq:subgdef}
    \E \exp \left( \langle v, W\rangle \right) \leq \exp\left( \frac{\sigma^2\|v\|^2}{2} \right).
\end{equation}
Similarly, we say that $W$ is $\sigma^2$-conditionally sub-Gaussian with respect to a $\sigma$-field $\mathcal{F}$ if \eqref{eq:subgdef} holds with $\E [\cdot]$ replaced by $\E[\cdot|\mathcal{F}]$. 
\end{definition}
The term $\sigma^2$ appearing in \eqref{eq:subgdef} is called the variance proxy of a sub-Gaussian random variable. The significance of this definition is that the one-dimensional projections $X =\langle v, W\rangle$ (with $\|v\|=1$) satisfy the tail inequality \eqref{eq:gaussiandevianineq}. While obviously Gaussian random variables are sub-Gaussian with their variance as variance-proxy, there are many examples beyond Gaussians that fit into this framework. It is for instance straightforward to show that bounded random variables have variance proxy proportional to the square of their width \citep[see eg.][Examples 2.3 and 2.4]{wainwright2019high}. Moreover, it is readily verified that the normalized sum mentioned above---$S_n / \sqrt{n}=(X_1+\dots+X_n)/\sqrt{n}$---satisfies the same bound \eqref{eq:gaussiandevianineq} provided that the entries of $X_{1:n}$ are independent, mean zero and $\sigma^2$-sub-Gaussian. To see this, notice that the moment generating function "tensorizes" across products. Namely, for every $\lambda \in \R$:
\begin{equation}
    \E \exp \left( \frac{\lambda}{\sqrt{n}} \sum_{i=1}^n X_i \right)=\prod_{i=1}^n \E \exp \left( \frac{\lambda}{\sqrt{n}}  X_i \right) \leq \prod_{i=1}^n \exp \left( \frac{\lambda^2 \sigma^2}{2n} \right)= \exp\left( \frac{\lambda^2\sigma^2}{2}\right).
\end{equation}
Hence, by the exact same reasoning leading up to \eqref{eq:gaussiandevianineq} such normalized sub-Gaussian sums satisfy the same tail bound \eqref{eq:gaussiandevianineq}.

When analyzing linear regression models, most variables of interest are typically either linear or quadratic in the variables of interest (cf. \eqref{eq:OLSerror}). Hence, we also need to understand how squares of sub-Gaussian random variables behave. The next result shows that sub-Gaussian quadratic forms exhibit similar tail  behavior to the Chi-squared distribution (often in the literature referred to as sub-exponential tails). It is known as the Hanson-Wright Inequality.

\begin{theorem}[\cite{hanson1971bound,rudelson2013hanson}]\label{thm:HWineqhighproba}
Let $M\in \R^{d\times d}$. Fix a random variable $W=W_{1:d}$ where each $W_i, i\in [d]$ is a scalar, mean zero and independent $\sigma^2$-sub-Gaussian random variable.  Then for every $s \in [0,\infty)$:
\begin{equation}
    \Pr\left(|W^\T M W -\E W^\T M W | > s \right)\leq  2  \exp \left( - \min \left( \frac{s^2}{144 \sigma^4 \| M\|_F^2} ,\frac{s}{16\sqrt{2}\sigma^2 \opnorm{M} } \right)\right).
\end{equation}
\end{theorem}

The proof of \Cref{thm:HWineqhighproba} is rather long and technical and thus relegated to \Cref{sec:proofofhw}. There, the reader may also find further useful concentration inequalities for quadratic forms in sub-Gaussian variables. In fact, there are plethora of useful concentration inequalities not covered here and the interested reader is urged to consult the first few chapters of \cite{vershynin_2018}.

\subsection{Covering and Discretization Arguments}
\label{subsec:eps}

We will often find ourselves in a situation where it is possible to obtain a scalar concentration bound but need this to hold uniformly for many random variables at once. The $\e$-net argument, which proceeds via the notion of covering numbers, is a relatively straightforward way of converting concentration inequalities for scalars into their counterparts for vectors, matrices and functions more generally.

The reader will for instance notice that the quantity being controlled by \Cref{thm:HWineqhighproba} is a scalar quadratic form in sub-Gaussian random variables. By contrast, the empirical covariance matrix \eqref{eq:empcov} is a matrix and so a conversion step is needed. This idea will be used frequently and in various forms throughout the manuscript, so we review it briefly here for the particular case of controlling the operator norm of a random matrix. To this end, we notice that for any matrix $M\in \R^{m\times d} $:
\begin{equation}\label{eq:variationaloperatornorm}
    \opnorm{M}^2 = \max_{v \in \mathbb{S}^{d-1}}  \langle Mv , Mv \rangle .
\end{equation}
Hence, the operator norm of a random matrix is a maximum of scalar random variables indexed by the unit sphere $\mathbb{S}^{d-1}$. 

Recall now that the union bound states that the probability that the maximum of a \emph{finite collection} ($|S|<\infty$)  $\{X_i\}_{i \in S}$ of random variables exceeds a certain threshhold can be bounded by the sum of their probabilities:
\begin{equation}\label{eq:theunionbound}
    \Pr \left(  \max_{i\in S} X_i > t  \right) \leq \sum_{i \in S} \Pr \left(   X_i > t  \right).
\end{equation}
Unfortunately, the unit sphere appearing \eqref{eq:variationaloperatornorm} is not a finite set and so the union bound \eqref{eq:theunionbound} cannot be directly applied. However, when the domain of optimization has geometric structure, one can often exploit this to leverage the union bound not directly but rather in combination with a discretization argument. Returning to our example of the operator norm of a matrix, the set $S$ appearing in \eqref{eq:theunionbound} will be a discretized version of the unit sphere $\mathbb{S}^{d-1}$.

The following notion is key.
\begin{definition}
    Let $(\sfX,d)$ be a compact metric space and fix $\e>0$. A subset $\mathcal{N}$ of $\sfX$ is called an $\e$-net of $\sfX$ if every point of $\sfX$ is within radius $\e$ of a point of $\mathcal{N}$:
    \begin{equation}\label{eq:coveringdefined}
        \sup_{x\in \sfX} \inf_{x'\in \mathcal{N}}d(x,x') \leq \e.
    \end{equation}
 Moreover, the minimal cardinality of $\mathcal{N}$ necessary such that \eqref{eq:coveringdefined} holds is called the covering number at resolution $\e$ of $(\sfX,d)$ and is denoted $\mathcal{N}(\e,\sfX,d)$.
\end{definition}

We will not explore this notion in full, but simply content ourselves to note that it plays very nicely with the notion of operator norm.

\begin{lemma}[Lemma 4.4.1 and Exercise 4.4.3 in \cite{vershynin_2018}]\label{lem:opnormdisc}
Let $M \in \R^{m \times d}$ and let $\e \in (0,1)$.  Then for any $\e $-net $\mathcal{N}$ of $(\mathbb{S}^{d-1},\|\cdot\|_2)$ we have that:
\begin{equation}\label{eq:opnormdisc}
    \opnorm{M} \leq \frac{1}{1-\e} \sup_{v \in \mathcal{N}} \| M v\|_2.
\end{equation}
If additionally $M$ is  symmetric, we also have that:
\begin{equation}\label{eq:opnormdisc2}
    \opnorm{M} \leq \frac{1}{1-2\e} \sup_{v \in \mathcal{N}} | v^\T M v|.
\end{equation}
\end{lemma}
Hence at a small multiplicative cost, the computation of the operator norm can be restricted to the discretized sphere $\mathcal{N}$. Our intention is now to apply the union bound \eqref{eq:theunionbound} to the right hand side of \eqref{eq:opnormdisc}. To do so, we also need control of the size (cardinality) of the $\e$-net.

\begin{lemma}[Corollary 4.2.13 in \cite{vershynin_2018}]\label{lem:volumetric}
For any $\e>0$ the covering numbers of $\mathbb{S}^{d-1}$ satisfy
\begin{equation}
    \mathcal{N}(\e,\mathbb{S}^{d-1},\|\cdot\|) \leq \left(1+ \frac{2}{\e}\right)^d.
\end{equation}
\end{lemma}

We now provide two instances of this covering argument combined with the union bound. The second of these uses an alternative variational characterization of the operator norm but otherwise similar ideas.

\begin{lemma}\label{lem:net argument 1}
  Let $M$ be an $m \times d$ random matrix, and $\epsilon \in (0, 1)$. Furthermore, let $\mathcal{N}$ be an $\epsilon$-net of $\mathbb{S}^{d-1}$ of minimal cardinality. Then for all $\rho > 0$, we have
  \begin{align*}
      \Pr \left( \opnorm{M} > \rho \right) & \leq \left( \frac{2}{\epsilon} + 1\right)^d \max_{v \in \mathcal{N}} \Pr\left( \| Mv\|_2 > (1-\epsilon)\rho \right).
  \end{align*}
\end{lemma}
\medskip
\begin{lemma}\label{lem:net argument 2}
  Let $M$ be an $d \times d$ symmetric random matrix, and let $\epsilon \in (0, 1/2)$. Furthermore, let $\mathcal{N}$ be an $\epsilon$-net of $\mathbb{S}^{d-1}$ with minimal cardinality. Then for all $\rho > 0$, we have
  \begin{align*}
     \Pr\left( \opnorm{M} > \rho \right) & \leq \left( \frac{2}{\epsilon} + 1\right)^d \max_{v \in \mathcal{N}} \Pr\left( \vert v^\top  M v \vert > (1-2 \epsilon)\rho \right).
  \end{align*}
\end{lemma}
\medskip
\noindent
Lemma \ref{lem:net argument 1} and Lemma \ref{lem:net argument 2} exploit two different variational forms of the operator norm. Namely for any $M$ we have that $
\opnorm{M}^2= \sup_{v \in \mathbf{S}^{d-1}} \| Mv\|^2
$ and in addition, when $M$ is symmetric we also have, $ \opnorm{M} = \sup_{v \in S^{d-1}} | v^\top Mv |$. The proof of these last two lemmas are standard and can be found for example in \cite[Chapter 4]{vershynin_2018}.

\subsection{Concentration of the Covariance Matrix of Linear Systems}
\label{subsec:conclinsys}
To not get lost in the weeds, let us provide an example showcasing the use of \Cref{thm:HWineqhighproba} due to \cite{jedra2022finite}. Recall that the matrix $\widehat \Sigma$ appearing in \eqref{eq:empcov} is crucial to the performance of the least squares estimator. We will now see that this matrix is well-conditioned when we consider stable first order auto-regressions of the form:
\begin{equation}\label{eq:LDS}
    X_{t+1} = A^\star X_t + W_t \qquad t=1,\dots,T \qquad W_{1:T}\textnormal{ \iid\ isotropic and $K^2$-subG}
\end{equation}
taking values in $\R^{\dx}$. By stable we mean that the largest eigenvalue of $A^\star$ has module strictly smaller than 1.

The following result is a consequence of the Hanson-Wright inequality together with the discretization strategy outlined in \Cref{subsec:eps}. The full proof is given in \Cref{sec:twosided}.

\begin{theorem}\label{thm:spectrum deviations}
Let $\varepsilon > 0$ and set  $M\triangleq \left( \sum_{t=1}^{T} \sum_{k=0}^{t-1} (A^\star)^k( A^{\star,\T})^k \right)^{-\frac{1}{2}}$. Let also $\mathbf{L}$ be the linear operator such that $X_{1:T}=\mathbf{L} W_{1:T}$. Then simultaneously for every $i \in [\dx]$:
\begin{equation*}
(1 - \varepsilon)^2 \lambda_{\min} \left( \sum_{t=1}^{T} \sum_{k=0}^{t-1} (A^\star)^k( A^{\star,\T})^k\right) \leq \lambda_{i} \left(\sum_{t=1}^T X_tX_t^\T \right) \leq  (1 + \varepsilon)^2  \lambda_{\max} \left(\sum_{t=1}^{T} \sum_{k=0}^{t-1} (A^\star)^k( A^{\star,\T})^k \right)
\end{equation*}
holds with probability at least
\begin{equation}\label{eq:yassirsburnin} 
1 - \exp \left(-  \frac{\varepsilon^2}{ 576 \, K^2 \opnorm{M}^2 \opnorm{\bfL}^2 }  + \dx \log(18)\right).
\end{equation}

\end{theorem}
Put differently, on the same event as in \Cref{thm:spectrum deviations}, the spectrum of 
\begin{equation}
    \widehat \Sigma  =\frac{1}{T} \sum_{t=1}^T X_t X_t^\T  
\end{equation}
is sandwiched by that of its population counterpart ($\E \widehat \Sigma$) within a small multiplicative factor. The result holds with high probability for strictly stable systems. 

The quantity $\opnorm{L}$ in \eqref{eq:yassirsburnin} grows very quickly as the spectral radius of $A^\star$ tends to $1$; \Cref{thm:spectrum deviations} becomes vacuous in the marginally stable regime. It turns out that requirement of two-sided concentration---the sandwiching of the entire spectrum---is too stringent a requirement to obtain bounds that degrade gracefully with the stability of the system. Fortunately, we only need sharp control of the lower half of the spectrum to control the error \eqref{eq:OLSerror}. This motivates \Cref{sec:onesided} below, in which we will see how to relax the stability assumption and analyze more general linear systems.

 \subsection{Notes}
The basic program carried out in \Cref{subsec:conclinsys} can be summarized as follows: (1) introduce a discretization of the problem considered---for matrices this is typically a discretization of the unit sphere; (2) prove an exponential inequality for a family of scalar random variables corresponding to one-dimensional projection of the discretization---in our case: prove bounds on the moment generating function of quadratic forms in random matrices; and (3) conclude to obtain a uniform bound by using the union bound across the discretization. This roughly summarizes the proof of \Cref{thm:spectrum deviations}. These tools are thematic throughout this manuscript.

%% file: sections/lowertail.tex
\section{The Lower Spectrum of the Empirical Covariance}
\label{sec:onesided}

Recall that our outline of the analysis of the least squares estimator in \Cref{subsec:LSE} consists of two main components, one of which being the lower tail of the empirical covariance matrix \eqref{eq:empcov}. In this section we provide a self-contained analysis of this random matrix for a class of "causal" systems. Moreover, we will emphasize only the lower tail of this random matrix as to sidestep issues with bounds degrading with the stability of the system considered. This allows us to quantitatively separate the notions of persistence of excitation and stability.

Let us now carry out this program. Fix two integers $T$ and $k$ such that $T/k\in \N$. We consider causal processes of the form $X_{1:T}=(X_1^\T,\dots,X_{T}^\T)^\T$ evolving on $\R^d$. More precisely, we assume the existence of an isotropic sub-Gaussian process evolving on $\R^p$, $W_{1:T}$ with $\E W_{1:T}W_{1:T}^\T = I_{pT}$ and a (block-) lower-triangular matrix $\mathbf{L} \in \R^{dT\times pT}$ such that
\begin{equation}\label{eq:linearcausalprocess}
    X_{1:T}=\mathbf{L}W_{1:T}.
\end{equation}
 We will assume that all the $pT$-many entries of $W_{1:T}$ are independent $K^2$-sub-Gaussian for some positive $K\in \R$.

We say that $X_{1:T}$ is $k$-causal if the matrix $\mathbf{L}$ has the block lower-triangular form:
\begin{align}\label{eq:Lopdefined}
    \mathbf{L} 
    = 
    \begin{bmatrix}
    \mathbf{L}_{1,1} &0 &0&0&0\\
    \mathbf{L}_{2,1} & \mathbf{L}_{2,2} & 0  &0 &0\\
    \mathbf{L}_{3,1} & \mathbf{L}_{3,2} & \mathbf{L}_{3,3}  &0 &0\\
    \vdots & \ddots & \ddots & \ddots &\vdots\\
    \mathbf{L}_{T/k,1} &\dots & \dots & \dots&\dots \mathbf{L}_{T/k,T/k}
    \end{bmatrix}
    =
    \begin{bmatrix}
    \mathbf{L}_{1}\\
    \mathbf{L}_{2}\\
    \mathbf{L}_{3}\\
    \vdots \\
    \mathbf{L}_{T/k}
    \end{bmatrix}
\end{align}
where each $\mathbf{L}_{ij} \in \R^{dk\times pk}, i,j \in [T/k] \triangleq \{1,2,\dots,T/k\}$. In brief, we say that $X_{1:T}$ satisfying the above construction is $k$-causal with independent $K^2$-sub-Gaussian increments.

Obviously, every $1$-causal process is $k$-causal for every $k\in \N$ as long as the divisibility condition holds. To analyze the lower tail of the empirical covariance of $X_{1:T}$ we will  also  associate a decoupled random process $$\tilde X_{1:T} = \mathrm{blkdiag}(\mathbf{L}_{11},\dots, \mathbf{L}_{T/k,T/k})W_{1:T}.$$
Hence, the process $\tilde X_{1:T}$ is generated in much the same way as $X_{1:T}$ but by removing the sub-diagonal entries of $\mathbf{L}$:
\begin{equation*}
    \tilde {\mathbf{L}}\triangleq    \begin{bmatrix}
    \mathbf{L}_{1,1} &0 &0&0\\
    0 & \mathbf{L}_{2,2} & \ddots  &\vdots \\
    \vdots & \ddots & \ddots &0\\
    0 &  \dots & 0 &  \mathbf{L}_{T/k,T/k}
    \end{bmatrix} \quad \Longrightarrow \quad \tilde X_{1:T}= \tilde {\mathbf{L}}W_{1:T}.
\end{equation*}
 We emphasize that by our assumptions on $W_{1:T}$ and the block-diagonal structure of $\tilde {\mathbf{L}}$ the variables $\tilde X_{1:k},\tilde X_{k+1:2k},\dots,\tilde X_{T-k+1:T}$ are all independent of each other; they have been decoupled. 
This decoupled process will effectively dictate our lower bound, and we will show under relatively mild assumptions that
\begin{align}\label{eq:lambdalowerbound}
    \lambda_{\min}\left (\frac{1}{T}\sum_{t=1}^{T} X_t X_t^\T \right)\gtrsim \lambda_{\min} \left(\frac{1}{T}\sum_{t=1}^{T}\E \tilde X_t \tilde X_t^\T\right)
\end{align}
with probability that approaches $1$ at an exponential rate in the sample size $T$. More precisely, the following statement is the main result of this section.

\begin{theorem}\label{thm:anticonc}
Fix an integer $k \in \N$, let $T \in N$ be divisible by $k$ and suppose $X_{1:T}$ is a $k$-causal process taking values in $\R^d$ with $K^2$-sub-Gaussian increments. Suppose further that the diagonal blocks are all equal: $\mathbf{L}_{j,j}=\mathbf{L}_{1,1}$ for all $j\in [T/k]$. Suppose $ \lambda_{\min} \left(\sum_{t=1}^{T} \E \tilde X_t \tilde X_t^\T \right)>0$. We have that:
 \begin{equation}
    \mathbf{P} \left(  \frac{1}{T} \sum_{t=1}^{T} X_t X_t^\T \nsucceq \frac{1}{8T} \sum_{t=1}^{T}   \E \tilde X_t \tilde X_t^\T \right)
   \leq  \left( C_{\mathsf{sys}} \right)^d 
 \exp \left(- \frac{T}{576K^2k} \right)
 \end{equation}
 where
 \begin{equation}\label{eq:csysdef}
     C_{\mathsf{sys}}  \triangleq 1+4\sqrt{2} \frac{\left(\frac{ T \opnorm{\mathbf{L} \mathbf{L} ^\T}}{18 k  \lambda_{\min}\left(\sum_{t=1}^{T} \E X_t X_t^\T \right) } + 9   \right) \lambda_{\max} \left(\sum_{t=1}^{T} \E X_t  X_t^\T\right) }{ \lambda_{\min} \left(\sum_{t=1}^{T} \E \tilde X_t \tilde X_t^\T \right)}.
 \end{equation}
\end{theorem}

To parse \Cref{thm:anticonc}, note that it simply  informs us that there exist a  a system-dependent constant $C_{\mathsf{sys}}$---which itself has no more than polynomial dependence on relevant quantities---such that if
\begin{equation}\label{eq:burninloweriso}
    T /k \geq 576 K^2 (d  \log C_{\mathsf{sys}} +\log (1/\delta) )
\end{equation}
then on an event with probability mass at least  $1-\delta$:
\begin{equation*}
     \frac{1}{T} \sum_{t=1}^{T} X_t X_t^\T  \succeq \frac{1}{8T} \sum_{t=1}^{T}   \E \tilde X_t \tilde X_t^\T.
\end{equation*}

\begin{remark}
Since the blocks of $\mathbf{L}$ can be regarded to specify the noise-to-output map, the assumption that the diagonal blocks are constant is for instance satisfied by linear time-invariant (LTI) systems. The assumption can be removed at the cost of a more complicated expression.
\end{remark}

The next example serves as the archetype for the reduction from $\mathbf{L}$ to $\tilde {\mathbf{L}}$.

\begin{example}\label{ex:restartedLDS}
Suppose that \eqref{eq:linearcausalprocess} is specified via 
\begin{equation}\label{eq:LDSar1}
    X_{t}=A_\star X_{t-1}+B_\star W_t
\end{equation}
 for $t \in [T]$ and where $(A_\star, B_\star)\in \R^{\dx \times \dx + \dx \times \dw}$. We set $d=\dx$ and $p = \dw$ in the theorem above. The reduction  from $X_{1:T}=\mathbf{L}W_{1:T}$ to $\tilde X_{1:T}=\mathrm{blkdiag}(\mathbf{L}_{11},\dots, \mathbf{L}_{T/k,T/k})W_{ 1:T}$ corresponds to replacing a single trajectory from the linear system \eqref{eq:LDSar1} of length $T$ by $T/k$ trajectories of length $k$ each and sampled independently of each other. The price we pay for decoupling these systems is that our lower bound is dictated by the gramians up to range $k$:
 \begin{equation}
\frac{1}{T} \sum_{t=1}^T \E \tilde X_t \tilde X_t^\T =\frac{1}{k} \sum_{t=1}^k \E \tilde X_t \tilde X_t^\T = \frac{1}{k} \sum_{t=1}^k  \sum_{j=0}^{t-1} (A^\star)^j B^\star B^{\star,\T} (A^{\star,\T})^j 
 \end{equation}
 instead of the gramians up to range $T$:
 \begin{equation}
   \frac{1}{T} \sum_{t=1}^T \E  X_t  X_t^\T =    \frac{1}{T}\sum_{t=1}^T \sum_{j=0}^{t-1} (A^\star)^j B^\star B^{\star,\T} (A^{\star,\T})^j.
 \end{equation}
 Put differently, the reduction from $\mathbf{L}$ to $\tilde {\mathbf{L}}$ can be thought of as restarting the system every $k$ steps. 
\end{example}

Comparing with \Cref{thm:spectrum deviations}, the advantage of \Cref{thm:anticonc} is that it allows us to provide persistence-of-excitation type guarantees that do not rely strongly on the stability of the underlying system. While \Cref{thm:spectrum deviations} gives in principle stronger two-sided concentration results, it comes at the cost of the guarantees becoming vacuous as the spectral radius of $A^\star$ in \Cref{ex:restartedLDS} tends to marginal stability (tends to $1$). By contrast, \Cref{thm:anticonc} does not exhibit such a blow-up since the dependence on $C_{\mathsf{sys}}$ in \eqref{eq:burninloweriso} is logarithmic (instead of polynomial). The distinction might seem small, but it is qualitatively important as it (almost) decouples the phenomena of stability and persistence of excitation.

\subsection{A Decoupling Inequality for sub-Gaussian Quadratic Forms}

Our proof of \Cref{thm:anticonc} will make heavy use of \Cref{prop:decoup} below. This is the crucial probabilistic inequality that allows us to decouple---or restart as discussed in \Cref{ex:restartedLDS}.

\begin{proposition}\label{prop:decoup}
Fix $K\geq 1$, $x\in \R^n$ and a symmetric positive semidefinite $Q \in \R^{(n+m)\times(n+m)}$ of the form $\displaystyle Q=\begin{bmatrix}Q_{11}& Q_{12}\\ Q_{21} & Q_{22} \end{bmatrix}$ with $Q_{22}\succ 0$. Let $W$ be an $m$-dimensional mean zero, isotropic and $K^2$-sub-Gaussian random vector with independent entries. Then for every $\lambda \in \left[0, \frac{1}{8\sqrt{2}K^2 \opnorm{Q_{22}} } \right]$ it holds true that:
\begin{align}\label{eq:prop:decoup2}
    \E \exp \left( -\lambda \begin{bmatrix}x \\ W
    \end{bmatrix}^\T  \begin{bmatrix}Q_{11}& Q_{12}\\ Q_{21} & Q_{22} \end{bmatrix} \begin{bmatrix}x \\ W
    \end{bmatrix}\right) 
    \leq \exp\left( -\lambda \tr Q_{22} +36K^4\lambda^2 \tr Q_{22}^2 \right).
\end{align}

\end{proposition}

By combining \Cref{lem:subG_decoup} below with the exponential form of Hanson-Wright (\Cref{prop:HWexpineq}) we obtain the exponential inequality \eqref{eq:prop:decoup2}, which in the sequel will allow us to control the lower tail of the conditionally random quadratic form $$ \begin{bmatrix}x \\ W
    \end{bmatrix}^\T  \begin{bmatrix}Q_{11}& Q_{12}\\ Q_{21} & Q_{22} \end{bmatrix} \begin{bmatrix}x \\ W
    \end{bmatrix}.$$

We point out that \eqref{eq:prop:decoup2} is not the best possible if the entries of are $W$ independent and Gaussian as opposed to just  isotropic and sub-Gaussian. In this case, the factor 
 $36 K^4\lambda^2  (\tr Q_{22})^2$ in \eqref{eq:prop:decoup2} can be improved to $\frac{\lambda^2}{2} \tr Q_{22}^2 $ and the inequality can be shown to hold for the entire range of non-negative $\lambda$ \cite[Lemma 2.1]{ziemann2022note}. Irrespectively, we will see in the sequel that it captures the correct qualitative behavior.

\begin{lemma}[sub-Gaussian Decoupling]
\label{lem:subG_decoup}
Fix $K\geq 1$, $x\in \R^n$ and a symmetric positive semidefinite $Q \in \R^{(n+m)\times(n+m)}$ of the form $\displaystyle Q=\begin{bmatrix}Q_{11}& Q_{12}\\ Q_{21} & Q_{22} \end{bmatrix}$. Let $W$ be an $m$-dimensional mean zero  and $K^2$-sub-Gaussian random vector. Then for every $\lambda \in \left[0, \frac{1}{4K^2 \opnorm{Q_{22}} } \right]$ it holds true that:
\begin{align}\label{eq:lem:decoup}
    \E \exp \left( -\lambda \begin{bmatrix}x \\ W
    \end{bmatrix}^\T  \begin{bmatrix}Q_{11}& Q_{12}\\ Q_{21} & Q_{22} \end{bmatrix} \begin{bmatrix}x \\ W
    \end{bmatrix}\right) 
    \leq \sqrt{\E \exp \left( -2 \lambda  W^\T Q_{22} W\right)}.
\end{align}
\end{lemma}

Once equipped with \eqref{eq:lem:decoup}, \Cref{prop:decoup} follows immediately by \Cref{prop:HWexpineq}. The proof of \Cref{lem:subG_decoup} is given in \Cref{sec:anticoncproofs}.

\subsection{The Lower Tail of the Empirical Covariance of Causal sub-Gaussian Processes}

Repeated application of \Cref{prop:decoup} to the process $X_{1:T}=\mathbf{L}W_{1:T}$ in combination with the tower property of conditional expectation yields the following exponential inequality that controls the lower tail of \eqref{eq:empcov} in any fixed direction.

\begin{theorem}\label{thm:expineq}
Fix an integer $k \in \N$, let $T \in N$ be divisible by $k$ and suppose $X_{1:T}$ is a $k$-causal  process driven by independent $K^2$-sub-Gaussian increments as described in \Cref{sec:onesided}. Fix also a matrix $\Delta\in \R^{d'\times d}$. Let $Q_{\max} \triangleq \max_{j\in [T/k]} \opnorm{\mathbf{L}_{j,j}^\T \mathrm{blkdiag}(\Delta^\T \Delta) \mathbf{L}_{j,j}}$ Then for every $\lambda \in \left[0, \frac{1}{8\sqrt{2}K^2 Q_{\max}}\right] $:
\begin{multline*}
    \E \exp \left(-\lambda \sum_{t=1}^{T}\|\Delta X_t\|_2^2 \right) 
    \\\leq \exp \Bigg( -\lambda\sum_{j=1}^{T/k} \tr\left( \mathbf{L}_{j,j}^\T \mathrm{blkdiag}(\Delta^\T \Delta) \mathbf{L}_{j,j}\right)
    + 36K^4\lambda^2 \sum_{j=1}^{T/k}\tr \left( \mathbf{L}_{j,j}^\T \mathrm{blkdiag}(\Delta^\T \Delta) \mathbf{L}_{j,j}\right)^2 \Bigg).
\end{multline*}
\end{theorem}
To appreciate the terms appearing in \Cref{thm:expineq}, it is worth to point out that $$\sum_{j=1}^{T/k} \tr\left( \mathbf{L}_{j,j}^\T \mathrm{blkdiag}(\Delta^\T \Delta) \mathbf{L}_{j,j}\right) = \sum_{t=1}^{T} \E \|\Delta \tilde X_t\|_2^2 .$$ Hence \Cref{thm:expineq} effectively passes the expectation inside the exponential at the cost of working with the possibly less excited process $\tilde X_{1:T}$ and a quadratic correction term. Note also that the assumption that $T$ is divisible by $k$ is not particularly important. If not, let $T'$ be the largest integer such that $T'/k \in \N$ and $T'\leq T$ and apply the result with $T'$ in place of $T$.

The significance of \Cref{thm:expineq} is demonstrated by the following simple observation, which is just the Chernoff approach applied to the exponential inequality in \Cref{thm:expineq}.
\begin{lemma}\label{lem:pointwiseanticonc}
Fix an integer $k \in \N$, let $T \in N$ be divisible by $k$ and suppose $X_{1:T}$ is a $k$-causal process with independent $K^2$-sub-Gaussian increments. Suppose further that the diagonal blocks are all equal: $\mathbf{L}_{j,j}=\mathbf{L}_{1,1}$ for all $j\in [T/k]$. For every size-conforming matrix $\Delta$ we have that:
\begin{equation}\label{eq:lem:pointwiseanticonc}
    \mathbf{P} \left( \sum_{t=1}^{T}  \|\Delta X_t\|^2_2\leq \frac{1}{2} \sum_{t=1}^{T} \E \|\Delta\tilde  X_t\|^2_2 \right) \leq \exp \left(- \frac{T}{576K^2 k} \right).
\end{equation}
\end{lemma}

Note that \Cref{lem:pointwiseanticonc} only yields \emph{pointwise} control  of the empirical covariance---i.e. pointwise on the sphere $\mathbb{S}^{d-1}$ . By setting $\Delta = v \in \mathbb{S}^{d-1}$, the result holds for a fixed vector on the sphere, but not uniformly for all such vectors at once. Thus, returning to our over-arching goal of providing control of the smallest eigenvalue of the empirical covariance matrix \eqref{eq:empcov}, we now combine \eqref{eq:lem:pointwiseanticonc} (using $d'=1$) with a union bound. This  approach yields \Cref{thm:anticonc}, of which the proof---along with that of its supporting lemmas---is given in full in \Cref{sec:anticoncproofs}.\footnote{Similar results can also be obtained for restricted eigenvalues.}

\subsection{Notes}

In this manuscript we have chosen a perhaps less well-known but conceptually simpler approach to establishing lower bounds on the empirical covariance matrix  \Cref{eq:lambdalowerbound}. The first proof of a statement similar to \Cref{thm:anticonc} is due to \cite{simchowitz2018learning} which in turn relies on a more advanced notion from probability theory known as the small-ball method, due to \cite{mendelson2014learning}. The emphasis therein is on anti-concentration---which can hold under milder moment assumptions---rather than concentration. However, the introduction of this tool is not necessary for Gaussian (or sub-Gaussian) system identification. For instance, \citet{sarkar2019near} leverage the method of self-normalized martingales introduced in \Cref{sec:selfnorm} below. 

Our motivation for providing a different proof is to streamline the exposition as to fit control of the lower tail into the "standard machinery", which roughly consists of: (1) prove a family of scalar exponential inequalities, (2) invoke the Chernoff method, and (3) conclude by a discretization argument and a union bound to port the result from scalars to matrices. Our proof here follows this outline and emphasizes the exponential inequality in \Cref{thm:expineq}. We finally remark that the proof presented here is new to the literature and extends a result in \cite{ziemann2022note} from the Gaussian setting to the sub-Gaussian setting.

%% file: sections/linearregression.tex
\section{Self-Normalized Martingale Bounds}
\label{sec:selfnorm}


The objective in this section is to bound the operator and Frobenius norms of the self-normalized term of \eqref{eq:OLSerror}: 
\begin{align}
    \label{eq:self-normalized}
    \left(\sum_{t=1}^{T} V_t X_t^\T \right)\left(\sum_{t=1}^{T} X_t X_t^\T \right)^{-1/2 }.
\end{align}
This object has special structure. Firstly, in many cases of interest, e.g. the autoregressive model in \eqref{eq:arx}, the noise $V_t$ is independent of $X_k$ for all $k \leq t$.  
This is what provides martingale structure, as will be made precise shortly.  Secondly, it is self-normalized: if the covariates $X_t$ are large for some $t$, then any increase in the left sum will be compensated by an increase in the sum in the term on the right. Together, these properties make the object above a \textit{self-normalized martingale} term. 

To express results  generally and compactly, several definitions are in order.
\begin{definition}(Filtration and Adapted Process)
    A sequence of sub-$\sigma$-algebras $\curly{\calF_t}_{t=1}^T$ is said to be a \emph{filtration} if $\calF_t \subseteq \calF_k$ for $t\leq k$. A stochastic process $\curly{W_t}_{t=1}^T$ is said to be \emph{adapted} to the filtration $\curly{\calF_t}_{t=1}^T$ if for all $t\geq 1$, $W_t$ is $\calF_t$-measurable.
\end{definition}
Conditioning on a sub-$\sigma$-algebra provides partial information about the total randomness. Therefore, the requirement that a filtration is non-decreasing captures the fact that information is not forgotten. An adapted process is one in which all the randomness at a particular time is explained by the information in the filtration up to that time. 

\begin{definition}(Martingale) Consider a stochastic process $\curly{W_t}_{t=1}^T$ which is adapted to a filtration $\curly{\calF_t}_{t=1}^T$. This process is called a \emph{martingale} if for all $1 \leq t \leq T$, $W_t$ is integrable and for all $1 \leq t < T$,  $\E[W_{t+1}\vert \calF_t] = W_t$. 
\end{definition}
Martingales model causal or non-anticipative processes. To better appreciate this, note that the increments $W_{t+1}-W_t$ are mean zero and conditionally orthogonal to the past; they can be thought of as the "next step" in a random walk whose path is traced out by $W_t$. 

In the context of the linear time-series model in \eqref{eq:regressionmodel}, we may define the sub-$\sigma$-algebras $\calF_t$ as those induced by the randomness up to time $t$: $\calF_t = \sigma(X_1, \dots, X_{t+1}, V_1, \dots, V_{t})$. In this case, the process $\curly{X_t}_{t=1}^T$ is adapted to the filtration $\curly{\calF_{t-1}}_{t=1}^T$ and the process $\curly{V_t}_{t=1}^T$ is adapted to the filtration $\curly{\calF_{t}}_{t=1}^T$.

Recall now again that the "numerator" in the ordinary least squares error is \eqref{eq:self-normalized}. We see that if we define the sum, $S_t \triangleq \sum_{s=1}^t V_s X_s^\top$, then the process $\curly{S_t}_{t=1}^T$ is adapted to $\curly{\calF_{t}}_{t=1}^T$. Furthermore, $\E\paren{S_{t+1}\vert \calF_{t}} = S_t + \E\paren{V_{t+1} \vert \calF_{t}} X_{t+1}^\top$. In particular, as long as the noise has conditional mean zero ($\E\paren{V_{t+1} \vert \calF_t} = 0$), the process $\curly{S_t}_{t=1}^T$ is a martingale.\footnote{Indeed, $S_t$ is a so-called martingale transform of $X_{1:t}$.} Normalizing the sum $S_t$ by the covariates as $S_t \paren{\sum_{s=1}^t X_s X_s^\top}^{-1/2}$ \emph{almost} preserves the martingale structure. Expressions of this type are called self-normalized martingales---although we stress that they are not strictly speaking martingales but only constructed from them.  

We now state bounds on the operator and Frobenius norms of the self-normalized martingale. The main idea behind the result---the technique of pseudo-maximization---is due to \cite{robbins1970boundary}. The formulations presented here are a consequence of Theorem 3.4 in \cite{abbasi2013online}. 
\begin{theorem}(Special cases of Theorem 3.4 in \cite{abbasi2013online})
    \label{thm:self-normalized martingale}
    Let $\curly{\calF_t}_{t=0}^T$ be a filtration such that $\curly{X_t}_{t=1}^T$ is adapted to $\curly{\calF_{t-1}}_{t=1}^T$ and $\curly{V_t}_{t=1}^T$ is adapted to $\curly{\calF_t}_{t=1}^T$. Additionally, suppose that for all $1 \leq t \leq T$, $V_t$ is $\sigma^2$-conditionally sub-Gaussian with respect to $\calF_{t}$. Let $\Sigma$ be a positive definite matrix in $\R^{\dx \times \dx}$.  For a fixed $T \in \mathbb{N}_+$ and $\delta \in (0,1)$, with probability at least $1-\delta$,
    \begin{align*}
        \bignorm{\sum_{t=1}^T V_t X_t^\top \paren{\Sigma + \sum_{t=1}^T X_t X_t^\top}^{-1/2}}_F^2 \leq \dy \sigma^2 \log\paren{\frac{\det\paren{\Sigma + \sum_{t=1}^T X_t X_t^\top}}{ \det(\Sigma)}} + 2\sigma^2 \log\frac{1}{\delta}. 
    \end{align*}
    Additionally, for a fixed $T \in \mathbb{N}_+$ and $\delta \in (0,1)$, with probability at least $1-\delta$,
    \begin{align*}
        &\bigopnorm{\sum_{t=1}^T V_t X_t^\top \paren{\Sigma + \sum_{t=1}^T X_t X_t^\top}^{-1/2}}^2 \\
        &\qquad \leq 4 \sigma^2 \log\paren{\frac{\det\paren{\Sigma + \sum_{t=1}^T X_t X_t^\top}}{ \det(\Sigma)}}+8\dy\sigma^2\log 5+ 8\sigma^2\log\frac{1}{\delta}.
    \end{align*}
\end{theorem}
Note that the quantities bounded above have a positive definite matrix $\Sigma$ added to the normalization quantities that was not present in the original term of interest, \eqref{eq:self-normalized}. Furthermore, the covariates $\sum_{t=1}^T X_t X_t$ appear in the upper bound. Hence, one typically combines the self-normalized martingale bound with some weak form of concentration.\footnote{Alternatively, in the analysis of ridge regression, $\Sigma$ takes the role of the penalizing matrix which can be tuned.} This is done in \Cref{sec:sysid}.


In the sequel, we prove the above bounds. To obtain the bound on the Frobenius norm, we directly consider the object
 \begin{align}
        \label{eq: matrix self normalized martingale}
        \bignorm{\left(\sum_{t=1}^T V_t X_t^\T \right)\left(\sum_{t=1}^T X_t X_t^\T \right)^{-1/2 }}_F,
\end{align}
while to obtain the bound on the operator norm we consider the following vector norm for an arbitrary unit vector $w$ as an intermediate step:
 \begin{align}
        \label{eq: vector self normalized martingale}
        \bignorm{\left(w^\top \sum_{t=1}^T V_t X_t^\T \right)\left(\sum_{t=1}^T X_t X_t^\T \right)^{-1/2 }}_2
\end{align}
and combine with a covering argument (recall \Cref{subsec:eps}).

Let us also briefly consider the dimensional dependencies of the Frobenius and operator norm bounds in \Cref{thm:self-normalized martingale}. The leading term in the Frobenius norm bound is $\dy$ multiplied by the $\log\det$ term, which scales with $\dx \log T$ when the empirical covariance is well-conditioned. In particular, the leading term scales with $\dx \dy \log T$. The factor of $\dy$ is no longer present on the $\log \det$ term for the operator norm. The term therefore scales as $\dx \log T$ when the the empirical covariance matrix is well-conditioned. There is, however, a term $8\dy\sigma^2\log 5$ which results from the covering argument. The operator norm bound therefore scales as $\max\curly{\dx  \log T, \dy}$. 

\subsection{Exponential Inequalities via Pseudo-maximization}
\label{sec:pseudomax}
We begin by neglecting the details of the process that generated the data in \eqref{eq:self-normalized}. In particular, consider a random matrix $P$ assuming values in $\R^{\da \times \dx}$ for $\deta \in \mathbb{N}_+$ and a random matrix $Q$ assuming values in $\R^{\dx \times \dx}$ with $Q$ almost surely nonsingular. Bounding the quantities in \eqref{eq: matrix self normalized martingale} and \eqref{eq: vector self normalized martingale} are special cases of bounding $\norm{P Q^{-\frac{1}{2}}}_F$. A naive first approach is to apply a Chernoff bound \eqref{eq:chernofftrick}. Doing so results in the inequality
\begin{align*}
    \bfP\paren{\norm{PQ^{-1/2}}_F \geq x} \leq \min_{\lambda \geq 0} \exp\paren{-\frac{\lambda}{2}x^2} \E\exp\paren{\frac{\lambda}{2}\norm{PQ^{-1/2}}_F^2}.
\end{align*}
If it is possible to bound the moment generating function $\E\exp\paren{\frac{\lambda}{2}\norm{PQ^{-1/2}}_F^2}$ by one for some $\lambda > 0$, then the above bound provides an exponential inequality. Obtaining a bound of the form $\E\exp\paren{\frac{\lambda}{2}\norm{PQ^{-1/2}}_F^2} \leq 1$ requires very strong assumptions on $P$ and $Q$ which would not be suitable for our purposes. However, we may observe that $\frac{1}{2} \norm{PQ^{-1/2}}_F^2 = \max_{\Lambda} \tr(P\Lambda  - \frac{1}{2}\Lambda^\top Q \Lambda)$. This motivates the following canonical assumption in self-normalized process theory: 
\begin{align}
    \label{eq: canoncial assumption}
    \max_{\Lambda \in \R^{\dx \times \da}} \E\exp\tr\paren{P \Lambda - \frac{1}{2} \Lambda^\top Q \Lambda} \leq 1. 
\end{align}
This inequality is called the canonical assumption because a wide variety of self-normalized processes satisfy it. We will demonstrate that it is satisfied for \eqref{eq: matrix self normalized martingale} and \eqref{eq: vector self normalized martingale} in \Cref{s: verifying canonical assumption}. 
\sloppy If we could exchange the order of the maximization with the expectation in \eqref{eq: canoncial assumption}, then the bound $\E\exp\paren{\frac{1}{2}\norm{PQ^{-1/2}}_F^2} \leq 1$ would be satisfied, and the Chernoff bound above would provide a valuable exponential inequality. As this exchange is not possible, we instead lower bound the maximization over $\Lambda$ by assigning a probability distribution to a random variable $\Psi$ which takes values in $\R^{\dx \times \da}$, and taking the expectation over this distribution. Doing so preserves the inequality \eqref{eq: canoncial assumption}: 
\begin{align*}
    \E \E\left[\exp\tr\paren{P\Psi - \frac{1}{2} \Psi^\top Q \Psi} \Big| \Psi\right] \leq 1.
\end{align*}
The order of expectation over $\Psi$ and over the random variables $P$ and $Q$ may then be exchanged by an appeal to Fubini's theorem:
\begin{align}
    \label{eq: relaxation of canoncial assumption}
    1 \geq \E \E \left[\exp\tr\paren{P\Psi - \frac{1}{2} \Psi^\top Q \Psi} \Big| \Psi\right] 
    =
    \E \E\left[\exp\tr\paren{P\Psi - \frac{1}{2} \Psi^\top Q \Psi} \Big| P, Q\right].
\end{align}
By selecting the distribution over $\Psi$ appropriately, the result is a so-called \textit{pseudo-maximization}. In particular, by completing the square, the inner conditional expectation on the right may be expressed as
\begin{equation}
\begin{aligned}
    \label{eq: self norm pseudo max}
    &\E \left[ \exp\tr\paren{P\Psi - \frac{1}{2} \Psi^\top Q \Psi} \Big| P, Q\right] 
    \\
    &=
    \exp\tr\paren{P Q^{-1} P^\top/2} \E\left[\exp\tr\paren{-\frac{1}{2}(\Psi-Q^{-1} P^\top)^\top Q (\Psi - Q^{-1} P^\top)} \Big| P, Q\right].
\end{aligned}
\end{equation}
For particular choices of the distribution of $\Psi$, the right side of the above expression approximates the maximimum value, $\exp\tr\paren{P Q^{-1} P^\top/2}$, of $\exp\tr(P \Lambda - \frac{1}{2} \Lambda^\top Q \Lambda)$. This allows us to apply a Chernoff argument similar to the one sketched above to obtain an exponential bound on a quantity related to $\norm {PQ^{-1/2}}_F$. The following lemma demonstrates one such bound that results by selecting the distribution of $\Psi$ as a matrix normal distribution. 


\begin{lemma}[Extension of Theorem 14.7 in \cite{pena2009self}]
    \label{lem: method of mixtures bound}
    Suppose that \eqref{eq: canoncial assumption} is satisfied. Let $\Sigma$ be a positive definite matrix in $\R^{\dx \times \dx}$. Then, for $\delta > 0$, with probability at least $1-\delta$,
    \begin{align*}
        \norm{P(Q+\Sigma)^{-1/2}}_F^2 \leq 2 \log\paren{\frac{\det(Q + \Sigma)^{\da/2} \det(\Sigma)^{-\da/2}}{\delta}}.
    \end{align*}
\end{lemma}

\subsection{Self-Normalized Martingales Satisfy the Canonical Assumption}
\label{s: verifying canonical assumption}

In order to make use of \Cref{lem: method of mixtures bound} to bound \eqref{eq: matrix self normalized martingale} or \eqref{eq: vector self normalized martingale}, we must ensure that the condition \eqref{eq: canoncial assumption} holds for 
\begin{align*}
    P = \sum_{t=1}^T \frac{\eta_t X_t^\top}{\sigma} , \quad Q = \sum_{t=1}^T X_t X_t^\top,
\end{align*}
where $\eta_t \in \R^{\deta}$ is either the noise process $V_t$ or the scalar process $w^\top V_t$ for some fixed unit vector $w$. The following lemma shows that it is sufficient for $\eta_t$ to be $\sigma^2$-conditionally sub-Gaussian. 

\begin{lemma}  
    \label{lem: canoncial assumption verified}
    Fix $T \in \mathbb{N}_+$. Let $\curly{\calF_t}_{t=0}^T$ be a filtration such that $\curly{X_t}_{t=1}^T$ is adapted to $\curly{\calF_{t-1}}_{t=1}^T$ and $\curly{\eta_t}_{t=1}^T$ is adapted to $\curly{\calF_t}_{t=1}^T$. Additionally, suppose that for all $t \geq 1$, $\eta_t$ is $\sigma^2$-conditionally sub-Gaussian with respect to $\calF_{t}$.
    Let $\Lambda \in \R^{\dx \times \deta}$ be arbitrary and consider for $t \in \curly{1,\dots, T}$
    \begin{align*}
        M_{t}(\Lambda) \triangleq  \exp\tr \paren{\sum_{s=1}^t \brac{\frac{\eta_s X_s^\top \Lambda}{\sigma}  - \frac{1}{2} \Lambda^\top X_s X_s^\top \Lambda}}. 
    \end{align*}
    Then $\E M_{T}(\Lambda) \leq 1$.
\end{lemma}

Combining the above exponential inequality with the ideas outlined in \Cref{sec:pseudomax} along with a covering argument yields \Cref{thm:self-normalized martingale}.

\subsection{Notes}

     \Cref{thm:self-normalized martingale} holds for a fixed $T \in \mathbb{N}_+$, which is sufficient for analyzing the system identification error. In contrast, the self-normalized margtingale bound in \cite{abbasi2013online} holds for an arbitrary stopping time and thus uniformly for all $T \in \mathbb{N}_+$ by a stopping time construction. This uniform bound may be required in some settings, e.g. in error bounds for adapative control.

%% file: sections/sysid.tex
\section{System Identification}\label{sec:sysid}
In this section, we analyze well-known linear system identification algorithms that rely on the least squares algorithm. Note that the problem formulation changes with the system parameterization (e.g., state space, ARMAX, etc.). However, a nice property of linear systems is that under certain conditions, we can obtain a linear non-parametric ARX model by regressing the system output to past outputs and inputs. Then, depending on the parameterization, we can recover a particular realization.
In the following, we first review ARX identification, which can be seen as a fundamental building block for many linear system identification algorithms. Then, we analyze identification of Markov parameters in the case of state-space systems. We focus exclusively on the case of single trajectory data.

\subsection{ARX Systems}
Consider an unknown vector autoregressive system with exogenous inputs (ARX)
\begin{equation}\label{eq:sysid_arx}
    \begin{aligned}
        Y_{t} &= \sum_{i=1}^{p} A^\star_i Y_{t-i} +\sum_{j=1}^{q} B^\star_j U_{t-j} +\Sigmaw^{1/2} W_{t},
    \end{aligned}
\end{equation}
where $Y_t\in\R^{\dy}$ are the system outputs, $U_t\in\R^{\du}$ are the control (exogenous) inputs, and $W_t\in\R^{\dy}$ is the normalized process noise with $\Sigmaw\in\R^{\dy\times \dy}$ capturing the (non-normalized) noise covariance. Matrices $\As_i,$ $i\le p$ and $\Bs_j$, $j\le q$ contain the unknown ARX coefficients. For the initial conditions, we assume $Y_{-1}=\cdots=Y_{-p}=0$, $U_{-1}=\cdots=U_{-q}=0$.
\begin{assumption}[System and Noise model]\label{assum:arx_noise}
Let the noise covariance $\Sigmaw\succ 0$ be full rank. Let the normalized process noise $W_t,t\ge 0$ be independent and identically distributed, $K^2$-sub-Gaussian~(see Definition~\ref{def:sub_G}), with zero mean and unit covariance $\E W_tW^\top_t=I_{\dy}$. The orders $p,q$ are known. System~\eqref{eq:sysid_arx} is non-explosive, that is, the eigenvalues of matrix 
\begin{equation}\label{eq:arx_A_matrix}  \mathcal{A}_{11}\triangleq\matr{\As_1&\As_{2}&\cdots&\As_{p-1}&\As_p\\I&0&\cdots&0&0\\0&I&\cdots&0&0\\ \vdots &&\ddots& &\vdots\\0&0&\cdots &I&0 },\,
\end{equation}
lie strictly on or inside the unit circle $\rho(\calA_{11})\le 1$.
\end{assumption}
The techniques below can provide meaningful finite-sample bounds only when the system is non-explosive. Deriving finite sample guarantees for identifying open-loop, explosively unstable partially-observed systems from single trajectory data is open to the best of our knowledge~\citep{tsiamis2022statistical}.

In this tutorial, we focus solely on white-noise excitation inputs. Excitation strategies---experiment designs---beyond this simple structure form a vast literature in system identification and statistics, see for instance~\cite{lennart1999system, gevers2005identification, pukelsheim2006optimal} and~\cite{bombois2011optimal} and the references therein. For a recent non-asymptotic treatment see also \cite{wagenmaker2020active}.

\begin{assumption}[White-noise excitation policy]\label{assum:arx_excitation}
We assume that the control input is generated by a random i.i.d. Gaussian process, that is, $U_t\sim\mathcal{N}(0,\sigma^2_u I)$.
\end{assumption}
Grouping all covariates into one vector and defining
\begin{equation}\label{eq:ARX_parameters_batch}
    X_t=\begin{bmatrix}Y^\top_{t-1:t-p}&U^\top_{t-1:t-q} \end{bmatrix}^\top,\, \Ms=\matr{\As_{1:p}&\Bs_{1:q}}
\end{equation}
we can re-write~\eqref{eq:sysid_arx} in terms of~\eqref{eq:regressionmodel}
\[
Y_t=\Ms X_t+\Sigmaw^{1/2} W_t,
\]
where $W_t$ is independent from $X_t$, but $X_t$ has the special time-dependent structure induced by~\eqref{eq:sysid_arx}. Given samples $(Y_{1:T},U_{0:T-1})$, the least-squares estimate is given by
\begin{equation}\label{eq:ARX_ls_solution}
    \Mls_T\triangleq \sum_{t=1}^T Y_t X^\top_t \left(\sum_{t=1}^{T} X_t X_t^\T \right)^\dagger,
\end{equation}
where we purposely highlight the dependence of the estimate on the number of samples with the subscript $T$.
Before we present the main result, let us define some quantities which are related to the quality of system estimates. The covariance at time $t\ge 0$ is defined as
\begin{equation}\label{eq:arx_covariance}
\CovX_t\triangleq \E X_tX^\top_t.
\end{equation}
It captures the expected richness of the data, i.e., how excited the modes of the system are on average. In particular, the relative excitation of the data compared to the noise magnitude significantly affects the quality of system identification. This motivates the definition of signal-to-noise (SNR) as the ratio of the (directionally) worst-case excitation over the worst-case noise magnitude
\begin{equation}\label{eq:arx_snr}
\snr_t\triangleq \frac{\lambda_{\min}(\CovX_t)}{\opnorm{\Sigmaw}K^2}.
\end{equation}

The following theorem provides a finite-sample upper bound on the performance of the least-square estimator.
\begin{theorem}[ARX Finite-Sample Bound]\label{thm:arx_pac}
   Let $(Y_{1:T},U_{0:T-1})$ be single trajectory input-output samples generated by system~\eqref{eq:sysid_arx} under Assumptions~\ref{assum:arx_noise},~\ref{assum:arx_excitation} for some horizon $T$.  Fix a failure probability $0<\delta<1$ and a time index $\tau\ge \max\{p,q\}$. Let 
    $\Tburn(\delta,\tau)\triangleq\min\{t:t\ge T_0(t,\delta/3,\tau)\}$, where $T_0$ is defined in~\eqref{eq:arx_burn_in}.
   If $T\ge \Tburn(\delta,\tau)$, then with probability at least $1-\delta$
\begin{equation}\label{eq:arx_pac}
\opnorm{\Mls_T-\Ms}^2
\le
\frac{C}{ \snr_{\tau} T} \left(
(p\dy+q\du)\log\frac{p\dy+q\du}{\delta}+\log\det\paren{\CovX_T\CovX^{-1}_\tau}\right),
\end{equation}
where  $C$ is a universal constant, i.e., it is independent of system,  confidence $\delta$ and index $\tau$. 
\end{theorem}
For non-explosive systems, matrix $\CovZ_T\CovZ^{-1}_{\tau}$ increases at most polynomially with $T$ in norm (in view of Lemma~\ref{lem:non_explosivity_powers}). Hence, the identification error decays with a rate of $\tilde{O}(1/\sqrt{T})$.

\textbf{Dimensional dependence.} Ignoring logarithmic terms, the bound implies that the number of samples $T$ should scale linearly with the dimension $p\dy+q\du$ of the covariates $X_t$. Since every sample $Y_t$ contains at least $\dy$ measurements, this implies that the total number of measurements should be linear with $\dy\times(p\dy+q\du)$. This scaling is qualitatively correct since $\Ms$ has $\dy\times(p\dy+q\du)$ unknowns, requiring at least as many independent equations.

\textbf{Logarithmic dependence on confidence.} The error norm scales linearly with $\sqrt{\log 1/\delta}$. In the asymptotic regime we also recover the same order of $\sqrt{\log 1/\delta}$ by applying the Central Limit Theorem (CLT). However, in the regime of finite samples, obtaining the rate is non-trivial, see~\cite{tsiamis2022statistical}, and requires the analysis presented in this tutorial.

\textbf{System theoretic constants.} The identification error is directly affected by the SNR of the system. The more the system is excited and the smaller the noise, the better the SNR becomes. However, excitability varies heavily depending on the system and the choice of excitation policy. In particular, the system's controllability structure can affect the degree of excitation dramatically. Systems with poor controllability structure can exhibit SNR which suffers from curse of dimensionality, i.e., the smallest eigenvalue of $\Sigma_{\tau}$ degrades exponentially with the system dimension~\cite{tsiamis2021linear}.

The upper bound also increases with the logarithm of the ``condition number" $\det(\CovZ_T\CovZ^{-1}_{\tau})$. For stable systems, this condition number is bounded since $\CovZ_T$ converges to a steady-state covariance as $T$ increases; we can neglect it in this case. On the other hand, the term might be significant in the case of general non-explosive systems. Let $\kappa$ be the size of the largest Jordan block of $\calA_{11}$ with eigenvalues on the unit circle. Then, this term can be as large as $\kappa \log T$.

 \textbf{Burn-in time.} The upper bound holds as soon as the number of samples exceeds a ``burn-in" time $\Tburn(\delta,\tau)$. If the system is non-explosive, $\Tburn(\delta,\tau)$ is always finite for fixed $\tau$. Exceeding the burn-in time guarantees that we have persistency of excitation, that is, all modes of the system are excited. The burn-in time increases as we require more confidence $\delta$ and we chose larger  time indices $\tau$. On the other had, larger $\tau$ leads to larger $\Sigma_{\tau}$, which improves the $\snr_{\tau}$. In other words, there is a tradeoff between improving the SNR and deteriorating the burn-in time. We analyze persistency of excitation in more detail, in the next subsection. 

\textbf{Proof outline.} We outline the proof below, the full proofs can be found in Appendix~\ref{appendix:sys_id}. 
To analyze the least squares error, observe that
\begin{equation}
\Mls_T-\Ms=\underbrace{\sum_{t=1}^T\Sigmaw^{1/2} W_tX^\top_t \paren{\sum_{t=1}^{T} X_t X_t^\top }^{-1/2}}_{\text{noise}} \times\underbrace{\paren{\sum_{t=1}^{T} X_t X_t^\top}^{-1/2}}_{\text{excitation}} 
\end{equation}
where we assumed that the inverse exists. To deal with the second term, we will prove persistency of excitation in finite time leveraging the techniques of Section~\ref{sec:onesided}, which requires most of the work. 
To deal with the noise part we will apply the self-normalized martingale methods, which are reviewed in Section~\ref{sec:selfnorm}.  We study both terms in the following subsections.

\subsubsection{Persistency of Excitation in ARX Models}
In this subsection, we leverage the result of Theorem~\ref{thm:anticonc} to prove persistency of excitation. By persistency of excitation, we refer to the case when we have rich input-output data, that is, data which characterize all possible behaviors of the system. 
Recall the definition~\eqref{eq:empcov} of the empirical covariance matrix
\[
\ECov_T\triangleq \frac{1}{T}\sum_{t=1}^{T}X_tX_t^\top.
\]
Using this definition, the excitation term in the least squares error can be re-written as $(T\ECov_T)^{-1/2}$.
We say that persistency of excitation holds if and only if the empirical covariance matrix is strictly positive definite (full rank). In the following, we show that the full rank condition holds with high probability, provided that the number of samples exceeds a certain threshold, i.e., the burn-in time. 

\begin{theorem}[ARX PE]\label{thm:arx_pe}
  Let $(Y_{1:T},U_{0:T-1})$ be input-output samples generated by system~\eqref{eq:sysid_arx} under Assumptions~\ref{assum:arx_noise},~\ref{assum:arx_excitation} for some fixed horizon $T$.  Fix a failure probability $0<\delta<1$ and a time index $\tau\ge \max\{p,q\}$. Then, $\lambda_{\min}(\Sigma_{\tau})>0$. Moreover, if $T$ is large enough
    \begin{equation}\label{eq:arx_burn_in}
T\ge T_0(T,\delta,\tau)\triangleq  1152\tau \max\{K^2,1\}\Big((p\dy +q\du )  \log \Csys(T,\tau) +\log (1/\delta)\Big)
    \end{equation}
     where
  \begin{equation*}
     \Csys(T,\tau) \triangleq \frac{2T}{3\tau}\frac{\opnorm{\CovX_T}^2}{\lambda^2_{\min}(\CovX_{\tau})},
 \end{equation*}
then, 
    \[
\Prob\paren{\ECov_T \succeq \frac{1}{16}\CovX_{\tau}}\ge 1-\delta.
    \]
\end{theorem}
%
The detailed proof can be found in Appendix~\ref{appendix:sys_id}, we only sketch the main ideas here. The first step is to express the covariates $X_{1:T}$ as a causal linear combination of the noises and the inputs, mimicking~\eqref{eq:linearcausalprocess}. The evolution of the covariates follows a state-space recursion
\begin{equation}\label{eq:arx2ss}
X_{t+1}=\mathcal{A}X_t+\mathcal{B}V_{t},
\end{equation}
where we concatenate noises and inputs $V_t\triangleq \matr{W^\top_{t}&U^\top_{t}}^\top$. Matrices $\mathcal{A},\,\mathcal{B}$ are given by
\begin{align*}
&\mathcal{A}\triangleq \matr{\calA_{11}&\calA_{12}\\0&\calA_{22}},\quad \calB\triangleq \matr{\calB_{1}&\calB_{2}},\text{ where }\\
&\calA_{11}\text{ is defined in~\eqref{eq:arx_A_matrix}}, \,&&\\&\mathcal{A}_{12}\triangleq \matr{1\\0\\\vdots\\0}\otimes\matr{B^\star_1&\cdots&B^\star_q}\in\R^{p\dy\times q\du},\,
&&\mathcal{A}_{22}\triangleq \matr{0&\cdots&0&0&0\\I_{\du}&\cdots&0&0&0\\ \vdots &\ddots& &\cdots&\\ 0&\cdots&I_{\du}&0&0\\0&\cdots&0&I_{\du}&0}\in\R^{q\du\times q\du}\qquad \qquad\qquad \\
&\mathcal{B}_1\triangleq\matr{\Sigmaw^{1/2}&0_{((p-1)\dy+q\du)\times \dy}}^\top,\,&&\mathcal{B}_2\triangleq \matr{0_{p\dy\times \du}&I_{\du}&0_{(q-1)\du\times \du}}^{\top}.
\end{align*}
The vector $X_{1:T}$ of all covariates satisfies the following causal linear relation
\begin{equation}
 \label{eq:LforARX}
X_{1:T}=\underbrace{\matr{\mathcal{B}&0&\cdots&0\\\mathcal{A}\mathcal{B}&\mathcal{B}&\cdots&0\\ \vdots & &\ddots& \\\mathcal{A}^{T-1}\mathcal{B}&\mathcal{A}^{T-2}\mathcal{B}&\cdots&\mathcal{B}}}_{\bfL}V_{0:T-1}.
\end{equation}
where the lower-block triangular matrix is the Toeplitz matrix generated by the Markov parameters matrices $\calA,\,\calB$.
The second step is to apply Theorem~\ref{thm:anticonc}. The details can be found in the Appendix. 

\begin{remark}[Existence of burn-in time.] For the above result to be meaningful, we need inequality~\eqref{eq:arx_burn_in} to be feasible. For non-explosive systems, the system theoretic term $\log\Csys(T,\tau)$ increases at most logarithmically with $T$, since $\CovZ_t$ increases polynomially with $t$ in view of Lemma~\ref{lem:non_explosivity_powers}. Hence, for any fixed $\tau$, or, in general, any $\tau$ that increases mildly (sublinearly) with $T$, e.g. $O(\sqrt{T})$, it is possible to satisfy~\eqref{eq:arx_burn_in}. Note that $\rho(\calA)=\rho(\calA_{11})$ due to the triangular structure of $\calA$. Hence, by Assumption~\ref{assum:arx_noise}, system $\calA$ is also non-explosive.
\end{remark}

\begin{remark}[Unknown system orders $p,q$]
The result of Theorem~\ref{thm:arx_pe} still holds if the orders $p,q$ are unknown and we use the wrong orders $\hat{p},\hat{q}$ in the covariates $X_t$. We just need to replace $p,q$ with $\hat{p},\hat{q}$ with $\hat{p},\hat{q}$ and revise the size of $\CovZ$ accordingly in~\eqref{eq:arx_burn_in}. The finite-sample bounds of Theorem~\ref{thm:arx_pac} also hold (by revising accordingly), but only if we overestimate $p,q$, that is $\hat{p}\ge p$, $\hat{q}\ge q$. This also generalizes the single trajectory result of~\cite{du2022sample} to non-explosive systems.
\end{remark}
The following supporting lemma proves that the $k$-th powers of non-explosive matrices increase at most polynomially with $k$.
\begin{lemma}[Lemma~1 in~\cite{tsiamis2021linear}]\label{lem:non_explosivity_powers}
 Let $\calA\in\R^{d\times d}$ have all eigenvalues inside or on the unit circle, with $\opnorm{\calA}\le M$, for some $M>0$. Then, 
\begin{equation}\label{eq:powers_bound}
    \opnorm{\calA^k}\le (ek)^{d-1}\max\set{M^{d},1}.
\end{equation}   
\end{lemma}
 As a corollary, the covariance matrices $\CovX_t$ also grow at most polynomially with the time $t$.
\subsubsection{Dealing with the Noise Term}
In this subsection, we modify the noise term so that we can leverage Theorem~\ref{thm:self-normalized martingale}, which cannot be applied directly. We first manipulate the inverse of $T\ECov_T$ to relate it to the inverse of $\Sigma+T\ECov_T$, for some carefully selected $\Sigma$. 
Inspired by~\cite{sarkar2019near}, we leverage the result of Theorem~\ref{thm:arx_pe}. Under the event that persistency of excitation holds we have $\ECov_T\succeq T\CovZ_{\tau}/16$. Thus, selecting $\Sigma=T\CovZ_{\tau}/16$ guarantees that
\[
\paren{ T\ECov_T }^{-1}\preceq 2\paren{ \Sigma+T\ECov_T }^{-1}.
\]
We can now apply Theorem~\ref{thm:self-normalized martingale}. 
To finish the proof we need to upper-bound the determinant of $\log\det(\Sigma+T\ECov_T)$. It is sufficient to establish a crude upper-bound on the empirical covariance $T\ECov_T$ as in the following lemma.
\begin{lemma}[Matrix Markov's inequality]\label{lem:arx_Markov}
    Fix a failure probability $\delta>0$. With probability at least $1-\delta$
    \begin{equation}\label{eq:arx_Markov}
     \ECov_T\preceq \frac{p\dy+q\du}{\delta}\CovX_T.
    \end{equation}
\end{lemma}
A more refined upper bound can also be applied (see e.g. the proof of \Cref{prop:sparse} below or the results in \cite{jedra2022finite}).
\subsection{State-Space Systems}
In this subsection, we derive finite-sample guarantees for learning Markov parameters of linear systems in state-space form. Consider the following state-space system in the so-called innovation form:
\begin{equation}
\begin{aligned}\label{eq:ss_innovation_equation}
X_{t+1}&=\As X_t+\Bs U_t+\Ks \Sigmae^{1/2} E_t\\
Y_t&=\Cs X_t+\Sigmae^{1/2} E_t,
\end{aligned}
\end{equation}
where $\As\in\R^{\dx\times\dx}$, $\Bs\in\R^{\dx\times \du}$, $\Ks\in\R^{\dx\times\dy}$, and $\Cs\in\R^{\dy\times\dx}$ are \emph{unknown} state-space parameters. For the initial condition, we assume $X_0=0$. We call the normalized noise process $E_t$ the innovation error process.
Similar to the ARX case, we focus on white-noise excitation inputs, namely Assumption~\ref{assum:arx_excitation} also holds here. Moreover, we assume the following.
\begin{assumption}[System and Noise model]\label{assum:ss_noise}
Let the noise covariance $\Sigmae\succ 0$ be full rank. Let the normalized innovation process $E_t$ be independent, identically distributed, $K^2$-sub-Gaussian~(see Definition~\ref{def:sub_G}), with zero mean and unit covariance $ \E E_tE^\top_t=I_{\dy}$. The order $\dx$ is unknown. System~\eqref{eq:ss_innovation_equation} is non-explosive, that is, the eigenvalues of matrix $\As$
lie strictly on or inside the unit circle $\rho(A)\le 1$. The system is also minimum-phase, i.e., the closed loop matrix 
\begin{equation}
\Acs\triangleq \As-\Ks\Cs
\end{equation}
has all eigenvalues inside the unit circle $\rho(\Acs)< 1$.
\end{assumption}
The innovation form~\eqref{eq:ss_innovation_equation} might seem puzzling at first. In particular, the correlation between process and measurement noise via $\Ks$, and the requirement $\rho(\Acs)<1$ seem restrictive. However, the representation~\eqref{eq:ss_innovation_equation} is standard in the system identification literature~\cite{verhaegen2007filtering}. Moreover, as we review below, standard state-space models have input-output second-order statistics, which are equivalent to the  ones generated by system~\eqref{eq:ss_innovation_equation} (for appropriate $\Ks$, $\Sigmae$). 

\begin{remark}[Generality of model]\label{rem:KF_equivalent_to_ss}
System class~\eqref{eq:ss_innovation_equation} captures general state-space systems driven by Gaussian noise.
Consider the following state-space model
\begin{equation}\label{eq:ss_standard_equation}
\begin{aligned}
S_{t+1}&=\As S_t+\Bs U_t+W_t\\
Y_t&=\Cs S_t+V_t,
\end{aligned}
\end{equation}
where $W_t,V_t$ are i.i.d., independent of each other, mean-zero Gaussian, with covariances $\Sigmaw$ and $\Sigmav$ respectively. Assume that $\Sigmav\succ 0$ is full rank, the pair $(\Cs,\As)$ is detectable, and the pair $(\As,\Sigmaw)$ is stabilizable. These three assumptions imply that the Kalman filter of system~\eqref{eq:ss_standard_equation} is well-defined~\citep{anderson2012optimal}. In particular, define the Riccati operator as
\begin{equation}\label{eq:ss_Ric}
    \RIC(P)\triangleq \As P(\As)^\top+\Sigmaw-\As P(\Cs)^\top(\Cs P (\Cs)^\top+\Sigmav)^{-1} \Cs P(\As)^\top
\end{equation}
and let $\Ps$ be the unique positive semidefinite solution of $\Ps=\RIC(\Ps)$. Then the Kalman filter gain is equal to 
\begin{equation}\label{eq:ss_KF_gain}
\Ks=-\As P(\Cs)^\top(\Cs P (\Cs)^\top+\Sigmav)^{-1}.
\end{equation}
Assume that the initial state is also mean-zero Gaussian with covariance $\Ps$ and independent of the noises. Finally set 
\begin{equation}\label{eq:ss_KF_innovation_error_cov}
\Sigmae=\Cs\Ps(\Cs)^\top+\Sigmav.
\end{equation}

Under the above assumptions and selection of $\Ks$, $\Sigmae$ systems~\eqref{eq:ss_innovation_equation} and~\eqref{eq:ss_standard_equation} are statistically equivalent from an input-output perspective, see~\cite{qin2006overview}. Both system descriptions lead to input-output trajectories with identical statistics. Moreover, due to the properties of Kalman filter, stability of $\Acs$ (minimum phase property) and independence of $E_t$ are satisfied automatically~\citep{anderson2012optimal}. 
\end{remark}

In this tutorial we will only focus on recovering the first few (logarithmic in $T$-many) Markov parameters $\Cs(\Acs)^{i}\Bs$, $i\ge 0$ and $\Cs(\Acs)^{j}\Ks$, $j\ge 0$ of  system~\eqref{eq:ss_innovation_equation}. From a learning theory point of view, this is also known as improper learning, since the search space (finitely many Markov parameters) does not exactly, but only approximately, coincide with the hypothesis class (state space models). In principle, this forms the backbone of the SSARX method introduced by \citet{jansson2003subspace}. One can then proceed to recover the original state-space parameters (up to similarity transformation) from the Markov parameters by employing some realization method. We refer to~\cite{oymak2021revisiting,tsiamis2022statistical} for a discussion on this approach from a finite sample perspective.

\subsubsection{Reduction to ARX learning with Bias}
Let $p>0$ be a past horizon. Denote the Markov parameters up to time $p$ by
\begin{equation}\label{eq:ss_Markov}
    \Ms_p\triangleq [\begin{array}{cccccc}\Cs\Bs&\cdots& \Cs(\Acs)^{p-1}\Bs&\Cs\Ks&\cdots& \Cs(\Acs)^{p-1}\Ks\end{array}].
\end{equation}
Note that the innovation errors are equal to $\Sigmae^{1/2}E_t=Y_t-\Cs X_t$. Replacing this expression into the state equation~\eqref{eq:ss_innovation_equation}, we obtain
\[
X_{t}=\Acs X_{t-1}+\Bs U_{t-1}+\Ks Y_{t-1}.
\]
Unrolling the state equation $p$ times, we get
\begin{equation}\label{eq:ss_arx_approximation}
Y_t=\underbrace{\Ms_p Z_t+\Sigmae^{1/2} E_t}_{\text{ARX}}+\underbrace{\Cs (\Acs)^{p}X_{t-p}}_{\text{bias}},
\end{equation}
where $Z_t$ includes the past $p$ covariates
\begin{equation}\label{eq:ss_parameters_batch}
    Z_t=\begin{bmatrix}Y^\top_{t-1:t-p}&U^\top_{t-1:t-p} \end{bmatrix}^\top.
\end{equation}
The above recursion is an approximate ARX equation. There is an additive bias error term on top of the statistical noise.  The least-squares solution is given by
\begin{equation}\label{eq:ss_ls_solution}
    \Mls_{p,T}\triangleq \sum_{t=1}^T Y_t Z^\top_t \left(\sum_{t=1}^{T} Z_t Z_t^\T \right)^\dagger,
\end{equation}
where we also highlight the dependence on the past $p$.
By the minimum phase assumption, the bias term decays exponentially with the past horizon $p$. This follows from the fact that $\Acs$ is asymptotically stable, while $X_{t}$ scales at most polynomially with $t$ (in view of Lemma~\ref{lem:non_explosivity_powers}). By selecting $p=\Omega(\log T)$, we can make the bias term decay very fast, making its contribution to the error $\Ms_p-\Mls_{p,T}$ negligible. On the other hand, increasing the past horizon $p$ increases the statistical error since the search space is larger.

\subsubsection{Non-Asymptotic Guarantees}
To derive finite-time guarantees for state space systems of the form~\eqref{eq:ss_innovation_equation}, we follow the same steps as in the case of ARX systems. However, we have to account for the bias term and the fact that $p$ grows with $\log T$.
Let us define again
the covariance at time $t\ge 0$
\begin{equation}\label{eq:ss_covariance}
\CovZ_{p,t}\triangleq \E Z_t Z^\top_t,
\end{equation}
where we highlight the dependence on both the past horizon $p$ and the time $t$. The covariance of the state is defined similarly
\begin{equation}\label{eq:ss_covariance_state}
\CovXState_{X,t}\triangleq \E X_t X^\top_t.
\end{equation}
Define the SNR as
\begin{equation}\label{eq:ss_snr}
\snr_{p,t}\triangleq \frac{\lambda_{\min}(\CovZ_{p,t})}{\opnorm{\Sigmae}K^2}.
\end{equation}
Unlike the ARX case, here the SNR might degrade since we allow $p$ to grow with $\log T$. For this reason, we require the following additional assumption.
\begin{assumption}[Non-degenerate SNR]\label{assum:ss_non_degenerate_snr}
We assume that the SNR is uniformly lower bounded for all possible past horizons
\[
\liminf_{t\ge 0} \snr_{t,t} >0.
\]
\end{assumption}
Later on, in Theorem~\ref{thm:non-degenerate_snr}, we show that the above condition is non-vacuous and is satisfied for quite general systems. 
\begin{theorem}[State Space Finite-Sample Bound]\label{thm:ss_pac}
   Let $(Y_{1:T},U_{0:T-1})$ be single trajectory input-output samples generated by system~\eqref{eq:ss_innovation_equation} under Assumptions~\ref{assum:arx_excitation}, ~\ref{assum:ss_noise},~\ref{assum:ss_non_degenerate_snr}, for some horizon $T$. Fix a failure probability $0<\delta<1$ and select $p=\beta \log T$, for $\beta$ large enough such that
   \begin{equation}\label{eq:ss_choice_of_p}
\opnorm{\Cs(\Acs)^p} \opnorm{\CovXState_{X,T}}\le T^{-3}. 
\end{equation}
Let  
    $\Tburnss(\delta,\beta)\triangleq\min\{t:t\ge T_0(t,\delta,\beta\log t)\}$, where $T_0$ is defined in~\eqref{eq:arx_burn_in}.
   If $T\ge \Tburnss(\delta,\beta)$, then with probability at least $1-2\delta$
\begin{equation}\label{eq:ss_pac}
\opnorm{\Mls_{p,T}-\Ms_p}^2
\le
\frac{C_1}{ \snr_{p,p} T} \left(
p(\dy+\du)\log\frac{p(\dy+\du)}{\delta}+\log\det\paren{\CovZ_{p,T}\CovZ^{-1}_{p,p}}\right),
\end{equation}
where  $C_1$ is a universal constant, i.e., it is independent of system,  confidence $\delta$ and past horizon $p$. 
\end{theorem}
For non-explosive systems, matrix $\CovZ_{p,T}\CovZ^{-1}_{p,p}$ increases at most polynomially with $T$ in norm. Since the SNR is uniformly lower bounded, the identification error decays with a rate of $\tilde{O}(1/\sqrt{T})$. The bound seems similar to the one for ARX systems for $\tau=p$. However, since $p\asymp\log T$, we have an extra logarithmic term. 

\textbf{Role of $\beta$.} Recall the approximate ARX relation~\eqref{eq:ss_arx_approximation}. For the bias term to be small, the exponentially decaying $(\Acs)^p$ should counteract the magnitude of the state $\opnorm{X_{t-p}}$. Intuitively, the state grows as fast as $\opnorm{\CovX_{X,t}}^{1/2}$, where $\CovX_{X,t}=\E X_t X^{\top}_t$. Hence the state norm grows at most polynomially with $T$. Meanwhile, $\opnorm{(\Acs)^p}=O(\rho^p)$ for some $\rho>\rho(\Acs)$. With the choice $p=\beta\log T$, we get $\opnorm{(\Acs)^p}=O(T^{-\beta/\log(1/\rho)})$. Hence, if we select large enough $\beta$, we can make the bias term very small, even smaller than the dominant $\tilde{O}(1/\sqrt{T})$ term.

\textbf{Burn-in time.} Since the system is non-explosive, $\Tburnss(\delta,\beta)$ is always finite under Assumption~\ref{assum:ss_non_degenerate_snr}, for any $\beta$. As before, exceeding the burn-in time guarantees that we have persistency of excitation. Naturally, larger $\beta$ lead to larger past horizons $p$, which, in turn, increase the burn-in time.


Finally, we prove that Assumption~\ref{assum:ss_non_degenerate_snr} is non-vacuous. It is sufficient for $\Ks$ and $\Sigmaw$ to be generated by a Kalman filter as in~\eqref{eq:ss_KF_gain},~\eqref{eq:ss_KF_innovation_error_cov}.
\begin{theorem}\label{thm:non-degenerate_snr}
Consider system~\eqref{eq:ss_innovation_equation} and the definition of $\snr_{p,t}$ in~\eqref{eq:ss_snr}. If 
    matrices $\Ks$, $\Sigmae$ are generated as in~\eqref{eq:ss_KF_gain},~\eqref{eq:ss_KF_innovation_error_cov}  with $(\As,\Sigmaw^{1/2})$ stabilizable, $(\Cs,\As)$ detectable and $\Sigmav\succ 0$,
then the SNR is uniformly lower bounded $\liminf_{t\ge 0} \snr_{t,t}>0$.
\end{theorem}
Both conditions are sufficient. It is subject of future work to extend the result to more general non-exposive systems.

\subsection{Notes}
The exposition above is inspired by prior work on identifying fully-observed systems~\citep{faradonbeh2018finite,simchowitz2018learning,sarkar2019near} and partially-observed systems~\citep{oymak2019non,simchowitz2019semi,sarkar2021finite,tsiamis2019finite,lee2019non,lale2021finite,lee2022improved}. For a wider overview of the literature, we refer the reader to~\cite{tsiamis2022statistical}. 


Let us further remark that the guarantee for the ARX model in \Cref{thm:arx_pac} is almost optimal. The use of Matrix Markov's inequality yields extraneous dependency on the problem dimension multiplying the deviation term $\log (1/\delta)$. This can in principle be removed by a more refined analysis (see e.g. the proof of \Cref{prop:sparse} below or the results in \cite{jedra2022finite}). Indeed, the signal-to-noise term \eqref{eq:arx_snr} is closely related to the Fisher Information Matrix appearing in the classical asymptotic optimality theory. Let us also point out that the question of optimality in identifying partially observed state-space systems is more subtle, and while consistent, the bounds presented here are not (asymptotically) optimal.


%% file: sections/basicineq.tex
\section{An Alternative Viewpoint: the Basic Inequality}\label{sec:basicineq}
In many situations, the choice of the model class $\mathsf{M}= \R^{\dy \times \dx}$ leading to \eqref{eq:OLSdef} is not appropriate. For instance physical or other modelling considerations might have already informed us that the true $\theta^\star$ belongs to some smaller model class such as the family of low rank or sparse matrices which are strict subsets of $\mathsf{M}$. Other properties one might wish to enforce include, stable, low norm, or even passivity-type properties. In either of the above examples no error expression of the form \eqref{eq:OLSerror} is directly available. Instead, we observe  by optimality of $\widehat{M}$ to the optimization program \eqref{eq:LSEdef} that
\begin{equation}\label{eq:opteqnLSE}
    \frac{1}{T}\sum_{t=1}^{T} \| Y_t -\widehat \theta X_t\|_2^2  \leq \frac{1}{T}\sum_{t=1}^{T} \| Y_t -\theta^\star X_t\|_2^2. 
\end{equation}
Expanding the squares and re-arranging terms we arrive at the so-called basic inequality of least squares:
\begin{equation}\label{eq:basicineqLSE}
    \frac{1}{T} \sum_{t=1}^T \|(\widehat \theta -\theta^\star )X_t\|_2^2 \leq  \frac{2}{T} \sum_{t=1}^T \langle V_t,(\widehat \theta -\theta^\star )X_t\rangle.
\end{equation}
The inequality \eqref{eq:basicineqLSE} serves as an alternative to the explicit error equation \eqref{eq:OLSerror}. To drive home this point, let us first re-arrange \eqref{eq:basicineqLSE} slightly:
\begin{equation}\label{eq:offsetbasicineqLSE}
    \frac{1}{T} \sum_{t=1}^T \|(\widehat \theta -\theta^\star )X_t\|_2^2 \leq  \frac{4}{T} \sum_{t=1}^T \langle V_t,(\widehat \theta -\theta^\star )X_t\rangle-\frac{1}{T} \sum_{t=1}^T \|(\widehat \theta -\theta^\star )X_t\|_2^2.
\end{equation}
Note now that $\widehat \theta -\theta^\star$ are elements of $\mathsf{M}_\star \triangleq \mathsf{M}-\theta^\star$. Hence---by considering the worst-case (supremum) right hand side of \eqref{eq:offsetbasicineqLSE}---we obtain:
\begin{equation}\label{eq:offsetbasicineqLSE2}
    \frac{1}{T} \sum_{t=1}^T \|(\widehat \theta -\theta^\star )X_t\|_2^2 \leq \sup_{\theta \in \mathsf{M}_\star}  \left\{\frac{4}{T} \sum_{t=1}^T \langle V_t, \theta X_t\rangle-\frac{1}{T} \sum_{t=1}^T \| \theta  X_t\|_2^2\right\}.
\end{equation}
In fact, if $\mathsf{M}=\R^{\dy\times \dx}$, the optimization on the right hand side of \eqref{eq:offsetbasicineqLSE2} has an explicit solution. This implies that we always have the following upper-bound on the event that the design is nondegenerate:
\begin{equation}\label{eq:offset_to_selfnorm}
    \begin{aligned}
        &\sup_{\theta \in \mathsf{M}_\star}  \left\{\frac{4}{T} \sum_{t=1}^T \langle V_t,\theta  X_t\rangle-\frac{1}{T} \sum_{t=1}^T \|\theta X_t\|_2^2\right\}
        \\
        &
        \leq \sup_{\theta \in \R^{\dy\times \dx}}  \left\{\frac{4}{T} \sum_{t=1}^T \langle V_t,\theta  X_t\rangle-\frac{1}{T} \sum_{t=1}^T \|\theta  X_t\|_2^2\right\} 
        &&(\mathsf{M}_\star \subset \R^{\dy\times \dx})
        \\
        &=
        \frac{4}{T}\left\|\left(\sum_{t=1}^{T} V_t X_t^\T \right)\left(\sum_{t=1}^{T} X_t X_t^\T \right)^{-1/2 } \right\|^2_F. && (\textnormal{direct calculation})
    \end{aligned}
\end{equation}
Hence, we have in principle recovered an in-norm version of \eqref{eq:OLSerror} with slightly worse constants. Put differently, we may regard \eqref{eq:offsetbasicineqLSE2} as a variational (or dual) form of the explicit error \eqref{eq:OLSerror}. Now, the advantage of \eqref{eq:offsetbasicineqLSE2} is twofold:
\begin{enumerate}
    \item \eqref{eq:offsetbasicineqLSE2} and \eqref{eq:offset_to_selfnorm} hold for any $\mathsf{M}_\star \subset \R^{\dy\times \dx}$ and hence allows us to analyze the LSE \eqref{eq:LSEdef} beyond OLS ($\mathsf{M}_\star = \R^{\dy\times \dx}$). This is important in identification problems where the parameter space is restricted.
    \item We do not have to rely on \eqref{eq:offset_to_selfnorm} to control \eqref{eq:offsetbasicineqLSE2}. In fact, for many reasonable classes of $\mathsf{M}_\star \subset \R^{\dy\times \dx}$ we are able to give alternative arguments that are much sharper (in terms of e.g. dimensional scaling) than the naive bound \eqref{eq:offset_to_selfnorm}. See \Cref{subsec:sparsity} below.
\end{enumerate}
A third advantage of the variational form \eqref{eq:offsetbasicineqLSE2} is that it generalizes straightforwardly beyond linear least squares. In fact, none of the steps \eqref{eq:opteqnLSE},\eqref{eq:basicineqLSE}, \eqref{eq:offsetbasicineqLSE} and \eqref{eq:offsetbasicineqLSE2} relied on the linearity of $ x \mapsto \widehat \theta x$ or that of  $x \mapsto \theta^\star x$ ($x\in \R^{\dx}$). We will explore this theme further in \Cref{subsec:sparsity} and \Cref{sec:nonlinear}.

\subsection{Sparse Autoregressions}\label{subsec:sparsity}
Before we proceed to sketch out how the basic inequality above extends to nonlinear problems in \Cref{sec:nonlinear}, let us use it to analyze a simple variation of the autoregression already encountered in \Cref{sec:sysid}. Namely, the autoregressive model \eqref{eq:sysid_arx} which---for simplicity---is further assumed one-dimensional:
\begin{equation}\label{eq:1dsparseAR}
    Y_t = \sum_{i=1}^pA^\star_i Y_{t-i} + W_t \qquad (W_{1:T} \textnormal{ \iid\, mean zero and $\sigma^2$-subG})
\end{equation}
and assume in addition that it is known that only $s\in \N$ of the $p$ entries of $\theta^\star = [A^\star_1,\dots,A^\star_p]$ are nonzero. Put differently, the vector $\theta^\star$ is known to be $s$-sparse and we write $\theta^\star \in \{\theta \in \R^{p} : \|\theta\|_0 \leq s \} \triangleq \mathsf{M}$. Hence, in this case the model class $\mathsf{M}$ is the union of $ {p} \choose s$ subspaces. Clearly, we could use OLS \eqref{eq:OLSdef} but this estimator does not take advantage of the additional information that $A^\star =\theta^\star$ lies in the $s$-dimensional submanifold $\mathsf{M}$. Intuitively, if $s\ll p$ this set should be much smaller than $\R^{p}$ and so one expects that identification occurs at a faster rate. 

In this section we demonstrate that the least squares estimator \eqref{eq:LSEdef} in which the search is restricted to the low-dimensional manifold $\mathsf{M}$ outperforms the OLS. We stress that this is \emph{not} a computationally efficient estimator and the results in this section should be thought of as little more than an illustrative example.

Returning to the problem of controlling the error of this estimator, we note that in this case there is no closed form for the LSE and we do not have direct access to the error equation \eqref{eq:OLSerror}.\footnote{Although, in this particular case an alternative analysis based on this equation is possible.} Hence, we instead use the offset basic inequality approach from \Cref{sec:basicineq}. As before, it is convenient to set $X_t = [ Y_{t-1},\dots, Y_{t-p}]^\T$. With this additional bit of notation in place, we recall from \eqref{eq:offsetbasicineqLSE2} that:
\begin{equation}
 \frac{1}{T} \sum_{t=1}^T \|(\widehat \theta -\theta^\star )X_t\|_2^2 \leq \max_{\theta \in \mathsf{M}_\star}  \left\{\frac{4}{T} \sum_{t=1}^T  W_t \theta X_t-\frac{1}{T} \sum_{t=1}^T | \theta  X_t|_2^2\right\}
\end{equation}
where $\mathsf{M}_\star$ is the translation $\mathsf{M}-\theta^\star$. Since $\mathsf{M}$ is the union of $ {p} \choose s$-many linear $s$-dimensional subspaces $S\subset \R^{\dx\times \dx}$, $\mathsf{M}_\star$ is the union of $ {p} \choose s$ affine subspaces $s$-dimensional affine subspaces of the form $S-\theta^\star$. Let us also note that $\mathsf{M}_\star \subset \mathsf{M}-\mathsf{M} =\{\theta \in \R^{p} : \|\theta\|_0 \leq 2s \}  $. Consequently:
\begin{equation}\label{eq:sparseoffset}
\begin{aligned}
 \frac{1}{T} \sum_{t=1}^T \|(\widehat \theta -\theta^\star )X_t\|_2^2 
&\leq \max_{\theta \in \mathsf{M}_\star}  \left\{\frac{4}{T} \sum_{t=1}^T  W_t \theta X_t-\frac{1}{T} \sum_{t=1}^T | \theta  X_t|_2^2\right\} \\
&\leq \max_{S } \max_{\theta \in S}  \left\{\frac{4}{T} \sum_{t=1}^T W_t \theta X_t-\frac{1}{T} \sum_{t=1}^T | \theta  X_t|_2^2\right\}.
\end{aligned}
\end{equation}
where maximization over $S$ occurs over the $p \choose 2s$-many sparse subspaces. Notice now that since $\theta$ in $\eqref{eq:sparseoffset}$ is $s$-sparse, the products $\theta X_t$ are just $\theta X_t = \sum_{i \in S} \theta_i (X_t)_i $ where we have abused notation and identified $S$ with its support set. Hence, by the same direct calculation as in \eqref{eq:offset_to_selfnorm}, if we denote $(X_t)_S$ the $s$-dimensional vector obtained by coordinate projection onto part of $S$ not constrained to be identically zero (i.e. the image of the projection onto $S$ represented as the $s$-dimensional Euclidean space) we find that:
\begin{equation}\label{eq:sparsebasicineq}
    \frac{1}{T} \sum_{t=1}^T \|(\widehat \theta -\theta^\star )X_t\|_2^2 \leq  \frac{4}{T} \max_S  \left\|\left(\sum_{t=1}^{T} W_t (X_t)_S \right)\left(\sum_{t=1}^{T} (X_t)_S (X_t)_S^\T \right)^{-1/2 } \right\|^2_2.
\end{equation}
The right hand side of \eqref{eq:sparsebasicineq} can be controlled by the self-normalized inequality in \Cref{thm:self-normalized martingale} for each fixed $S$. Moreover, there are only $p \choose 2s$ such subspaces, so we can apply a union bound to control the maximum over these subspaces. Note also that the left hand side of \eqref{eq:sparsebasicineq} can be controlled by the tools developed in \Cref{sec:onesided}. Carrying out these steps leads to the following guarantee.


\begin{proposition}\label{prop:sparse}
    Fix $T,k \in \N$ with $T$ divisible by $k$ and let $\mathbf{L}$ be the linear operator defined in \eqref{eq:LforARX}. Let $\widehat \theta$ be the LSE \eqref{eq:LSEdef} over the set $\mathsf{M}=\{\theta \in \R^{p} : \|\theta\|_0 \leq s \}$ for the system \eqref{eq:1dsparseAR}. Define $\Sigma_j\triangleq  \frac{1}{j}\sum_{t=1}^j \E X_t X_t^\T $ for $j \in [T]$ and 
    \begin{equation*}
        \mathrm{cond}_{\mathsf{sys}}(T,k)\triangleq \left(1+\frac{  \opnorm{\mathbf{L} \mathbf{L} ^\T}}{ k \lambda_{\min}\left(\Sigma_T \right)} \right) \frac{\lambda_{\max }\left(\Sigma_T \right)}{\lambda_{\min}\left(\Sigma_k\right) }.
    \end{equation*}
    There exist univeral positive constants $c,c'$ such that for 
    any $\delta \in  (0,1)$ it holds with probability at least $1-\delta$  that:
    \begin{equation}\label{eq:sparseguarantee}
           \|(\widehat \theta -\theta^\star )\sqrt{\Sigma_k}\|_2^2
           \leq
             c \sigma^2 \times  \frac{    s\log \left(\frac{p \times  \mathrm{cond}_{\mathsf{sys}}(T,k)}{s}   \right) +\log (1/\delta) }{T}
    \end{equation}
    as long as
    \begin{equation}\label{eq:sparseburnininthm}
        T/k \geq  c' \sigma^2 \left( s \left[\log \left( \mathrm{cond}_{\mathsf{sys}}(T,k)\right)+\log(p/s)  \right] + \log(1/\delta)\right).
    \end{equation}
\end{proposition}
A few remarks are in order. The guarantee \eqref{eq:sparseguarantee} depends on the dimension $s$ of $\mathsf{M}$, and not the total parameter dimension $p$. Similarly, the burn-in \eqref{eq:sparseburnininthm} exhibits a similar win, by depending linearly on $s$ and only logarithmically on $p$. There is also the 
difference that the left hand side of \eqref{eq:sparseguarantee} is given in the problem-dependent Mahalanobis norm induced by $\Sigma_k$ and opposed to just the standard Euclidean $2$-norm. This implies that if we actually want parameter identification in the sense of the previous section, a restricted eigenvalue condition on $\Sigma_k$ is needed.\footnote{Note that $\widehat \theta - \theta^\star$ is $2s$-sparse.} Indeed, for some positive number $\lambda_{\mathrm{restricted}}$, one requires that $v^\T \Sigma_k v \geq \lambda_{\mathrm{restricted}}$ for all $2s$-sparse vectors $v$ on the unit sphere: $v \in \mathbb{S}^{p-1}$ and $\|v\|_0\leq 2s$. Obviously the requirements on $\lambda_{\mathrm{restricted}}$ are much milder than the corresponding ones on $\lambda_{\min}(\Sigma_k)$ and we always have $\lambda_{\mathrm{restricted}} \geq \lambda_{\min}(\Sigma_k)$.

The following lemma is central. Namely, we begin the proof of \Cref{prop:sparse} by restricting to an event in which the designs $\sum_{t=1}^{T} (X_t)_S (X_t)_S^\T $ are sufficiently well-conditioned for all the subspaces $S$ at once. The requirements on this event are relatively milder than the corresponding one over $\R^p$ and explains the "dimensional win" (when $s \ll p$) of the sparse estimator over OLS. 

\begin{lemma}\label{lem:sparseisomorph}
    Let $\mathbf{L}$ be the linear operator defined in \eqref{eq:LforARX}. Fix $\delta \in (0,1)$ and let $T$ be divisible by $k\in \N$.    There exist universal positive constants $c_1,c_2,c_3 \in \R$ such that
    the following two-sided control holds uniformly in $S$ with probability $1-\delta$:
    \begin{multline}\label{eq:loosetwosidedsparse}
          \frac{c_1}{k}\sum_{t=1}^{k} \E  \left[(X_t)_S (X_t)_S^\T\right] \preceq \frac{1}{T}\sum_{t=1}^{T} (X_t)_S (X_t)_S^\T 
          \\
          \preceq c_2\left(1+\frac{ T \opnorm{\mathbf{L} \mathbf{L} ^\T}}{ k \lambda_{\min}\left(\sum_{t=0}^{T-1} \E X_t X_t^\T \right)} \right)\left(  \frac{1}{T}\sum_{t=1}^{T} \E  \left[(X_t)_S (X_t)_S^\T\right] \right)
    \end{multline}
     as long as 
    \begin{equation}\label{eq:sparseburnin}
        T \geq c_3 \sigma^2 \left( s \left[\log C_{\mathsf{sys}}+\log(p/s)  \right] + \log(1/\delta)\right).
    \end{equation}
\end{lemma}
\Cref{eq:sparseburnin} is  revealing about the advantage of using the sparse estimator searching over $\mathsf{M}=\{\theta \in \R^{p} : \|\theta\|_0 \leq s \}$. The burn-in period in \eqref{eq:sparseburnin} is proportional to the dimension of the low-dimensional parameter manifold $\mathsf{M}$ instead of that of the latent space $\R^p$. Finally, as usual we have relegated the full proof of \Cref{prop:sparse} to the appendix, see \Cref{subsec:proofsforbasic}.

\subsection{Notes}
The variational formulation of the least squares error---the basic inequality \eqref{eq:basicineqLSE}---is standard in the nonparametric statistics literature \citep[see e.g.][Chapters 13 and 14]{wainwright2019high}. The idea to rewrite the basic inequality \eqref{eq:basicineqLSE} as \eqref{eq:offsetbasicineqLSE} was introduced to the statistical literature by \cite{liang2015learning}, but has its roots in online learning  \citep{rakhlin2014online}.

%% file: sections/extension.tex
\section{Beyond Linear Models}
\label{sec:nonlinear}

Let us now make another gradual shift of perspective. Instead of considering the linear model \eqref{eq:regressionmodel} introduced in \Cref{subsec:probfo} we consider the following \emph{nonlinear} regression model:
\begin{equation}\label{eq:nonlinearregression}
    Y_t = f^\star(X_t)+V_t, \qquad t\in[T].
\end{equation}
As before, $Y_{1:T}$,$X_{1:T}$ and $V_{1:T}$ are stochastic processes taking values in $\R^{\dy}$ and $\R^{\dx}$ respectively. However, this time $f^\star$ is no longer constrained to be a linear map of the form $x \mapsto Ax$ for matrix $A$. Rather, we suppose that $f^\star$ in \eqref{eq:nonlinearregression} belongs to some (square integrable) space of functions $\scrF$ such that $\scrF \ni f :  x \mapsto f(x)$. It is perhaps now that the motivation behind the change of perspective from \Cref{sec:basicineq} becomes most apparent: the basic inequality \eqref{eq:offsetbasicineqLSE} remains valid. To be precise, let us define the \emph{nonlinear}  least squares estimator
\begin{equation}\label{eq:nonlinearlsedef}
    \widehat f \in \argmin_{f\in \scrF} \left\{ \frac{1}{T}\sum_{t=1}^T\|Y_t-f(X_t)\|_2^2 \right\}.
\end{equation}

Let $\scrF_\star \triangleq \scrF -f^\star$. By the exact same optimality argument as in \Cref{sec:basicineq}, the reader can now readily verify that:
\begin{equation}\label{eq:nonparamoffset}
    \frac{1}{T} \sum_{t=1}^{T} \| \widehat f (X_t) -f^\star(X_t) \|^2_2 \leq \sup_{f\in \scrF_\star} \frac{1}{T}\left(\sum_{t=1}^{T} 4\langle V_t,  f(X_t)\rangle -  \sum_{t=1}^{T} \| f(X_t) \|^2_2\right).
\end{equation}

What does \eqref{eq:nonparamoffset} entail in terms of estimating the unknown function $f^\star$? To answer this, we first need to define a performance criterion. The simplest one is small average $L^2$-norm-error, where
\begin{equation} \label{eq:l2norm}
  f\in \scrF : \quad  \| f\|_{L^2}^2 \triangleq \frac{1}{T}\sum_{t=1}^T \E \|f(X_t)\|^2_2 .
\end{equation}

The program we have carried out in the previous sections now generalizes as follows:
\begin{itemize}
    \item First, prove a so-called lower uniform law. That is to say, we wish to show that with overwhelming probability
    \begin{equation}\label{eq:loweruniformlaw}
        \|f-f_\star\|_{L^2}^2 \leq \frac{C}{T}\sum_{t=1}^T  \|f(X_t)-f_\star(X_t)\|^2 \quad (\textnormal{simultaneously }\forall f \in \scrF).
    \end{equation}
    for some universal positive constant $C$.
    \item Second, control  the supremum of the \emph{empirical process}:
    \begin{equation}\label{eq:empprocess}
        f \mapsto \left(\sum_{t=1}^{T} 4\langle V_t,  f(X_t)\rangle -  \sum_{t=1}^{T} \| f(X_t) \|^2_2\right)
    \end{equation}
    in terms of the noise level $\sigma$ and the complexity of the class $\scrF$.
\end{itemize}

By combining \eqref{eq:loweruniformlaw} and \eqref{eq:empprocess} we arrive at a high probability bound of the form:
\begin{equation}\label{eq:informalnonlinear}
    \|\widehat f-f^\star\|^2_{L^2} \leq \frac{C}{T}\sum_{t=1}^T  \|f(X_t)-f_\star(X_t)\|^2 \leq \frac{C \times \mathrm{comp}(\scrF,\sigma^2)+\textnormal{deviation term}}{T}.
\end{equation}
A statement of this form is given as \Cref{thm:nonlinearthm} below.

\begin{remark}\label{rem:upperlowerremark}
It is worth to take pause and appreciate the analogy to the analysis of linear regression models. The first step \eqref{eq:loweruniformlaw} exactly corresponds to controlling the lower spectrum of the empirical covariance matrix. Suppose for simplicity that $\dy=1$. Then for a linear map $\mathbb{S}^{\dx-1} \ni f \mapsto \langle f, x\rangle $ we have:
\begin{equation}
    \frac{1}{T}\sum_{t=1}^T  \|f(X_t)\|^2_2 = \frac{1}{T}\sum_{t=1}^T \langle f, (X_t X_t^\T) f \rangle = \left\langle f, \left[\frac{1}{T}\sum_{t=1}^T  (X_t X_t^\T) \right] f \right\rangle
\end{equation}
which are just the one-dimensional projections of the empirical covariance matrix \eqref{eq:empcov}. In the context of linear models, establishing \eqref{eq:loweruniformlaw} was the topic of \Cref{sec:onesided}. Analogously, for a linear predictor, the $L^2$-norm \eqref{eq:l2norm} becomes a Mahalanobis norm: $f \in \R^{\dx} \Rightarrow \|f\|_{L^2}^2 = \langle f, \Sigma_X f\rangle$ for $\Sigma_X = \frac{1}{T}\sum_{t=1}^T\E X_t X_t^\T$.

Moreover, For linear models, we had:
\begin{equation}
    \sup_{f\in \R^{\dx}} \left(\sum_{t=1}^{T} 4\langle V_t,  f(X_t)\rangle -  \sum_{t=1}^{T} \| f(X_t) \|^2_2 \right) = 4\left\|\left(\sum_{t=1}^{T} V_t X_t^\T \right)\left(\sum_{t=1}^{T} X_t X_t^\T \right)^{-1/2 } \right\|^2_F.
\end{equation}
Analyzing terms of this form was the topic of \Cref{sec:selfnorm}.

In other words, the approach outlined above is very much in the same spirit as that in the rest of the manuscript. There are a few changes that need to be made since we less access to linearity in our argument, but in principle the key difference is that we will have to replace the indexing set  $\mathbb{S}^{d-1}$ with a more general function class $\scrF_\star$.
\end{remark}

\subsection{Many Trajectories and Finite Hypothesis Classes}
In order to make the exposition self-contained, we will now make two simplifying assumptions relating to the finiteness of the hypothesis class and the dependence structure of the covariate process $X_{1:T}$. A more general treatment without these can be found in \cite{ziemann2022learning}. Here, we impose the following:

\begin{enumerate}
    \item[A1.] The hypothesis class $\scrF$ is finite; $|\scrF|< \infty$.
    \item[A2.] We have access to $T/k$-many independent trajectories from the same process: there exists an integer $k\in \N$ dividing $T$ such that $X_{1:k},X_{k+1:2k},\dots$ are drawn \iid.
\end{enumerate}
We will also impose the following rather minimal integrability condition:
\begin{enumerate}   
    \item[A3.] All functions $f\in \scrF$ are such that $\E \| f (X_t)\|_2^4 < \infty$ for all $t\in [T]$.
\end{enumerate}
Moreover, as in \Cref{sec:sysid}, we require the noise to be a sub-Gaussian martingale difference sequence:
\begin{enumerate}
    \item[A4.]  For each $t\in [T]$, $V_t | X_{1:t}$ is $\sigma^2$ conditionally-sub-Gaussian and mean zero.
\end{enumerate}

\begin{remark}
    Note that A4. above entails that $f_\star(x) = \E [Y_t | X_t=x]$ for every time instance $t$ and so the setup is akin to the study of "predictor models" from system identification~\citep[see e.g.][Chapter 2.6]{davis1985stochastic}.
\end{remark}

Under these assumptions, the main result of \cite{ziemann2022learning} essentially simplifies to the following theorem.

\begin{theorem}\label{thm:nonlinearthm}
    Impose A1-A4, fix $\delta \in (0,1)$ and define
    \begin{equation}\label{eq:hypcon}
        \mathrm{cond}_{\scrF} \triangleq \max_{f\in \scrF_\star} \max_{t\in T} \frac{\sqrt{\E\| f(X_t)\|^4_2}}{\E\| f(X_t)\|^2_2 }.
    \end{equation}
Suppose further that
    \begin{equation*}
        T/k \geq  4\mathrm{cond}_{\scrF}^2\left(\log|\scrF|+\log(2/\delta) \right) 
    \end{equation*}
    then we have that:
    \begin{equation}\label{eq:nonlinearthmguarantee}
        \|\widehat f -f^\star\|_{L^2}^2 \leq 16 \sigma^2\left(\frac{ \log (|\scrF|) +\log (2/\delta)}{T}\right).
    \end{equation}
\end{theorem}
A few remarks are in order. The structure of \Cref{thm:nonlinearthm} is by now familiar and it is very much of the same structure as our previous results, cf. \eqref{eq:informalsnr}. The key differences are that: (1) we now control the $L^2$ norm of our estimator instead of the Euclidean or spectral norm; and (2) that the dimensional dependency has been replaced by the complexity term $\log |\scrF_\star|$. The proof is also structurally similar, as noted in \Cref{rem:upperlowerremark}. We also caution the reader that \eqref{eq:nonlinearthmguarantee} is strictly a statistical guarntee; we have said nothing---and will say nothing more---about the computational feasibility of the estimator \eqref{eq:nonlinearlsedef}.

Let us now discuss A1-A4. Assumption A1 informs us that the search space for the LSE \eqref{eq:nonlinearlsedef} is finite. This is mainly imposed to avoid the introduction of the chaining technique which is the standard alternative to the bounds from \Cref{sec:selfnorm}. Using this technique, similar statements can for instance be derived for compact subsets of bounded function classes \citep{ziemann2022learning}. Assumption A2 controls the dependence structure of the process. Here, we assume that we are able to restart the process every $k$ time steps. Again, a more general statement relying on stochastic stability can be found in \cite{ziemann2022learning}. Assumption A3 is relatively standard. Arguably the strongest assumption is A4, which in principle yields that the conditional expectation (given all past data) is a function in the search space $\scrF$. It is a so-called realizability assumption---the model \eqref{eq:nonlinearregression} is well-specified---and it is not currently known how to remove it and still obtain sharp bounds beyond  linear classes \citep[for an analysis of linear misspecified models, see][]{ziemann2023noise}.

\subsection{Notes}
As noted in the previous section, the idea of using the ``offset'' basic inequality relied on here is due to \cite{rakhlin2014online,liang2015learning}. The ``many trajectores''-style of analysis used here is due to \cite{tu2022learning} who introduced it in the linear setting. Here, we have extended their style of analysis to simplify the exposition of \cite{ziemann2022learning} who consider the single trajectory setting, but rely on a rather more advanced exponential inequality due to \cite{samson2000concentration}.  Alternatively, one can also easily extend the lower uniform law in \Cref{prop:loweruniformlaw} to certain classes of mixing processes by invoking the blocking technique of \cite{yu1994mixing} combined with a truncation style of argument such as that used in the proof of Theorem 14.12 of \cite{wainwright2019high}, see also \citet[the proof of Theorem 4.3]{ziemann2023noise}. Note however, that all the analyses above and in this section necessitate some degree of stability (mixing). This should be contrasted with the system identification bounds of \Cref{sec:sysid}, which work even in the marginally regime. In principle, the consequence of this is that while the convergence rates for bounds such as \Cref{thm:nonlinearthm} are correct, the burn-ins are deflated by various dependency measures.

There have also been other, more algorithmically focused, approaches to nonlinear identification problems in the recent literature. Noteably, gradient based methods in generalized linear models of the form $X_{t+1}=\phi(A^\star X_t)+V_t$ (with $\phi$ a known nonlinearity) have been the topic of a number of recent papers \citep[see e.g.][]{foster2020learning,sattar2022non}. The sharpest bounds for parameter recovery in this setting are due to \cite{kowshik2021near}.  

%% file: acks.tex
\paragraph{Acknowledgements.}

Ingvar Ziemann is supported by a Swedish Research Council international postdoc grant. Nikolai Matni is funded by NSF awards CPS-2038873, CAREER award ECCS-2045834, and ECCS-2231349.

%% file: sections/proofhw.tex
\section{Proof of The Hanson-Wright Inequality}\label{sec:proofofhw}

Let us now proceed with the proof of the Hanson-Wright Inequality (\Cref{thm:HWineqhighproba}). We will need a few intermediate exponential inequalities relating sub-Gaussian variables to Gaussian counterparts.

\subsection{Gaussian Comparison Inequalities for sub-Gaussian Quadratic Forms}

\begin{lemma}\label{lem:subgchisq}
     Fix symmetric positive semidefinite $M\in \R^{d\times d}$ and let $W$ be a $\sigma^2$-sub-Gaussian assuming values in $\R^d$. Let further $G$ be an independent $N(0,I_d)$-distributed random variable. For every $\lambda \in \R$ with $|\lambda| \leq 1/(\sqrt{2}\sigma^2\opnorm{M})$:
     \begin{equation}
         \E \exp\left( \lambda W^\T M W\right) \leq\E \exp\left( \lambda\sigma^2 G^\T M G\right)\leq \exp\left(\lambda^2 \sigma^4 \|M\|_F^2\right)
     \end{equation}
\end{lemma}

\begin{proof}
    For the first inequality note that:
    \begin{equation}
        \begin{aligned}
             \E \exp\left( \lambda W^\T M W\right) &= \E \exp\left( \sqrt{2\lambda} W^\T \sqrt{M} G\right) &&(\textnormal{Gaussian MGF})\\
             &\leq    \E \exp\left( \lambda\sigma^2 G^\T M G\right). && (\textnormal{sub-Gaussian MGF})  
        \end{aligned}
    \end{equation}
    To obtain the second inequality, let $V S V^\T$ be the eigendecomposition of $M$ with orthonormal $V$. Let $v_i, i\in [d]$ be the columns of $V$. Then
    \begin{equation}
        G^\T M G = \sum_{i=1}^d  s_i^2 (v_i^\T G )^2
    \end{equation}
Note by Gaussian rotational invariance that the vector $G'$ with entries $v_i^\T G $ is equal in distribution to $G$. Hence:
    \begin{equation}
        \begin{aligned}
            \E \exp \left( \lambda \sigma^2 G^\T M G \right)
            &=\prod_{i=1}^d\E \exp \left(\lambda \sigma^2 s_i^2 (v_i^\T G )^2\right)\\
            &=\prod_{i=1}^d\E \exp \left(\lambda \sigma^2 s_i^2 G_i^2\right) &&(\textnormal{Gaussian Rotation Invariance})\\
             &=\prod_{i=1}^d \left(1-\lambda^2 \sigma^4 s_i^2 \right)^{-1/2} && (\textnormal{Gaussian-Squared MGF; $\lambda$ in our range})\\
             &= \prod_{i=1}^d \exp\left(\frac{-1}{2} \log (1-\lambda^2 \sigma^4 s_i^2)  \right) \\
             &\leq \prod_{i=1}^d  \exp\left( \lambda^2\sigma^4 s_i^2  \right) && \left(-\log (1-x)\leq 2x \textnormal{ if } x\in [0,1/2]\right)\\
             &=\exp\left(\lambda^2 \sigma^4 \|M\|_F^2\right) && \left(\sum_{i=1}^d s_i^2 = \|M\|_F^2\right)
        \end{aligned}
    \end{equation}
    as per requirement.
\end{proof}

\begin{lemma}[Comparison to Gaussians]\label{lem:comparisontogaussians}
    Fix $M\in \R^{d\times d}$ and let $W,W'$ be independent $\sigma^2$-sub-Gaussian assuming values in $\R^d$. Let further $G,G'$ be two independent $N(0,I_d)$-distributed random variables. For every $\lambda \in \R$:
    \begin{equation*}
        \E \exp\left( \lambda W^\T M W'\right) \leq \E \exp\left( \lambda {\sigma^2}G^\T M G' \right).
    \end{equation*}
\end{lemma}

\begin{proof}
    The result follows by straightforward computation, alternatingly using the sub-Gaussian property and the closed form of the Gaussian MGF. Namely:
    \begin{equation}
        \begin{aligned}
            \E\exp\left( \lambda W^\T M W'\right) &\leq\E \exp\left( \frac{\sigma^2\lambda^2 \| M W'\|_2^2}{2} \right) && (W\textnormal{ is } \sigma^2-\textnormal{subG})\\
            &=\E \exp \left( \lambda \sigma G^\T M W' \right) && (\textnormal{MGF of } G)\\
            &\leq\E \exp\left( \frac{\lambda^2\sigma^4 \| M G\|_2^2}{2} \right) && (W'\textnormal{ is } \sigma^2-\textnormal{subG})\\
            &= \E\exp\left( \lambda {\sigma^2} G^\T M G'\right).&& (\textnormal{MGF of } G')
        \end{aligned}
    \end{equation}
    The desired result has been established.
\end{proof}

\begin{lemma}[MGF of Gaussian Chaos]\label{lem:mgfgaussianchaos}
    Let $W,W' \sim N(0,I_d)$ be independent and let $M \in \R^{d\times d}$. Then for every $\lambda \in \R$ with $|\lambda|\leq 1/(\sqrt{2}\opnorm{M})$:
    \begin{equation}
        \E \exp \left( \lambda W^\T M W' \right)\leq \exp \left(\lambda^2 \|M\|_F^2 \right)
    \end{equation}
\end{lemma}
\begin{proof}
    By the singular value decomposition we may write $M = USV^\T$ where the columns $u_{1:d}$ ($v_{1:d}$) of $U$ (of $V$) form orthonormal bases of $\R^d$. Hence
    \begin{equation}
        W^\T M W = \sum_{i=1}^d s_i (W^\T u_i)((W')^\T v_i)
    \end{equation}
    where $s_{1:d}$ are the singular values of $M$.  Note next that $G_i\triangleq (W^\T u_i)$ and $G_i'\triangleq ((W')^\T v_i)$ are again independent and standard normal by Gaussian rotational invariance. Hence:
    \begin{equation}
        \begin{aligned}
            \E \exp \left( \lambda W^\T M W \right)&= \E \exp \left( \lambda\sum_{i=1}^d s_i G_i G_i'\right)\\
            &=\prod_{i=1}^d\E \exp \left( \lambda s_i G_i G_i'\right)
            &&(\textnormal{independence})\\
            &= \prod_{i=1}^d\E \exp \left( \lambda^2 s_i^2 G_i^2/2\right)&&(\textnormal{Gaussian MGF})\\
             &=\prod_{i=1}^d \left(1-\lambda^2 s_i^2 \right)^{-1/2} && (\textnormal{Gaussian-Squared MGF; $\lambda$ in our range})\\
             &= \prod_{i=1}^d \exp\left(\frac{-1}{2} \log (1-\lambda^2 s_i^2)  \right) &&(\exp \circ \log \textnormal{is the identity function}) \\
             &\leq \prod_{i=1}^d  \exp\left( \lambda^2 s_i^2  \right) && \left(-\log (1-x)\leq 2x \textnormal{ if } x\in [0,1/2]\right)\\
             &=\exp\left(\lambda^2 \|M\|_F^2\right) && \left(\sum_{i=1}^d s_i^2 = \|M\|_F^2\right)
        \end{aligned}
    \end{equation}
    as per requirement.
\end{proof}

\subsection{Finishing the proof of \Cref{thm:HWineqhighproba}}

We need one more preliminary tool before we arrive at the proof of \Cref{thm:HWineqhighproba}. The next Theorem constitutes a useful decoupling inequality which allows us to treat the mixed terms in the quadratic form $W^\T M W$ as independent.

\begin{theorem}[Theorem 6.1.1 in of \cite{vershynin_2018}]\label{thm:chaosdecoupled}
Let $W$ be a $d$-dimensional random vector with mean zero and independent entries. For every convex function $f$ and every $M=(m_{ij})_{i,j=1}^d\in \R^{d\times d}$ it holds true that:
\begin{equation}
    \E f\left(\sum_{i,j=1, i\neq j}^d m_{ij} W_iW_j \right) \leq \E f\left(4\sum_{i, j=1}^d m_{ij} W_iW_j'\right)
\end{equation}
where $W'$ is an independent copy of $W$ (i.e., equal to $W$ in distribution but independent of $W$).
\end{theorem}


\begin{proposition}[Hanson-Wright Exponential Inequality Form]\label{prop:HWexpineq}
Let $M\in \R^{d\times d}$. Fix a random variable $W=W_{1:d}$ where each $W_i, i\in [d]$ is an independent scalar $\sigma^2$-sub-Gaussian random variable. For every $\lambda \in \R$ with $|\lambda|\leq \frac{1}{8\sqrt{2}\sigma^2\opnorm{ M} }$ we have that:
\begin{equation}
   \max\left\{  \E \exp \left( \lambda W^\T M W -\lambda \E W^\T M W \right),  \E \exp \left( \lambda W^\T M W  \right) \right\} \leq \exp \left(36\lambda^2\sigma^4 \| M\|_F^2 \right)
\end{equation}
\end{proposition}

\begin{proof}
    By rescaling $\lambda$ we may assume without loss of generality that $\sigma^2=1$. Let now $m_{ij}$ with $i,j \in [d]$ denote the entries of $M$. We begin by writing
    \begin{equation}
        W^\T M W -\E W^\T M W = \sum_{i=1}^d m_{ii} (W_i^2-\E W_i^2) + \sum_{i\neq j} m_{ij} W_i W_j.
    \end{equation}
    Hence by the Cauchy-Schwarz inequality:
    \begin{multline}\label{eq:CSinHWused}
         \E \exp \left( \lambda W^\T M W -\lambda \E W^\T M W \right)\\
         \leq \sqrt{\E \exp \left( 2\lambda \sum_{i=1}^d m_{ii} (W_i^2-\E W_i^2)\right)}\times  \sqrt{\E \exp \left(2\lambda \sum_{i\neq j} m_{ij} W_i W_j \right)}.
    \end{multline}

    We proceed to analyze both terms appearing on the RHS of \eqref{eq:CSinHWused} separately. We begin by analyzing the diagonal term. Let $W'$ be an independent copy of $W$. Then:
    \begin{equation}\label{eq:HWCS_ic}
    \begin{aligned}
        &\E \exp \left( 2\lambda \sum_{i=1}^d m_{ii} (W_i^2-\E W_i^2)\right)\\
        &\leq\E \exp \left( 2\lambda \sum_{i=1}^d m_{ii} (W_i^2-\E (W_i')^2)\right) && (W=W'\textnormal{ in distribution})\\
        &\leq \E  \exp \left( 2\lambda \sum_{i=1}^d m_{ii} (W_i^2- (W_i')^2)\right) &&(\textnormal{Jensen's inequality})\\
        &=\E  \exp \left( 2\lambda \sum_{i=1}^d m_{ii} W_i^2\right)\E  \exp \left( 2\lambda \sum_{i=1}^d (-m_{ii} ) (W_i)^2\right). &&(\textnormal{independence})
    \end{aligned}
    \end{equation}

    Hence, if we combine \Cref{lem:subgchisq} with \eqref{eq:HWCS_ic} we find that:
    \begin{equation}\label{eq:HWmgfdiag}
        \sqrt{\E \exp \left( 2\lambda \sum_{i=1}^d m_{ii} (W_i^2-\E W_i^2)\right)} \leq \exp \left(4\lambda^2 \|M\|_F^2 \right)
    \end{equation}
    as long as $|\lambda| \leq 1/(2\sqrt{2}\opnorm{M})$. 

    Next, we argue similarly for the off-diagonal term and use \Cref{thm:chaosdecoupled} to control the off-diagonal term in \eqref{eq:CSinHWused}. Again letting $W'$ be an independent copy of $W$ and also letting $G,G'$ be two independent isotropic Gaussians in $\R^d$, we have that:
    \begin{equation}\label{eq:HWmgfoffidag}
        \begin{aligned}
            \E \exp \left(2\lambda \sum_{i\neq j} m_{ij} W_i W_j \right) 
            &\leq \E \exp \left(8\lambda \sum_{i,j=1}^d m_{ij} W_i W_j' \right) && (\textnormal{\Cref{thm:chaosdecoupled}}) \\
            &\leq 
             \E \exp \left(8 \lambda\sum_{i,j=1}^d m_{ij} G_i G_j' \right) && (\textnormal{\Cref{lem:comparisontogaussians}})\\ 
             &\leq
            \exp \left(64\lambda^2 \| M\|_F^2 \right) &&(\textnormal{\Cref{lem:mgfgaussianchaos}})
        \end{aligned}
    \end{equation}
    as long as $|\lambda|\leq 1/(8\sqrt{2}\opnorm{ M})$. The result follows by combining \eqref{eq:HWmgfdiag} and \eqref{eq:HWmgfoffidag} with \eqref{eq:CSinHWused} and then finally rescaling $\lambda$. We also note that the non-centered result follows analogously by skipping the step \eqref{eq:HWCS_ic}.
\end{proof}

The Hanson-Wright Inequality is usually stated in high probability form as in \Cref{thm:HWineqhighproba}. We finish the proof of this result below.

\paragraph{Finishing the proof of \Cref{thm:HWineqhighproba}}
We employ the the Chernoff trick as in \eqref{eq:chernofftrick} combined with the MGF bound of \Cref{prop:HWexpineq}. For $\lambda \in \R$ with $|\lambda|\leq \frac{1}{8\sqrt{2}\sigma^2\opnorm{ M} }$ we have that:
\begin{equation}\label{eq:chernoffHWineq}
    \begin{aligned}
          \Pr\left(W^\T M W -\E W^\T M W  > s \right)\leq \exp \left(-\lambda s +36\lambda^2\sigma^4 \| M\|_F^2 \right).
    \end{aligned}
\end{equation}
An admissible choice of $\lambda$ is:
\begin{equation}\label{eq:lambdachoseninhw}
    \lambda = \begin{cases}\frac{s}{72\sigma^4 \|M\|_F^2 },  & \textnormal{if}\quad \frac{s}{72\sigma^4 \|M\|_F^2 } \leq \frac{1}{8\sqrt{2}\sigma^2\opnorm{ M}}, \\
    \frac{1}{8\sqrt{2}\sigma^2\opnorm{ M}}, &\textnormal{if}\quad  \frac{s}{72\sigma^4 \|M\|_F^2 } > \frac{1}{8\sqrt{2}\sigma^2\opnorm{ M}}.
    \end{cases}
\end{equation}
Note that the second condition in \eqref{eq:lambdachoseninhw} can be rewritten as $\frac{72\sigma^2 \|M\|_F^2 }{8\sqrt{2} \opnorm{M}}< s $. Hence with the choice \eqref{eq:lambdachoseninhw} inserted into \eqref{eq:chernoffHWineq} we have:
\begin{equation*}
    \Pr\left(W^\T M W -\E W^\T M W  > s \right)
    \leq \exp \left( - \min \left( \frac{s^2}{144 \sigma^4 \| M\|_F^2} ,\frac{s}{16\sqrt{2}\sigma^2 \opnorm{M} } \right)\right).
\end{equation*}
The result follows by applying the same calcuation to $-M$ and using a union bound. \hfill $\blacksquare$

%% file: sections/twosided.tex
\section{Proof of \Cref{thm:spectrum deviations}}
\label{sec:twosided}



To prove Theorem \ref{thm:spectrum deviations}, we will need a couple of ingredients. We will require a variant of Hanson-Wright inequality, namely \Cref{thm:ARV} which is simple consequence of \Cref{thm:HWineqhighproba}. In \Cref{lem2}, we present the main concentration result. Finally, we will also require a result on approximate isometries which is presented in \Cref{lem:approximate_isometry}. The proof of \Cref{thm:spectrum deviations} is concluded in \Cref{sec:B:finishing-proof}.

\medskip 

Before we proceed, let us introduce some definitions that will ease notations in the derivations to come. We define 
\begin{align*}
    X = \begin{bmatrix}
        X_1^\T \\
        X_2^\T \\
        \vdots \\
        X_T^\T 
    \end{bmatrix} \in \mathbb{R}^{T \times d_x}, \qquad \text{and} \qquad \bfL = \matr{ 
    I_{\dx} & 0 & \cdots & 0\\ 
    A^\star & I_{\dx} &\cdots& 0 \\ 
    \vdots & & \ddots& \\
    (A^\star)^{T-1} & (A^\star)^{T-2} &\cdots& I_{\dx}
    } \in \mathbb{R}^{T\dx \times T\dx}. 
\end{align*}
Observe that $X^\T X = \sum_{t=1}^T X_t X_t^\top$. Moreover, if, with a slight abuse of notation, we think of $X_{1:T}$ and $W_{0:T-1}$ as long vectors of dimension $T\dx$, then we can write $X_{1:T} = \bfL W_{0:T-1}$

\subsection{A Variant of the Hanson-Wright Inequality}

In our analysis, we use the following concentration result on chaos, 
originally due to \cite{hanson1971bound}.

\begin{theorem}\label{thm:ARV}

Let $R\in \R^{r\times d}$. Fix a random variable $W=W_{1:d}$ where each $W_i, i\in [d]$ is a scalar, mean zero and independent $\sigma^2$-sub-Gaussian random vector. We further assume that $W$ is isotropic. Then for every $\varepsilon \in [0,\infty)$:
    \begin{align*}
        \Pr\left( \left\vert \Vert RW \Vert_2^2 - \Vert R \Vert_F^2  \right\vert  > \varepsilon \Vert R  \Vert_F^2  \right) \leq 2 \exp \left( - \min \left( \frac{\varepsilon^2}{144 \sigma^4  } ,\frac{\varepsilon}{16\sqrt{2}\sigma^2} \right) \frac{\Vert R \Vert_F^2}{\opnorm{R}^2} \right).
    \end{align*}
\end{theorem}
\begin{proof} First, we observe that  $\Vert R W \Vert_{2}^2 - \Vert R \Vert_F^2 =  (RW)^\T (RW) -\E  (RW)^\T (RW)$, where we used the isotropy of the noise vector $W$. Next, by using Hanson-Wright inequality, which we presented in Theorem  \ref{thm:HWineqhighproba}, we obtain that for all $\rho > 0$,
    \begin{align*}
        \Pr\left( \left\vert \Vert RW \Vert_2^2 - \Vert R \Vert_F^2  \right\vert  > \rho \right)
        & \leq 2 \exp \left( - \min \left( \frac{\rho^2}{144 \sigma^4 \| R
    ^\T R \|_F^2} ,\frac{\rho}{16\sqrt{2}\sigma^2 \opnorm{R^\T R} } \right)\right). \\
    & \leq 2 \exp \left( - \min \left( \frac{\rho^2}{144 \sigma^4 \opnorm{R}^2 \| R
 \|_F^2 } ,\frac{\rho}{16\sqrt{2}\sigma^2 \opnorm{R}^2 } \right)\right),
    \end{align*}
    where in the last inequality we used the elementary fact that $\Vert R^\T R \Vert_F \le \opnorm{R^\T} \Vert R \Vert_F$ and $\opnorm{R^\T R} = \opnorm{R}^2$.  Let $\varepsilon > 0$. If we choose $\rho = \varepsilon \Vert R \Vert_F^2$, we finally obtain 
    \begin{align*}
        \Pr\left( \left\vert \Vert RW \Vert_2^2 - \Vert R \Vert_F^2  \right\vert  > \varepsilon \Vert R  \Vert_F^2  \right) \le 2 \exp \left( - \min \left( \frac{\varepsilon^2}{144 \sigma^4  } ,\frac{\varepsilon}{16\sqrt{2}\sigma^2} \right) \frac{\Vert R \Vert_F^2}{\opnorm{R}^2} \right)
    \end{align*}
    as per requirement.
\end{proof}

\subsection{The Main Concentration Inequality}





Below we present \Cref{lem2}, which constitutes the main concentration result upon which we build to prove \Cref{thm:spectrum deviations}. This lemma is established by expressing $\opnorm{ (XM)^\T XM - I_{\dx}}$ as the suprema of a {\it chaos process}, which we can control combining Hanson-Wright inequality (\Cref{thm:ARV}) and a classical $\epsilon$-net argument (\Cref{lem:net argument 2}). 

\begin{lemma}\label{lem2}
Let $M=\left(\sum_{t=1}^{T} \sum_{k=0}^{t-1} (A^\star)^k (A^{\star,\T})^k\right)^{-\frac{1}{2}}$. Then, for any $\varepsilon \in [0, +\infty)$, we have 
$$
\Pr\left(\opnorm{(XM)^\T XM - I_{\dx} } > \max(\varepsilon, \varepsilon^2) \right) \le \exp \left(-  \frac{\varepsilon^2}{ 576 \, K^2 \opnorm{M}^2 \opnorm{\bfL}^2 }  + \dx \log(18)\right)
$$
\end{lemma}

\begin{proof} We split the proof into three steps.

\medskip
    
\noindent \textbf{Step 1.} In this step, we express $\Vert (XM)^\T XM - I_{\dx} \Vert$ as the supremum of a chaos process. We start by remarking that $W_{0:T-1}$ is a sequence of isotropic random vector. This ensures that 
$$
 \E X^\T X = \E \sum_{t=1}^{T} X_t X_t^\T = \sum_{t=1}^{T} \sum_{k=0}^{t-1} (A^\star)^k (A^{\star,\T})^k. 
$$
Observe that $ \E X^\T X$ is a symmetric positive definite matrix. We introduce the inverse of its square root matrix as $M$, i.e., 
$$
M = \left(\E X^\T X\right)^{-1/2}  = \left( \sum_{t=1}^{T} \sum_{k=0}^{t-1} (A^\star)^k (A^{\star,\T})^k  \right)^{-1/2}.
$$
 We think of multiplying $M$ by $X$ as a proper normalization of $X$ across all directions. Now, we have:
\begin{align}
    \opnorm{  (XM)^\T XM - I_{\dx} }  & =  \sup_{v\in S^{\dx-1}} \left \vert v^\T \left( (XM)^\T X M  -   I_{\dx} \right) v \right\vert    \label{eq:suprima1}  \\
    & = \sup_{v\in \mathbb{S}^{\dx-1}} \left \vert v^\T (XM)^\T XM v -  \E v^\T (XM)^\T XM v \right\vert \label{eq:suprima2} \\
    & = \sup_{v\in \mathbb{S}^{\dx-1}} \left \vert \Vert X M v \Vert_2^2 -   \E \Vert X M v \Vert_2^2 \right\vert. \nonumber \\
     & = \sup_{v\in \mathbb{S}^{\dx-1}} \left \vert \Vert  {\Lambda_{Mv}}^\T \bfL W_{0:T-1} \Vert_2^2 - \E \Vert  {\Lambda_{Mv}}^\T \bfL W_{0:T-1}  \Vert_2^2  \right\vert  \label{eq:suprima4} \\
   & = \sup_{v\in \mathbb{S}^{\dx-1}} \left \vert \Vert  {\Lambda_{Mv}}^\T \bfL W_{0:T-1} \Vert_2^2 - \Vert  {\Lambda_{Mv}}^\T \bfL \Vert_F^2  \right\vert, \label{eq:suprima5}
\end{align}
where, in equation \eqref{eq:suprima1}  we used the variational form of the operator norm for symmetric matrices, in equation \eqref{eq:suprima2} we used the fact that $\E X^\T X = M^{-2}$, and in equation \eqref{eq:suprima4}, the $T\dx\times T$ matrix $\Lambda_{Mv}$ is defined as
$$
\Lambda_{Mv} =
\begin{bmatrix}
Mv &   &        & O \\
  & Mv &        &  \\
  &    & \ddots &  \\
 O &   &        & Mv
\end{bmatrix}.
$$
Finally, equation \eqref{eq:suprima5} follows from the fact that $W_{0:T-1}$ is an isotropic vector. Additionally, we note from the definition of $M$, that $\Vert  {\Lambda_{Mv}}^\T \bfL \Vert_F^2 = 1$ for all $v \in \mathbb{S}^{\dx-1}$. In conclusion, we have proved in this step that $\opnorm{ (XM)^\T XM - I_{\dx} }$ may be expressed as the suprema of a \emph{chaos} process $(W_{0:T-1}^\T  \; C  \; W_{0:T-1})_C$, where the parametrizing matrix $C$ is
$$
C =  (\Lambda_{Mv}^\T \bfL)^\T({\Lambda_{Mv}}^\T \bfL).
$$
To control the suprema, we will use a covering  argument, but before that we require a pointwise bound over the parametrizing matrices $C$. 

\medskip
\noindent
\textbf{Step 2.} In this step, we establish a uniform concentration bound on the chaos terms $\left \vert \Vert  {\Lambda_{Mv}}^\T \bfL W_{0:T-1} \Vert_2^2 - \Vert  {\Lambda_{Mv}}^\T \bfL \Vert_F^2  \right\vert$ for all $v \in \mathbb{S}^{\dx-1}$. Let $v \in \mathbb{S}^{\dx-1}$. We recall that $W_{0:T-1}$ is a vector of independent, zero-mean, $K^2$-sub-Gaussian random variables. Therefore we may apply \Cref{thm:ARV}, and obtain: for all $\rho> 0$, the following event, 
$$
\left \vert \Vert \Lambda_{Mv}^\T \bfL W_{0:T-1} \Vert_2^2 - \Vert \Lambda_{Mv}^\T \bfL \Vert_F^2  \right\vert > \rho \Vert \Lambda_{Mv}^\T \bfL \Vert_F^2 
$$
holds with probability at most 
$$
2 \exp \left(-  \min\left(\frac{\rho^2}{144 K^4}, \frac{\rho}{16\sqrt{2} K^2}\right) \frac{\Vert  {\Lambda_{Mv}}^\T \bfL \Vert_F^2}{ \opnorm{  {\Lambda_{Mv}}^\T \bfL }^2 } \right)
$$
Now, note that $\Vert  {\Lambda_{Mv}}^\T \bfL \Vert_F^2 = 1$, and  $\opnorm{ {\Lambda_{Mv}}^\T \bfL } \le \opnorm{M} \opnorm{\bfL} $. Moreover, isometry of the noise ensures that $2K^2 \ge 1$, thus $144K^4 \ge 16\sqrt{2}K^2$. Therefore, starting from the above probability upper bound,  we obtain  
\begin{align}\label{proof-B:ineq-1}
    \Pr \left( \left \vert \Vert \Lambda_{Mv}^\T \bfL W_{0:T-1} \Vert_2^2 - \Vert \Lambda_{Mv}^\T \bfL \Vert_F^2  \right\vert > \rho   \right)   \le 2 \exp \left(-  \frac{\min(\rho^2, \rho)}{144 K^4 \opnorm{M}^2 \opnorm{\bfL }^2 }  \right).
\end{align}

\medskip
\noindent
\textbf{Step 3.} To conclude the proof, we use an $\epsilon$-net argument, namely \Cref{lem:net argument 2}. Recalling the equations \eqref{eq:suprima1} and \eqref{eq:suprima4} in step 1, as a consequence of equation \eqref{proof-B:ineq-1}, we have: for all $\epsilon\in [0, 1/2)$, for all $v \in \mathbb{S}^{\dx-1}$, the following
\begin{align*}
 \Pr\left(\left \vert v^\T \left((XM)^\T X M - I_{\dx}\right) v \right\vert > (1-2\epsilon)\rho  \right) \le 2 \exp \left(-  \frac{\min((1-2\epsilon)^2\rho^2, (1-2\epsilon)\rho)}{144 K^4 \opnorm{M}^2 \opnorm{\bfL }^2 }  \right). 
\end{align*}
Now choosing $\epsilon = 1/4$ and applying \Cref{lem:net argument 2}, leads to 
\begin{align*}
\Pr\left(
\opnorm{ (XM)^\T XM - I_{\dx} } > \rho \right) & \le 2 \cdot 9^{\dx} \exp \left(-  \frac{\min((\rho/2)^2, (\rho/2))}{144 K^4 \opnorm{M}^2 \opnorm{\bfL }^2 }  \right)  \\
& \le  18^{\dx} \exp \left(-  \frac{\min(\rho^2, \rho)}{576 K^4 \opnorm{M}^2 \opnorm{\bfL }^2 }  \right)
\end{align*}
Choosing $\rho = \max(\varepsilon, \varepsilon^2)$, which is equivalent to $\varepsilon^2 = \min(\rho^2, \rho)$, we obtain that
$$
\Pr\left(\opnorm{(XM)^\T XM - I_{\dx} } > \max(\varepsilon, \varepsilon^2) \right) \le \exp \left(-  \frac{\varepsilon^2}{ 576 \, K^2 \opnorm{M}^2 \opnorm{\bfL}^2 }  + \dx \log(18)\right)
$$
This completes the proof of Lemma \ref{lem2}.

\end{proof}

\subsection{Approximate Isometries}

\begin{lemma}[Approximate isometries]\label{lem:approximate_isometry}
Let $X\in \mathbb{R}^{T\times d}$, and let $M\in \mathbb{R}^{d\times d}$ be a full rank matrix. Let $\varepsilon>0$ and assume that
\begin{equation}\label{eq:approximate isometry}
    \opnorm{(XM)^\T XM - I_d} \le \max(\varepsilon, \varepsilon^2).
\end{equation}
Then the following holds
$$
\frac{(1 - \varepsilon)}{s_1(M)} \le s_d(X) \le \dots \le s_1(X) \le  \frac{(1 + \varepsilon) }{s_d(M)}.
$$
\end{lemma}

\begin{proof}
Observe that
  $$
  \opnorm{ (XM)^\top XM - I_d} = \sup_{v \in {\mathbb{S}^{d-1}}} \left\vert \Vert XMv \Vert_2^2 - 1  \right\vert.
  $$
  Using the fact that $\vert z^2 - 1 \vert \ge \max(\vert z-1 \vert, \vert z-1 \vert^2)$ for all $z > 0$, we obtain for all $v \in \mathbb{S}^{d-1}$
  \begin{align*}
    \max(\left\vert \Vert XMv \Vert_2 - 1  \right\vert, \left\vert \Vert XMv \Vert_2 - 1  \right\vert^2) \le \max(\varepsilon, \varepsilon^2).
  \end{align*}
By the monotonicity of $z \mapsto \max(z, z^2)$ for $z \ge  0$, we deduce that:
\begin{align*}
    \left \vert \Vert X M v \Vert_2 - 1 \right\vert \le \varepsilon,
  \end{align*}
  from which we conclude
  \begin{align*}
    1 - \varepsilon \le \Vert XM v \Vert_2 \le 1 + \varepsilon.
  \end{align*}
  It follows immediately that:
  \begin{align*}
    1-\varepsilon  \le s_d(XM) \le \dots \le s_1(XM)\le 1+\varepsilon.
  \end{align*}
  The proof is completed observing that:
  $$
  s_d(X M) \le s_d(X) s_1(M) \quad \textrm{and} \quad   s_1(X) s_d(M)\le s_1(XM).
  $$
  \end{proof}

\subsection{Finishing the proof of \Cref{thm:spectrum deviations}}\label{sec:B:finishing-proof}

To prove \Cref{thm:spectrum deviations}, we first invoke \Cref{lem2}, and obtain 
\begin{align}\label{eq:deviation:fact1}
    \Pr\left(\Vert (XM)^\T XM - I_{\dx} \Vert \le  \max(\varepsilon, \varepsilon^2) \right) \ge 1- \exp \left(-  \frac{\varepsilon^2}{ 576 \, K^2 \opnorm{M}^2 \opnorm{\bfL}^2 }  + \dx \log(18)\right)
\end{align}
Next, we use \Cref{lem:approximate_isometry}, which entails that 
\begin{align}\label{eq:deviation:fact2}
    \Pr\left(   \frac{(1-\e)}{s_1(M)} \le s_{\dx}(X) \le \dots \le s_{1}(X) \le \frac{(1+\e)}{s_{\dx}(M)} \right) \ge \Pr\left(\Vert (XM)^\T XM - I_{\dx} \Vert \le  \max(\varepsilon, \varepsilon^2) \right)
\end{align}
Recalling the definition of $M$, and combining the inequalities \eqref{eq:deviation:fact1} and \eqref{eq:deviation:fact2} yields the desired result.

%% file: sections/anticoncproofs.tex
\section{Proofs Relating to the Lower Tail of the Empirical Covariance}\label{sec:anticoncproofs}

\paragraph{Proof of \Cref{lem:subG_decoup}}
First, we remark that we may prove the lemma under the additional hypothesis that $Q_{22}\succ0$ without loss of generality by regrouping terms. We now proceed to prove the lemma under this additional hypothesis. 

Let us introduce the new variable $\mu = Q_{22}^{-1/2} Q_{12}x $. Rewriting our quadratic form in terms of this new variable yields:
\begin{equation}\label{eq:subglem:cov}
    \begin{aligned}
     &\E \exp \left( -\lambda \begin{bmatrix}x \\ W
    \end{bmatrix}^\T  \begin{bmatrix}Q_{11}& Q_{12}\\ Q_{21} & Q_{22} \end{bmatrix} \begin{bmatrix}x \\ W
    \end{bmatrix}\right) \\
    & =
      \E \exp \left( -\lambda x^\T Q_{11} x- \lambda W^\T Q_{22}W +2\lambda x^\T Q_{12}^\T W \right)\\
    &=
       \E \exp \left( -\lambda x^\T   Q_{11} x- \lambda W^\T Q_{22}W +2\lambda \mu^\T Q_{22}^{1/2} W -\lambda \mu^\T \mu +\lambda x^\T Q_{12}^\T Q_{22}^{-1} Q_{12} x  \right)\\
    &= 
      \E \exp \left( -\lambda x^\T   (Q_{11} - Q_{12}^\T Q_{22}^{-1} Q_{12}   )x- \lambda W^\T Q_{22}W +2\lambda \mu^\T Q_{22}^{1/2} W -\lambda \mu^\T \mu \right)\\
    &\leq
     \E \exp \left( -\lambda \begin{bmatrix}\mu \\ W
    \end{bmatrix}^\T  \begin{bmatrix}I_n& Q_{22}^{1/2}\\ Q_{22}^{1/2} & Q_{22} \end{bmatrix} \begin{bmatrix}\mu \\ W
    \end{bmatrix}\right).
    \end{aligned}
\end{equation}
where the last inequality uses the fact that $(Q_{11} - Q_{12}^\T Q_{22}^{-1} Q_{12}   )$ is the Schur complement of $Q_{22}$ in the positive semidefinite matrix $Q$ and hence positive semidefinite.

To finish the proof, we note that
\begin{equation}
    \begin{aligned}
       & \E \exp \left( -\lambda \begin{bmatrix}\mu \\ W
    \end{bmatrix}^\T  \begin{bmatrix}I_n& Q_{22}^{1/2}\\ Q_{22}^{1/2} & Q_{22} \end{bmatrix} \begin{bmatrix}\mu \\ W
    \end{bmatrix}\right)\\
    &=
    \E \exp\left( -2\lambda \mu^\T \sqrt{Q_{22} }W  -\lambda \mu^\T\mu    -\lambda W^\T Q_{22} W \right)\\
    &\leq
    \sqrt{\E \exp\left( -2\lambda \mu^\T \sqrt{Q_{22} }W -\lambda \mu^\T\mu \right)} \sqrt{ \E \exp \left(    -\lambda W^\T Q_{22} W \right)} && (\textnormal{Cauchy-Schwarz})\\
    & \leq
    \sqrt{\E \exp\left( 4  \lambda^2  K^2 \mu^\T Q_{22} \mu  -\lambda \mu^\T\mu \right)} \sqrt{ \E \exp \left(    -\lambda W^\T Q_{22} W \right)}. &&(\textnormal{sub-Gaussianity})
    \end{aligned}
\end{equation}
The result follows by noting that $\sqrt{\E \exp\left( 4  \lambda^2  K^2 \mu^\T Q_{22} \mu  -\lambda \mu^\T\mu \right)} \leq 1$ for our range of $\lambda$. \hfill $\blacksquare$

\paragraph{Proof of \Cref{thm:expineq}}
By repeated use of the tower property we have that:
\begin{equation}\label{eq:towerprop}
\begin{aligned}
      &\E \exp \left(-\lambda \sum_{t=0}^{T-1}\|\Delta X_t\|_2^2 \right) =   \E\exp \left(-\lambda \sum_{t=0}^{k-1}\|\Delta X_t\|_2^2\right) \times 
      \dots \times \E_{T-k-1} \exp \left(-\lambda \sum_{t=T-k}^{T-1}\|\Delta X_t\|_2^2 \right).
\end{aligned}
\end{equation}
We will bound each conditional expectation in \eqref{eq:towerprop} separately, starting with the outermost. Observe that
\begin{align*}
    &\sum_{t=T-k}^{T-1}  \|\Delta X_t\|_2^2  =  \begin{bmatrix}
    \Delta X_{T-k}\\
    \vdots\\
    \Delta X_{T-1}
    \end{bmatrix}^\T \begin{bmatrix}
    \Delta X_{T-k}\\
    \vdots\\
    \Delta X_{T-1}
    \end{bmatrix}
    %
    =  
    W_{0:T-1}^\T\mathbf{L}_{T/k}^\T \mathrm{blkdiag}(\Delta^\T \Delta) \mathbf{L}_{T/k} W_{0:T-1}
\end{align*}
We now apply \Cref{prop:decoup} conditionally with $x=W_{0:T-k-1}$. We find that:
\begin{multline*}
   \E_{T-k-1} \exp \left(-\lambda \sum_{t=T-k}^{T-1}\|\Delta X_t\|_2^2\right) 
   \\
    \leq \exp \Bigg( -\lambda\tr\left[ \mathbf{L}_{T/k,T/k}^\T \mathrm{blkdiag}(\Delta^\T \Delta) \mathbf{L}_{T/k,T/k}\right] 
    + 36K^4\lambda^2\left[\tr \mathbf{L}_{T/k,T/k}^\T \mathrm{blkdiag}(\Delta^\T \Delta) \mathbf{L}_{T/k,T/k}\right]^2 \Bigg).
\end{multline*}
Repeatedly applying \Cref{prop:decoup} as above yields the result. \hfill $\blacksquare$

\paragraph{Proof of \Cref{lem:pointwiseanticonc}}

For any fixed $\Delta \in \R^{d'\times d} \setminus \{0\}$ and $\lambda \geq 0$ to be determined below in \eqref{eq:simplechernoff} we have that:
\begin{equation}
\label{eq:simplechernoff}
\begin{aligned}
&\mathbf{P} \left( \sum_{t=1}^{T}  \|\Delta X_t\|^2 \leq \frac{1}{2} \sum_{t=1}^{T} \E \|\Delta\tilde  X_t\|^2 \right)\\
&\leq  \E \exp \left(  \frac{\lambda}{2} \sum_{t=1}^{T} \E \|\Delta\tilde X_t\|^2 -\lambda \sum_{t=1}^{T} \|\Delta X_t\|^2   \right) && (\textnormal{Chernoff})\\
&\leq  \exp \Bigg( -\frac{\lambda}{2}\sum_{j=1}^{T/k} \tr\left[ \mathbf{L}_{j,j}^\T \mathrm{blkdiag}(\Delta^\T \Delta) \mathbf{L}_{j,j}\right] \\
&+ 36\lambda^2K^4 \sum_{j=1}^{T/k} \tr\left[ \mathbf{L}_{j,j}^\T \mathrm{blkdiag}(\Delta^\T \Delta) \mathbf{L}_{j,j}\right]^2 \Bigg) && (\textnormal{\Cref{thm:expineq}})\\
 &=\exp \Bigg( -\frac{\lambda T }{2 k} \tr\left[ \mathbf{L}_{1,1}^\T \mathrm{blkdiag}(\Delta^\T \Delta) \mathbf{L}_{1,1}\right] \\
&+ \frac{36\lambda^2 TK^4}{ k}  \tr\left[ \mathbf{L}_{1,1}^\T \mathrm{blkdiag}(\Delta^\T \Delta) \mathbf{L}_{1,1}\right]^2 \Bigg) && (\mathbf{L}_{j,j}= \mathbf{L}_{1,1})\\
 &=\exp \Bigg( -\frac{\lambda T }{2 k} \tr\left[ \mathbf{L}_{1,1}^\T \mathrm{blkdiag}(\Delta^\T \Delta) \mathbf{L}_{1,1}\right] \\
&+ \frac{36\lambda^2 TK^4}{ k}  \left[\tr \mathbf{L}_{1,1}^\T \mathrm{blkdiag}(\Delta^\T \Delta) \mathbf{L}_{1,1}\right]^2 \Bigg)  &&(\textnormal{Cauchy-Schwarz})
 \\
&\leq \exp \left(- \frac{T}{576K^2k} \right) && \left( \lambda = \frac{ \tr\left[ \mathbf{L}_{1,1}^\T \mathrm{blkdiag}(\Delta^\T \Delta) \mathbf{L}_{1,1}\right] }{144K^2 \left[ \tr\mathbf{L}_{1,1}^\T \mathrm{blkdiag}(\Delta^\T \Delta) \mathbf{L}_{1,1}\right]^2 } \right)
\end{aligned}
\end{equation}
by optimizing $\lambda$ in the last line. We point out that the $\lambda$ used in the above calculation is admissible since $ \left[ \tr\mathbf{L}_{1,1}^\T \mathrm{blkdiag}(\Delta^\T \Delta) \mathbf{L}_{1,1}\right]^2 \geq \opnorm{\mathbf{L}_{1,1}^\T \mathrm{blkdiag}(\Delta^\T \Delta) \mathbf{L}_{1,1}}\tr\left[ \mathbf{L}_{1,1}^\T \mathrm{blkdiag}(\Delta^\T \Delta) \mathbf{L}_{1,1}\right]$ and so can be seen to satisfy the constraints of \Cref{thm:expineq}.
\hfill $\blacksquare$

\paragraph{Proof of \Cref{thm:anticonc}}

Let $\mathcal{N}_\e$ be an optimal $\e$-cover of the unit sphere $\mathbb{S}^{d-1}$ and fix a multiplier $q\in (1,\infty)$. We define the events (i.e. $\Delta = v^\T$):
\begin{equation}
\begin{aligned}
    \mathcal{E}_1 &= \bigcup_{v\in \mathcal{N}_e} \left\{ \frac{1}{T} \sum_{t=0}^{T-1} v^\T X_tX_t^\T v \leq \frac{1}{2T} \sum_{t=0}^{T-1} \E v^\T \tilde X_t\tilde X_t^\T v  \right\}
    \\
  \mathcal{E}_2&  =\Bigg\{\left\| \sum_{t=0}^{T-1}X_tX_t^\T\right\|_{\mathsf{op}}
    \geq q \times  \left\| \sum_{t=0}^{T-1}\E X_tX_t^\T\right\|_{\mathsf{op}} \Bigg\}.
\end{aligned}
\end{equation}
for any $v$, it is true on the complement of $\mathcal{E}= \mathcal{E}_1 \cup \mathcal{E}_2$ that for every $v_i \in \mathcal{N}_\e$:
\begin{equation}
    \begin{aligned}
       &\frac{1}{T} \sum_{t=0}^{T-1} v^\T X_tX_t^\T v \\
       &\geq \frac{1}{2T} \sum_{t=0}^{T-1} v^\T_i X_tX_t^\T v_i - \frac{1}{T}\sum_{t=0}^{T-1} (v-v_i)^\T X_tX_t^\T (v-v_i)  &&(\textnormal{parallellogram}) \\
        &\geq \frac{1}{2T} \sum_{t=0}^{T-1} v^\T_i X_tX_t^\T v_i - \frac{\e^2}{T} \left\|\sum_{t=0}^{T-1} X_tX_t^\T \right\|_{\mathsf{op}}  &&(\textnormal{covering}) \\
       &\geq \frac{1}{4T}\sum_{t=0}^{T-1} \E  v_i^\T \tilde X_t \tilde X_t^\T v_i-\frac{q\e^2}{T}\left\|\sum_{t=0}^{T-1} \E [X_tX_t^\T]  \right\|_{\mathsf{op}} &&(\mathcal{E}^c)
    \end{aligned}
\end{equation}
where we used that $v-v_i$ has norm at most $\e$ for some choice of $v_i$ by the covering property. For this choice we have that:
\begin{align*}
      \frac{1}{T} \sum_{t=0}^{T-1} v^\T X_tX_t^\T v
       &\geq \frac{1}{8T}\sum_{t=0}^{T-1}  v_i^\T \E [\tilde X_t \tilde X_T^\T] v_i
\end{align*}
as long as:
\begin{align}\label{eq:epsilonconstraint}
    \e^2 \leq \frac{ \lambda_{\min} \left(\sum_{t=0}^{T-1} \E \tilde X_t \tilde X_t^\T \right)}{8q  \lambda_{\max} \left(\sum_{t=0}^{T-1} \E X_t  X_t\right) }.
\end{align}
To finish the proof, it suffices to estimate the failure probabilities $\mathbf{P}(\mathcal{E}_1)$ and $\mathbf{P}(\mathcal{E}_2)$. By \eqref{eq:simplechernoff}, a volumetric argument \citep[see e.g.][Example 5.8]{wainwright2019high} and our particular choice of $\e$ we have:
\begin{align*}
    \mathbf{P}(\mathcal{E}_1)& \leq \left(1+\frac{2}{\e^2} \right)^d \exp \left(- \frac{T}{576K^2k} \right).
\end{align*}

To estimate $\mathbf{P}(\mathcal{E}_2)$, observe first that for $\mathbb{S}^{d-1}_{1/4}$ a $1/4$-net of $\mathbb{S}^{d-1}$, we have by \Cref{lem:opnormdisc}, \eqref{eq:opnormdisc2}:
\begin{equation}\label{eq:epsnetusedinthm}
    \left\|\sum_{t=0}^{T-1} X_tX_t^\T \right\|_{\mathsf{op}} \leq 2 \sup_{\tilde v\in \mathbb{S}^{d-1}_{1/4}} \sum_{t=0}^{T-1} \tilde v^\T X_tX_t^\T \tilde v
\end{equation}
We now invoke \Cref{prop:HWexpineq} with $M=\mathbf{L}_{\tilde v} \mathbf{L}_{\tilde v} ^\T$ where $\mathbf{L}_{\tilde v} = ( I_{T}\otimes {\tilde v}^\T )\mathbf{L}$ for fixed $ \tilde v$ of the form $(v-v_i)/\e$ as before and finish by a union bound. For fixed $\tilde v$ and $\lambda \geq 0$ to be determined below, \Cref{prop:HWexpineq}  yields via a Chernoff argument:
\begin{equation}\label{hwusedinthm}
    \begin{aligned}
        &\mathbf{P}\left(  \sum_{t=0}^{T-1} \tilde v^\T X_tX_t^\T \tilde v 
    \geq q \times \sum_{t=0}^{T-1} \tilde v^\T  \E [X_tX_t^\T] \tilde v\right)\\
    &=\mathbf{P}\left(  W_{0:T-1}^\T \mathbf{L}_{\tilde v} \mathbf{L}_{\tilde v}^\T W_{0:T-1}^\T 
    \geq q \tr  \mathbf{L}_{\tilde v} \mathbf{L}_{\tilde v} ^\T\right)\\
    &\leq  \exp\left(- \lambda q \tr  \mathbf{L}_{\tilde v} \mathbf{L}_{\tilde v}  + 36\lambda^2 K^4 \tr  (\mathbf{L}_{\tilde v} \mathbf{L}_{\tilde v}^\T)^2  \right)
    &&
    \left(\lambda \leq \frac{1}{8\sqrt{2}K^2 \opnorm{\mathbf{L}_{\tilde v} \mathbf{L}_{\tilde v} ^\T}}\right)
    \\
    &\leq  \exp\left(- \lambda q \tr  \mathbf{L}_{\tilde v} \mathbf{L}_{\tilde v}  + 36\lambda^2 K^4 \opnorm{\mathbf{L}_{\tilde v} \mathbf{L}_{\tilde v} ^\T}   \tr  (\mathbf{L}_{\tilde v} \mathbf{L}_{\tilde v}^\T)  \right) \\
    &=\exp \left( - \frac{\tr  (\mathbf{L}_{\tilde v} \mathbf{L}_{\tilde v}^\T)}{K^2\opnorm{\mathbf{L}_{\tilde v} \mathbf{L}_{\tilde v} ^\T} }  \left( \frac{q}{8\sqrt{2}}-\frac{9}{32} \right)\right) && \left(\lambda = \frac{1}{8\sqrt{2}K^2 \opnorm{\mathbf{L}_{\tilde v} \mathbf{L}_{\tilde v} ^\T}}\right)\\
    &\leq\exp \left(- \frac{T}{576K^2k} \right)  && \left( q\geq\frac{8\sqrt{2} T \opnorm{\mathbf{L}_{\tilde v} \mathbf{L}_{\tilde v} ^\T}}{576 k \tr  (\mathbf{L}_{\tilde v} \mathbf{L}_{\tilde v}^\T)} + \frac{9 \times 8\sqrt{2}}{32} \right).
    \end{aligned}
\end{equation}
Observe that 
\begin{equation}\label{eq:qischosen}
q=\frac{8\sqrt{2} T \opnorm{\mathbf{L} \mathbf{L} ^\T}}{576 k  \lambda_{\min}\left(\sum_{t=0}^{T-1} \E X_t X_t^\T \right) } + \frac{9 \times 8\sqrt{2}}{32} 
\end{equation}
satisfies the constraint from \eqref{hwusedinthm} for all $\tilde v$. Hence by a union bound and  a volumetric argument \citep[see e.g.][Example 5.8]{wainwright2019high} with probability at least $1-\left( 1+\frac{2}{\e^2}\right)^d$ (observing that $\e \leq 1/4$ in \eqref{eq:epsilonconstraint}):
\begin{equation}
    \begin{aligned}
        \left\|\sum_{t=0}^{T-1} X_tX_t^\T \right\|_{\mathsf{op}} &\leq 2 \sup_{\tilde v\in \mathbb{S}^{d-1}_{1/4}} \sum_{t=0}^{T-1} \tilde v^\T X_tX_t^\T \tilde v
        && (\textnormal{by } \eqref{eq:epsnetusedinthm})
        \\
        &
        \leq 
2q \sup_{\tilde v\in \mathbb{S}^{d-1}_{1/4}} \sum_{t=0}^{T-1} \tilde v^\T \E X_tX_t^\T \tilde v
        && (\textnormal{union bound over } \eqref{hwusedinthm})
        \\
        &
        \leq 2q \left\|\sum_{t=0}^{T-1} \E X_tX_t^\T \right\|_{\mathsf{op}}. && ( \mathbb{S}^{d-1}_{1/4} \subset  \mathbb{S}^{d-1})
    \end{aligned}
\end{equation}

Hence with the choice of $q$ from \eqref{eq:qischosen} and another union bound (again observing that $\e \leq 1/4$):
\begin{equation}
     \mathbf{P}(\mathcal{E}_1)+ \mathbf{P}(\mathcal{E}_2) \leq 2  \left(1+\frac{2}{\e^2} \right)^d \exp \left(- \frac{T}{576K^2k} \right).
\end{equation}
In light of \eqref{eq:epsilonconstraint} we may choose with the above choice of $q$ (from \eqref{eq:qischosen}):
\begin{equation}
    \begin{aligned}
        \e^2 &=   \frac{ \lambda_{\min} \left(\sum_{t=0}^{T-1} \E \tilde X_t \tilde X_t^\T \right)}{8q  \lambda_{\max} \left(\sum_{t=0}^{T-1} \E X_t  X_t\right) } \\
       &=\frac{ \lambda_{\min} \left(\sum_{t=0}^{T-1} \E \tilde X_t \tilde X_t^\T \right)}{8 \left(\frac{8\sqrt{2} T \opnorm{\mathbf{L} \mathbf{L} ^\T}}{576 k  \lambda_{\min}\left(\sum_{t=0}^{T-1} \E X_t X_t^\T \right) } + \frac{9 \times 8\sqrt{2}}{32}  \right) \lambda_{\max} \left(\sum_{t=0}^{T-1} \E X_t  X_t\right) }.
    \end{aligned}
\end{equation}
Thus:
\begin{equation}
\begin{aligned}
    \left( 1+\frac{2}{\e^2}\right)^d &=   \left( 1+16 \frac{\left(\frac{8\sqrt{2} T \opnorm{\mathbf{L} \mathbf{L} ^\T}}{576 k  \lambda_{\min}\left(\sum_{t=0}^{T-1} \E X_t X_t^\T \right) } + \frac{9 \times 8\sqrt{2}}{32}  \right) \lambda_{\max} \left(\sum_{t=0}^{T-1} \E X_t  X_t\right) }{ \lambda_{\min} \left(\sum_{t=0}^{T-1} \E \tilde X_t \tilde X_t^\T \right)} \right)^d    \\
    &=
    \left( 1+4\sqrt{2} \frac{\left(\frac{ T \opnorm{\mathbf{L} \mathbf{L} ^\T}}{18 k  \lambda_{\min}\left(\sum_{t=0}^{T-1} \E X_t X_t^\T \right) } + 9   \right) \lambda_{\max} \left(\sum_{t=0}^{T-1} \E X_t  X_t\right) }{ \lambda_{\min} \left(\sum_{t=0}^{T-1} \E \tilde X_t \tilde X_t^\T \right)} \right)^d 
\end{aligned}
\end{equation}
The result has been established.
\hfill $\blacksquare$

%% file: sections/selfnormproofs.tex
\section{Proof of the Self-Normalized Martingale Theorem}

\paragraph{Proof of \Cref{lem: method of mixtures bound}}
    Let $\Psi$ be independent from $P$ and $Q$ with  a matrix normal distribution with density 
    \begin{align*}
        p(\Lambda) = \frac{1}{\sqrt{(2\pi)^{\dx\da} \det(\Sigma^{-1})^{\da}}} \exp\paren{-\frac{1}{2} \tr\paren{\Lambda^\T \Sigma \Lambda}}.
    \end{align*}
    Then the conditional expectation on the right of \eqref{eq: self norm pseudo max} becomes
    \begin{align*}
        &\E\paren{\exp\tr\paren{-\frac{1}{2}(\Psi-Q^{-1} P^\T)^\T Q (\Psi - Q^{-1} P^\T)} \vert P, Q}\\
        &\overset{(i)}{=}\int_{\R^{\dx \times \da}}  \frac{1}{\sqrt{(2\pi)^{\dx\da} \det(\Sigma^{-1})^{\da}}} \exp\paren{-\frac{1}{2}
        \tr\paren{\Lambda^\T \Sigma \Lambda}} \exp\paren{-\frac{1}{2} \tr\paren{(\Lambda-Q^{-1} P^\T)^\T Q (\Lambda - Q^{-1} P^\T)}} d \Lambda \\
       &\overset{(ii)}{=} \frac{\sqrt{\det (\Sigma)^{\da}}}{\sqrt{\det (Q+\Sigma)^{\da}}} \int_{\R^{\dx \times \da}} \frac{\exp\paren{-\frac{1}{2}\tr\paren{(\Lambda-Q^{-1} P)^\T Q (\Lambda - Q^{-1} P) + \Lambda^\T \Sigma \Lambda}}}{\sqrt{(2 \pi)^{\dx \da} \det((Q+\Sigma)^{-1})^{\da}}} d \Lambda,
    \end{align*}
    where $(i)$ follows from the definition of conditional expectation, and $(ii)$ follows by pulling the term $\sqrt{\det (\Sigma)^{\da}}/\sqrt{\det (Q+\Sigma)^{\da}}$ out of the integral. 
    Completing the square, we can massage the exponent of the integrand above:
    \begin{align*}
        &(\Lambda-Q^{-1} P^\T)^\T Q (\Lambda - Q^{-1} P^\T) + \Lambda^\T \Sigma \Lambda \\
        &= (\Lambda - (Q+\Sigma)^{-1} P^\T)^\T(Q+\Sigma)(\Lambda-(Q+\Sigma)^{-1} P^\T) - P (Q+\Sigma)^{-1} P^\T  + P Q^{-1} P^\T.
    \end{align*}
    Then by pulling the exponential terms that do not depend on $\Lambda$ outside of the integral, the integral in the expression for the conditional expectation above simplifies to the integral of a matrix normal density. The conditional expectation therefore simplifies to
    \begin{align*}
        &\E\paren{\exp\paren{-\frac{1}{2}(\Lambda-Q^{-1} P^\T)^\T Q (\Lambda - Q^{-1} P^\T)} \vert P,Q} \\
        &= \frac{\sqrt{\det (\Sigma)^{\da}}}{\sqrt{\det (Q+\Sigma)^{\da}}} \exp\paren{\frac{1}{2}\tr\paren{P (Q+\Sigma)^{-1}P^\T}} \exp\paren{-\frac{1}{2}\tr\paren{P Q^{-1}P^\T}}.
    \end{align*}
    Substituting this result into \eqref{eq: self norm pseudo max}, and leveraging the canonical assumption as in \eqref{eq: relaxation of canoncial assumption}, we have that 
    \begin{align*}
        \E\paren{\frac{\sqrt{\det (\Sigma)^{\da}}}{\sqrt{\det (Q+\Sigma)^{\da}}} \exp\paren{\frac{1}{2}\tr\paren{P(Q+\Sigma)^{-1}P^\T }}} \leq 1.
    \end{align*}
    By the Chernoff trick and Markov's inequality,
    \begin{align*}
        & \bfP\paren{\norm{P(Q+\Sigma)^{-1/2}}_F^2 > 2\log\paren{\frac{\det(Q + \Sigma)^{\da/2} \det(\Sigma)^{-\da/2}}{\delta}}} \\
        &= \bfP\paren{ \exp\paren{ \frac{1}{2}\norm{P(Q+\Sigma)^{-1/2}}_F^2} \frac{\sqrt{\det(\Sigma)^{\da}}}{\sqrt{\det(Q+\Sigma)^{\da}}} > \frac{1}{\delta}} \\
        &\leq \delta \E \paren{\exp\paren{ \frac{1}{2}\norm{P(Q+\Sigma)^{-1/2}}_F^2} \frac{\sqrt{\det(\Sigma)^{\da}}}{\sqrt{\det(Q+\Sigma)^{\da}}} } \leq \delta .
    \end{align*}
\hfill $\blacksquare$

\paragraph{Proof of \Cref{lem: canoncial assumption verified}}
 
    Define $D_t(\Lambda) \triangleq \exp\tr\paren{\frac{\eta_t X_t^\T \Lambda}{\sigma}  - \frac{1}{2} \Lambda^\T X_t X_t^\T \Lambda}$. Observe that
    \begin{align*}
        \E M_t(\Lambda)&\overset{(i)}{=} \E \E \paren{M_t(\Lambda) \vert \calF_{t}} \\
        &\overset{(ii)}{=} \E\paren{ D_1(\Lambda) D_1(\Lambda) \dots D_{t-1}(\Lambda) \E\paren{D_t(\Lambda) \vert \calF_{t}}} \\
        &\overset{(iii)}{=} \E \paren{M_{t-1}(\Lambda) \E\paren{D_t(\Lambda) \vert \calF_{t}}},
    \end{align*}
    where $(i)$ follows by the tower rule, while $(ii)$ and $(iii)$ follow by writing the exponential of a sum as a product of exponentials. By the assumption that the process $\eta_t$ is conditionally $\sigma^2$-sub-Guassian with respect to $\calF_t$, we have that
    \begin{align*}
        \E\paren{D_t(\Lambda) \vert \calF_t} &\overset{(i)}{=} \E\paren{\exp\tr\paren{\frac{\eta_t X_t^\T \Lambda}{\sigma}}  \vert \calF_t}  \exp\tr\paren{-\frac{1}{2} \Lambda^\T X_t X_t^\T \Lambda} \\
        &\overset{(ii)}{=} \E\paren{\exp\paren{\frac{\langle \Lambda^\T X_t,  \eta_t\rangle}{\sigma}}  \vert \calF_t}  \exp\tr\paren{-\frac{1}{2} \Lambda^\T X_t X_t^\T \Lambda} \\
        &\overset{(iii)}{\leq} \exp\paren{\frac{1}{2} \norm{\Lambda^\T X_t}^2} \exp\tr\paren{-\frac{1}{2} \Lambda^\T X_t X_t^\T \Lambda} = 1,
    \end{align*}
    where $(i)$ follows from the fact that $X_t$ is $\calF_t$-measureable, $(ii)$ follows from the trace-cyclic property, and $(iii)$ follows from the fact that $\eta_t$ is conditionally sub-Gaussian.
    Then for all $t > 0$, $\E M_t(\Lambda) \leq \E M_{t-1}(\Lambda)$, and $\E M_0(\Lambda) =1$. This in turn implies that $\E M_t(\Lambda) \leq 1$ for all $0 \leq t \leq T$. \hfill $\blacksquare$

\subsection{Proof of \Cref{thm:self-normalized martingale}}
    The Frobenius norm bound follows immediately by using \Cref{lem: canoncial assumption verified} to verify that \eqref{eq: canoncial assumption} holds for 
    \begin{align*}
        P = \sum_{t=1}^T \frac{V_t X_t^\T}{\sigma} , \quad Q= \sum_{t=1}^T X_t X_t^\T,
    \end{align*}
    followed by application of \Cref{lem: method of mixtures bound}. 

    For the operator norm bound, we employ a covering argument. Let $\calN_\e = \curly{w_i}_{i=1}^{5^{\dy}}$ be an $\e$-net of the $\dy$-dimensional unit sphere with $\e=0.5$. Such a net exists by virtue of \Cref{lem:volumetric}. For any $i \in [5^{\dy}]$, we may bound 
    \begin{align*}
        \bignorm{w_i^\T \sum_{t=1}^T V_t X_t^\T \paren{\Sigma + \sum_{t=1}^T X_t X_t^\T}^{-1/2}}
    \end{align*}
    by using \Cref{lem: canoncial assumption verified} to verify that \eqref{eq: canoncial assumption} holds for 
    \begin{align*}
        P = \sum_{t=1}^T \frac{w_i^\T V_t X_t^\T}{\sigma} , \quad Q = \sum_{t=1}^T X_t X_t^\T,
    \end{align*}
    and then applying \Cref{lem: method of mixtures bound}. To move from this bound to the operator norm bound, note that by \Cref{lem:net argument 1}
    \begin{align*}
       & \bigopnorm{\left(\sum_{t=1}^T V_t X_t^\T \right)\left(\Sigma + \sum_{t=1}^T X_t X_t^\T \right)^{-1/2 }}^2 = \sup_{w \in \R^{\dy},  \norm{w} \leq 1} \bignorm{w^\T \left(\sum_{t=1}^T V_t X_t^\T \right)\left(\Sigma+ \sum_{t=1}^T X_t X_t^\T \right)^{-1/2 }}^2 \\
        &\leq  (1 - \e)^{-2} \sup_{w_i \in \calN_{\e}} \bignorm{w^\T \left(\sum_{t=1}^T V_t X_t^\T \right)\left(\Sigma+ \sum_{t=1}^T X_t X_t^\T \right)^{-1/2 }}^2.
    \end{align*}
   A union bound over the $5^{\dy}$ elements of $\calN_\e$ tells us that with probability at least $1- 5^{\dy} \delta'$,
    \begin{align*}
        \sup_{w_i \in \calN_\e} \bignorm{w_i^\T \left(\sum_{t=1}^T V_t X_t^\T \right)\left(\Sigma + \sum_{t=1}^T X_t X_t^\T \right)^{-1/2 }}^2 \leq 2 \sigma^2 \log\paren{\frac{\det\paren{\Sigma+ \sum_{t=1}^T X_t X_t^\T }^{1/2} }{\det(\Sigma)^{ 1/2}\delta'}}.
    \end{align*}
     The claim follows by combining results along with the change of variables $\delta = 5^{\dy} \delta'$.  \hfill $\blacksquare$

%% file: sections/proof_sysid.tex
\section{Proofs for System Identification}\label{appendix:sys_id}
\subsection{Proof of Theorem~\ref{thm:arx_pac}}
Select $\Sigma=\frac{T}{16}\CovX_\tau$ and consider the following events
\begin{align*}
\calE_1&\triangleq \set{\ECov_T\succeq \frac{1}{16}\CovX_\tau}\\
\calE_2&\triangleq \set{\ECov_T\preceq 3\frac{p\dy+q\du}{\delta}\CovX_T}\\
\calE_3&\triangleq \left\{   
       \bigopnorm{\sum_{t=1}^T W_t X_t^\top \paren{\Sigma + T\ECov_T}^{-1/2}}^2 \right.\\
        &\left.\qquad \leq 4 K^2 \log\paren{\frac{\det\paren{\Sigma + T\ECov_T}}{ \det(\Sigma)}}+8\dy K^2\log 5+ 8K^2\log\frac{3}{\delta}\right\}.
\end{align*}
It follows that $\Prob(\calE^{c}_i)\le\delta/3$, for $i=1,\,2,\,3$ by invoking Theorem~\ref{thm:arx_pe}, Lemma~\ref{lem:arx_Markov}, and Theorem~\ref{thm:self-normalized martingale} respectively. By a union bound $\Prob(\calE^c_1\cup \calE^c_2\cup \calE^c_3)\le \delta$.
Hence $\Prob(\calE_1\cap \calE_2\cap \calE_3)\ge 1-\delta$.

It remains to show~\eqref{eq:arx_pac} holds under the intersection of the aforementioned three events. Based on the error decomposition
\[
\opnorm{\Mls_T-\Ms}^2\le \opnorm{\Sigmaw}\bigopnorm{\sum_{t=1}^TW_tX^\top_t \paren{T\ECov_T}^{-1/2}}^2\,\opnorm{(T\ECov_T)^{-1}}
\]
The \textbf{excitation term} can be bounded immediately based on $\calE_1$
\[
\opnorm{(T\ECov_T)^{-1}}\le \frac{16}{T\lambda_{\min}(\CovX_\tau)} .
\]
For the \textbf{noise term}, we have to modify the inverse to resemble the one in the event $\calE_3$. Note that under the event $\calE_1$
\[
2T\ECov_T \succeq T\ECov_T+\Sigma.
\]
Hence the noise term is upper bounded by
\[
\bigopnorm{\sum_{t=1}^TW_tX^\top_t \paren{T\ECov_T}^{-1/2}}^2\le 2\bigopnorm{\sum_{t=1}^TW_tX^\top_t \paren{\Sigma+T\ECov_T}^{-1/2}}^2
\]
Putting everything together and upper bounding $\ECov_T$ based on $\calE_2$, we obtain
\begin{equation}
\opnorm{\Mls_T-\Ms}^2\le 16\frac{8}{ \snr_{\tau} T} \paren{
\log\det\paren{I + 48\frac{p\dy+q\du}{\delta}\CovX_T\CovX^{-1}_\tau} +2\dy \log 5+ 2 \log\frac{3}{\delta}}.
\end{equation}
Next, we consider the following
\[
\det\paren{I + 48\frac{p\dy+q\du}{\delta}\CovX_T\CovX^{-1}_\tau}\le \paren{49\frac{p\dy+q\du}{\delta}}^{p\dy+q\du}\det\paren{ \CovX_T\CovX^{-1}_\tau},
\]
where we used the fact that  $I\preceq  \frac{p\dy+q\du}{\delta}\CovX_T\CovX^{-1}_\tau$ (follows from $\Sigma_{\tau}\preceq \Sigma_{T}$) along with the properties of the determinant. To simplify the final expression, we absorb minor order terms into the higher order terms by inflating the constants accordingly.\hfill \qedsymbol


The following supporting lemma establishes monotonicity of the covariance sequence with respect to the positive definite cone. 
\begin{lemma}[Covariance monotonicity]\label{lem:arx_cov_monotonicity}
Let $X_t$ be defined as in~\eqref{eq:ARX_parameters_batch}, with $\Sigma_t$ its covariance. Under the assumption that $\Sigma_0=0$, the sequence $\Sigma_t$ is increasing with respect to the positive semidefinite cone
\[
\Sigma_{t+1}\succeq \Sigma_t,\,t\ge 0.
\]
\end{lemma}
\begin{proof}
    From the extended state space recursion~\eqref{eq:arx2ss}, we obtain that
    \[
\Sigma_{t+1}=\mathcal{A}\Sigma_t\mathcal{A}^{\top}+\mathcal{B}\E V_{t+1}V_{t+1}^\top\mathcal{B}^\top.
    \]
    Since the noise processes are i.i.d., we have constant covariance over time $\Gamma_V\triangleq \E V_{t+1}V_{t+1}^\top,$ for all $t\ge 0$.
  Trivially, we have $\Sigma_1\succeq 0=\Sigma_0$. We prove the other cases by induction. Assume that $\Sigma_{t}\succeq \Sigma_{t-1}$. Then
  \begin{align*}
\Sigma_{t+1}= &\mathcal{A}\Sigma_{t}\mathcal{A}^{\top}+\mathcal{B}\Gamma_V\mathcal{B}^\top \\&\succeq \mathcal{A}\Sigma_{t-1}\mathcal{A}^{\top}+\mathcal{B}\Gamma_V\mathcal{B}^\top=\Sigma_{t},
  \end{align*}
  where the inequality follows from $\Sigma_t\succeq \Sigma_{t-1}$.
\end{proof}

\subsection{Proof of Theorem~\ref{thm:arx_pe}}
Recall~\eqref{eq:LforARX}. Before we apply Theorem~\ref{thm:anticonc}, we need to pick a block size $k$ and prove that the minimum eigenvalue of the decoupled process
$$\tilde X_{1:T} = \mathrm{blkdiag}(\mathbf{L}_{11},\dots, \mathbf{L}_{T/k,T/k})V_{0:T-1}$$
is strictly positive. Pick $k=2\tau$ such that $\tau\ge \max\{p,q\}$. We will first show that
\begin{equation}\label{eq:decoupled_lower_bound}
\lambda_{\min}(\sum_{t=1}^{T}\E\tilde{X}_t\tilde{X}^\top_t) \ge \frac{T}{2}\lambda_{\min}\CovX_\tau>0.
\end{equation} Then, we will characterize the burn-in time by applying Theorem~\ref{thm:anticonc}.

\noindent\textbf{Part a). Lower bound for decoupled process. }
Notice that
\[
\sum_{t=1}^{T}\E\tilde{X}_t\tilde{X}^\top_t=\frac{T}{k}\sum_{t=1}^{k}\E\tilde{X}_t\tilde{X}^\top_t=\frac{T}{k}\sum_{t=1}^{k}\E X_t X^\top_t,
\]
where the equality follows from the fact that the blocks of the decoupled random process are identical to the first block, which coincides with the true process.
By semi-definiteness of the covariance and by monotonicity (Lemma~\ref{lem:arx_cov_monotonicity})
\[
\sum_{t=1}^{T}\E\tilde{X}_t\tilde{X}^\top_t=\frac{T}{k}\sum_{t=1}^{k}\E X_t X^\top_t\succeq \frac{T}{k}\sum_{t=\tau+1}^{k}\E X_t X^\top_t\succeq \frac{T}{2}\CovX_{\tau}.
\]
What remains to show is that $\CovX_{\tau}$ is non-degenerate. We will denote by $\star$ the elements that can be non-zero but do not matter.
Let $t\ge\max{p,q}$. Then, we have
\[
X_t=\matr{\calT_{p}&\star\\0&I_{q\du}}\matr{W_{t:t-p+1}\\U_{t:t-q+1}}+\matr{\star&\star\\\star&\star}\matr{W_{0:t-p}\\U_{0:t-q}},
\]
where $\calT_p=\matr{\calB_1&\cdots&\calA_{11}^{p-1}\calB_1}$ is an upper triangular Toeplitz matrix
\[
\calT_p=\matr{\Sigmaw^{1/2}&\star&\cdots&\star\\0&\Sigmaw^{1/2}&\cdots&\star\\ \vdots &&\ddots&\\0&0&\cdots&\Sigmaw^{1/2}}.
\]
Since the noise is non-degenerate, matrix $\calT_p$ is full rank. As a result, we obtain
\[
\CovX_{\tau}\succeq \matr{\calT_{p}&\star\\0&I_{q\du}}\matr{I_{p\dy}&0\\0&\sigma^2_u I_{q\du}}\matr{\calT_{p}&\star\\0&I_{q\du}}^\top\succ 0.
\]

\noindent\textbf{Part b). Burn-in time characterization.} Applying Theorem~\ref{thm:anticonc} and by~\eqref{eq:decoupled_lower_bound}, we obtain
\begin{equation*}
    \mathbf{P} \left(  \CovX_T \succeq  \frac{1}{16}  \E X_t X_t^\T \right)
   \ge 1-\delta 
 \end{equation*}
 if 
 \begin{equation}\label{eq:arx_burn_in_unprocessed}
T\ge 1152 K^2 \tau (\log 1/\delta+(p\dy+q\du)\log \Csys)
 \end{equation}
where
 \begin{equation*}
     C_{\mathsf{sys}}  \triangleq 1+2\sqrt{2} \frac{\left(\frac{ T \opnorm{\mathbf{L} \mathbf{L} ^\T}}{36\tau  \lambda_{\min}\left(\sum_{t=1}^{T} \E X_t X_t^\T \right) } + 9   \right) \lambda_{\max} \left(\sum_{t=1}^{T} \E X_t  X_t^\T\right) }{ \lambda_{\min} \left(\sum_{t=1}^{T} \E \tilde X_t \tilde X_t^\T \right)}
 \end{equation*}
 and $\mathbf{L}$ is defined in~\eqref{eq:LforARX}. What remains is to prove that~\eqref{eq:arx_burn_in} implies~\eqref{eq:arx_burn_in_unprocessed}. To do so we need to simplify the term $\Csys$.

\paragraph{Upper bounding $\opnorm{\bfL\bfL^\top}$}
Note that $\bfL=\matr{\bfL_1^\top &\cdots&\bfL^\top_T}^\top$, where
\[
{\bfL}_i=\matr{\mathcal{A}^{i-1}\mathcal{B}&\cdots&\mathcal{B}&0&\cdots&0}.
\]
Hence we can write
\begin{align}  
\opnorm{\bfL\bfL^\top}&=\opnorm{\bfL^\top\bfL}=\opnorm{\sum_{t=1}^T \bfL^\top_t \bfL_t}\\
&\le\sum_{t=1}^T \opnorm{\bfL_t \bfL^\top_t}=\sum_{t=1}^{T}\opnorm{\sum_{k=0}^{t-1} \mathcal{A}^k\mathcal{B}\mathcal{B}^\top(\mathcal{A}^\top)^k}\\
&=\sum_{t=1}^{T} \opnorm{\CovX_t}\le T\opnorm{\CovX_T}
\end{align}
\paragraph{Upper bounding $\lambda_{\max} \left(\sum_{t=1}^{T} \E X_t  X_t^\T\right) $}
We just have $\opnorm{\sum_t\CovX_t}\le T\opnorm{\CovX_T}$
\paragraph{Lower bounds.} To lower bound the smallest eigenvalues, we invoke~\eqref{eq:decoupled_lower_bound}.
\paragraph{Final bound on $\Csys$}
Combining the above we get
\begin{align*}
\Csys&= 1+4\sqrt{2} \frac{\left(\frac{ T \opnorm{\mathbf{L} \mathbf{L} ^\T}}{36 \tau  \lambda_{\min}\left(\sum_{t=1}^{T} \E X_t X_t^\T \right) } + 9   \right) \lambda_{\max} \left(\sum_{t=1}^{T} \E X_t  X_t\right) }{ \lambda_{\min} \left(\sum_{t=1}^{T} \E \tilde X_t \tilde X_t^\T \right)}\\
&\le 1+4\sqrt{2}\paren{\frac{T}{36\tau}\frac{\opnorm{\CovX_T}}{\lambda_{\min}(\CovX_{\tau})}+9}\frac{\opnorm{\CovX_T}}{\lambda_{\min}(\CovX_{\tau})}.
\end{align*}
Under~\eqref{eq:arx_burn_in}, we have $T\ge 1152 \tau$. Hence, $\frac{T}{36\tau}\frac{\opnorm{\CovX_T}}{\lambda_{\min}(\CovX_{\tau})}\ge 9$ which leads to the simplified expression
\[
\Csys\le 16\sqrt{2}/36 \frac{T}{\tau}\frac{\opnorm{\CovX_T}^2}{\lambda^2_{\min}(\CovX_{\tau})}.
\]
Using $16\sqrt{2}/36\le 2/3$, we finally get that
$
  \Csys\le\Csys(T,\tau).
  $
  Finally, based on the above inequality, condition~\eqref{eq:arx_burn_in} implies~\eqref{eq:arx_burn_in_unprocessed}.
  \hfill \qedsymbol

\subsection{Proof of Theorem~\ref{thm:non-degenerate_snr}}
It is sufficient to provide a uniform lower bound on the least singular value of $\CovZ_{p,p}$.
Expanding the state-space equation~\eqref{eq:ss_innovation_equation}, we obtain
\begin{align*}
Y_{t}=\sum_{s=1}^{t}(\Cs(\As)^{s-1} \Ks \Sigmae^{1/2} E_{t-s}+\Cs(\As)^{s-1} \Bs  U_{t-s})+\Sigmae^{1/2} E_{t}.
\end{align*}
Hence, the vector of covariates at time $p$ is equal to
\begin{equation}
    Z_{p}=\matr{\calT_{e,p}&\calT_{u,p}\\0& I}\matr{E_{p-1:0}\\U_{p-1:0}},
\end{equation}
where
$\calT_{e,p}$, $\calT_{u,p}$ are Toeplitz matrices
\[
\calT_{e,s}\triangleq \matr{\Sigmae^{1/2}&\Cs\Ks\Sigmae^{1/2}& &\Cs(\As)^{s-2}\Ks\Sigmae^{1/2}\\0&\Sigmae^{1/2}&\cdots&\Cs(\As)^{s-3}\Ks\Sigmae^{1/2}\\ \vdots&\vdots& &\vdots \\0&0&\cdots&\Sigmae^{1/2}}, 
\calT_{u,s}\triangleq \matr{0&\Cs\Bs& &\Cs(\As)^{s-2}\Bs\\0&0&\cdots&\Cs(\As)^{s-3}\Bs\\ \vdots&\vdots& &\vdots \\0&0&\cdots&0}. 
\]
The covariance at time $p$ can be now expressed compactly as
\[
\CovZ_{p,p}=\matr{\calT_{e,p}\calT^\top_{e,p}+\sigma^2_u\calT_{u,p}\calT^\top_{u,p}&\sigma_u\calT_{u,p}\\\sigma_u\calT^\top_{u,p}&\sigma^2_u I}.
\]
We need to show that the covariance is uniformly lower bounded.
By Lemmas~\ref{lem:positivity_of_SE},~\ref{lem:dominance_of_TE}
\[
\calT_{e,p}\calT^\top_{e,p}\succeq c_1 I+c_2 \calT_{u,p}\calT^\top_{u,p}
\]
for some positive $c_1,c_2$.
Using the upper-bound and splitting carefully
\begin{align*}
\CovZ_{p,p}&\succeq \matr{c_1 I &0\\0&c_2/(\sigma^2_u+c_2)I}+\matr{(\sigma^2_u+c_2)\calT_{u,p}\calT^\top_{u,p}&\sigma_u\calT_{u,p}\\\sigma_u\calT^\top_{u,p}&\sigma^2_u/(\sigma^2_u+c_2)I}\\
&\succeq \matr{c_1 I &0\\0&c_2/(\sigma^2_u+c_2)I}+ \matr{\sqrt{\sigma^2_u+c_2}\calT_{u,p}\\\sigma_u/\sqrt{\sigma^2_u+c_2}I}\matr{\sqrt{\sigma^2_u+c_2}\calT_{u,p}\\\sigma_u/\sqrt{\sigma^2_u+c_2}I}^\top\\
&\succeq \matr{c_1 I &0\\0&c_2/(\sigma^2_u+c_2)I}.
\end{align*}
The lower bound on $\CovZ_{p,p}$ is uniform and independent of $p,T$. \hfill $\blacksquare$

\begin{lemma}\label{lem:positivity_of_SE}
Under the assumptions of Theorem~\ref{thm:non-degenerate_snr}, the Toeplitz matrix $\calT_{e,t}$ is lower bounded by 
\[
\calT_{e,t}\calT^\top_{e,t}\succeq \lambda_{\min}(\Sigmav) I
\]
\end{lemma}
\begin{proof}
The proof is identical to the one of Lemma~B.2 in~\cite{tsiamis2019finite}.
\end{proof}

\begin{lemma}\label{lem:dominance_of_TE}
Under the assumptions of Theorem~\ref{thm:non-degenerate_snr}, there exists a constant $c$ such that
\[
\calT_{e,t}\calT^\top_{e,t}\succeq c \calT_{u,t}\calT^\top_{u,t}.
\]
\end{lemma}
\begin{proof}
    Without loss of generality assume that $t$ is a multiple of the state dimension $\dx$. The Toeplitz matrices are invariant under similarity transformations. Consider the similarity transformation $S^{-1}\As S=\matr{A_1&0\\0&A_2}$, where $A_1$ includes the modes on the unit circle and $A_2$ includes the modes inside the unit circle. Define $\matr{C_1&C_2}=\Cs S$, $F=S^{-1}\Ks$, and $B=S^{-1}\Bs$.
    
    Let us decompose the matrices into non-explosive and stable parts respectively $\calT_{e,t}=\calT_{e,t,1}+\calT_{e,t,2}$,  $\calT_{u,t}=\calT_{u,t,1}+\calT_{u,t,2}$. By Lemma~\ref{lem:positivity_of_SE} and since matrix $\calT_{u,t,2}$ depends on the stable part, we have
    \begin{equation}\label{eq:dominance_of_TE_helper1}
    \calT_{e,t}\calT_{e,t}^\top \succeq c_2\calT_{u,t,2}\calT^\top_{u,t,2},
    \end{equation}
    for $c_2=\frac{\lambda_{\min}(\Sigmav)}{\sup_t\opnorm{T_{u,t,2}}^2}$, which is non-zero.
    
    The main difficulty is bounding the non-explosive part. 
    Define the controllability matrix
    \begin{equation}\label{eq:controllability_matrix}
    \mathcal{C}(A_1,G)\triangleq \matr{G&\cdots&A^{\dx-1}_1 G}
    \end{equation}
    along with the observability matrix
    \[
    \mathcal{O}_j\triangleq \matr{C_1A^{j\dx-1}_1\\\vdots\\C_1A^{(j-1)\dx}_1}
    \]
    Observe that
    \[
    \calT_{u,t,1}=\blkdiag(\calT_{u,\dx,1})+\underbrace{\matr{0&\mathcal{O}_{1}&\mathcal{O}_{2}&\cdots&\mathcal{O}_{t/\dx}\\0&0&\mathcal{O}_{1}&\cdots&\mathcal{O}_{t/\dx-1}\\ \vdots&&&\ddots\\0&0&0&\cdots&\mathcal{O}_{1}\\0&0&0&\cdots&0}\blkdiag(\mathcal{C}(A_1,B_1))}_{\calT_{u,t,3}}
    \]
    and
     \[
    \calT_{e,t,1}=\blkdiag(\calT_{e,\dx,1})+\underbrace{\matr{0&\mathcal{O}_{1}&\mathcal{O}_{2}&\cdots&\mathcal{O}_{t/\dx}\\0&0&\mathcal{O}_{1}&\cdots&\mathcal{O}_{t/\dx-1}\\ \vdots&&&\ddots\\0&0&0&\cdots&\mathcal{O}_{1}\\0&0&0&\cdots&0}\blkdiag(\mathcal{C}(A_1,F_1))}_{\calT_{e,t,3}}.
    \]
    By Lemma~\ref{lem:controllability_dominance} we get
    \begin{equation}\label{eq:dominance_of_TE_helper2}
    \calT_{e,t,3}\calT^\top_{e,t,3} \succeq c_3 \calT_{u,t,3}\calT^\top_{u,t,3},
    \end{equation}
    for some $c_3>0$. 
    Since $\opnorm{\blkdiag(\calT_{u,\dx,1})}=\opnorm{\calT_{u,\dx,1}}$ does not increase with $t$ we have
    \begin{equation}\label{eq:dominance_of_TE_helper3}
    \calT_{e,t}\calT_{e,t}^\top \succeq c_1\blkdiag(\calT_{u,\dx,1})\blkdiag(\calT_{u,\dx,1})^\top,
    \end{equation}
    for $c_1=\frac{\lambda_{\min}(\Sigmav)}{\opnorm{\calT_{u,\dx,1}}^2}$
    
    Exploiting the inequality $(A+B)(A+B)^\top\preceq 2AA^\top+2BB^\top$
    \[
    \calT_{u,t}\calT^\top_{u,t}\preceq 2\calT_{u,t,2}\calT^\top_{u,t,2}+4\blkdiag(\calT_{u,\dx,1})\blkdiag(\calT_{u,\dx,1})^\top+4 \calT_{u,t,3}\calT^\top_{u,t,3}.   
    \]
    From~\eqref{eq:dominance_of_TE_helper1},~\eqref{eq:dominance_of_TE_helper2},~\eqref{eq:dominance_of_TE_helper3}
    \[
    \calT_{u,t}\calT^\top_{u,t}\preceq (2c_2+4c_1) \calT_{e,t}\calT_{e,t}^\top+4c_3 \calT_{e,t,3}\calT_{e,t,3}^\top.
    \]
To finalize the proof notice that
\[
\calT_{e,t,3}\calT_{e,t,3}^\top\preceq 2\calT_{e,t}\calT_{e,t}^\top+4\calT_{e,t,2}\calT_{e,t,2}^\top+4\blkdiag(\calT_{e,\dx,1})\blkdiag(\calT_{e,\dx,1})^\top.
\]
    Repeating the same arguments as in~\eqref{eq:dominance_of_TE_helper1},~\eqref{eq:dominance_of_TE_helper3} gives
    \[
\calT_{e,t,2}\calT_{e,t,2}^\top \preceq c_4 \calT_{e,t}\calT_{e,t}^\top,\,\blkdiag(\calT_{e,\dx,1})\blkdiag(\calT_{e,\dx,1})^\top\preceq c_5 \calT_{e,t}\calT_{e,t}^\top
    \]
    for some positive $c_4$, $c_5$, which completes the proof. 
\end{proof}

\begin{lemma}\label{lem:controllability_dominance}
    Consider the assumptions of Theorem~\ref{thm:non-degenerate_snr} and recall the definition of the controllability matrix in~\eqref{eq:controllability_matrix} with $A_1$ having all eigenvalues on the unit circle. There exists  a positive constant $c$ such that
    \[
\mathcal{C}(A_1,F_1)\mathcal{C}^{\top}(A_1,F_1)\succeq c\, \mathcal{C}(A_1,B_1)\mathcal{C}^{\top}(A_1,B_1).
    \]
\end{lemma}
\begin{proof}
It is sufficient to prove that $\mathcal{C}(A_1,F_1)$ has full rank or equivalently the pair $(A_1,F_1)$ is controllable. Assume that it is not. Then, by the Popov-Belevitch-Hautus test, there exists a $v_1$ such that $v_1^* F_1=0$ with $v_1^* A_1=\lambda_1 v_1^*$, $|\lambda_1|=1$. The same will hold for the original system $(\As,\Ks)$ and
$
v\triangleq S^{-1}\matr{v^\top_1&0}^\top.
$
Multiplying~\eqref{eq:ss_Ric} from the left and right by $v^*,v$
\[
v^* \Ps v=|\lambda_1|^2 v^* \Ps v+v^* \Sigmaw v=v^* \Ps v+v^* \Sigmaw v,
\]
which is possible only if $v^* \Sigmaw=0$, violating detectability of the pair $(\Sigmaw^{1/2},\As)$. As a result, the pair $(A_1,F_1)$ is controllable.
\end{proof}

\subsection{Proof of Theorem~\ref{thm:ss_pac}}
Define the empirical covariance as
\begin{equation}\label{eq:ss_empirical_covariance}
    \ECovZ_{p,T}=1/T\sum_{t=1}^{T}Z_tZ^\top_t.
\end{equation}
From~\eqref{eq:ss_arx_approximation}, it follows that the least-squares error is equal to
\[
\Mls_{p,T}-\Ms_p=\underbrace{\sum_{t=1}^T\Sigmae^{1/2}E_tZ^\top_t(T\ECovZ_{p,T})^\dagger}_{\text{statistical error}}+\underbrace{\Cs(\Acs)^p\sum_{t=1}^TX_{t-p}Z^\top_t(T\ECovZ_{p,T})^\dagger}_{\text{bias}}.
\]
The proof of bounding the statistical error term is identical to the one of Theorem~\ref{thm:arx_pac}. 
We get that if $T\ge \Tburn(\delta,\beta)$, then with probability at least $1-\delta$
\[
\opnorm{\text{statistical error}}^2\le \frac{C}{ \snr_{p,p} T} \paren{
p(\dy+\du)\log\frac{p(\dy+\du)}{\delta}+\log\det\paren{\CovZ_{p,T}\CovZ^{-1}_{p,p}}},
\]
for some universal constant $C$ along with 
$
\ECov_{p,T}\succeq \frac{1}{16}\CovZ_{p,p}.
$
Next, we upper-bound the bias term. Note that $T\ECovZ_{p,T}\succeq Z_tZ_t$ and, thus,
\[
\opnorm{Z^\top_t (T\ECovZ_{p,T})^{-1} Z_t}\le 1.
\]
By the triangle inequality, Cauchy-Schwarz, and the inequality above
\begin{align*}
    \opnorm{\sum_{t=1}^{T}X_{t-p}Z^\top_t (T\ECovZ_{p,T})^{-1/2}}\le \sum_{t=p}^{T}\opnorm{X_{t-p}}\opnorm{Z^\top_t (T\ECovZ)^{-1/2}}\le \sqrt{\sum_{t=1}^{T-p}\opnorm{X_{t}}^2} \sqrt{T}.
\end{align*}
By the above inequality, Lemma~\ref{lem:ss_bias_term_upper_bound} below, and $\ECov_{p,T}\succeq \frac{1}{16}\CovZ_{p,p}$, we finally obtain 
\[
\opnorm{\text{bias error}}^2\le \frac{C'\dx}{ \snr_{p,p}}\log \frac{1}{\delta} \Big(\opnorm{\Cs(\Acs)^p} T \opnorm{\CovXState_{X,T}}\Big) , 
\]
for some universal constant $C'$.
Finally, under~\eqref{eq:ss_choice_of_p}
 the bias error becomes
\[
\opnorm{\text{bias error}}^2\le \frac{C'\dx}{ T^2\snr_{p,p}}\log \frac{1}{\delta}, 
\]
and is dominated by the statistical error.

What remains to show is that there exists a $\beta$ such that~\eqref{eq:ss_choice_of_p} is satisfied. Since $\Acs$ is asymptotically stable, there exist $M$, $\lambda<1$ such that $\opnorm{\Cs(\Acs)^p}\le M \lambda^p$. It is sufficient to satisfy
\[
M \lambda^{\beta \log T} \opnorm{\CovXState_{X,T}} \le T^{-3}.
\]
After some algebraic manipulations, we arrive at
\[
\beta \log(1/\lambda) \ge \frac{\log \opnorm{\CovXState_{X,T}}}{\log T}+3+M/\log T.
\]
Since the system is non-explosive, the ratio
\[
\sup_T \frac{\log \opnorm{\CovXState_{X,T}}}{\log T}<\infty
\]
is uniformly bounded. Hence, there exists such a $\beta$.\hfill $\blacksquare$
\begin{lemma}\label{lem:ss_bias_term_upper_bound}
    Fix a failure probability $\delta>0$. There exists a universal constant $c$ such that with probability at least $1-\delta$:
    \begin{equation}
    \sum_{t=1}^{T-p}\opnorm{X_{t}}^2\le c K^2 \dx T  \opnorm{\CovXState_{X,T}}\log 1/\delta.    
    \end{equation}
\end{lemma}
\begin{proof}
The result follows from a naive application of Hanson's Wright inequality (Theorem~\ref{thm:HWineqhighproba}).
 Let $V_t\triangleq \matr{E^{\top}_t& U^{\top}_t/{\sigma_u}}^\top$. 
 Unrolling the state-space equations we get
 \[
 X_t=\sum_{k=1}^{t}(\As)^{k-1}(\Bs U_{t-k}+\Ks \Sigmae E_{t-k}).
 \]
Hence, there exists a lower triangular matrix $\bfL_x=\matr{\bfL^\top_{x,1}&\cdots&\bfL^\top_{x,T-p}}^\top$ such that $X_{t}=\bfL_{x,t} V_{0:T-p-1}$, for all $t\le T-p$. Under the new notation, the sum is equal to
\[
\sum_{t=1}^{T-p}\opnorm{X_{t-p}}^2=V_{0:T-p-1}^\top \bfL_x^\top \bfL_x V_{0:T-p-1}. 
\]
Moreover, by definition
\[
\CovXState_{X,t}=\bfL_{x,t}\bfL^\top_{x,t}.
\]
Define $M=\bfL^\top \bfL$ and note that
\[\opnorm{M}\le \sum_{t=1}^{T-p}\opnorm{\bfL^\top_{x,t}\bfL_{x,t}}=\sum_{t=1}^{T-p}\opnorm{\bfL_{x,t}\bfL^\top_{x,t}}\le T \opnorm{\CovXState_{X,T}}\]
where the second inequality follows from monotonicity of $\CovXState_{X,t}$ (see~Lemma~\ref{lem:arx_cov_monotonicity}).
We also have 
\[\norm{M}_F\le \sum_{t=1}^{T-p}\norm{\bfL^\top_{x,t}\bfL_{x,t}}_F\le \sqrt{\dx}\sum_{t=1}^{T-p}\opnorm{\bfL^\top_{x,t}\bfL_{x,t}}\le \sqrt{\dx}T\opnorm{\CovXState_{X,T}},\]
where we used the identity $\norm{A}_F\le \textsf{rank}(A)\opnorm{A}$.
By Theorem~\ref{thm:HWineqhighproba}
\[
\Prob\Big(\sum_{t=1}^{T-p}\opnorm{X_{t-p}}^2\ge T\dx\opnorm{\CovXState_{X,T}}+s\Big)\le 2  \exp \left( - \min \left( \frac{s^2}{144 K^4 \| M\|_F^2} ,\frac{s}{16\sqrt{2}K^2 \opnorm{M} } \right)\right).
\]
In view of the above inequalities, choosing
\[
s=16\sqrt{2}\dx K^2 T\opnorm{\CovXState_{X,T}}\log(2/\delta)
\]
gives failure probability less than $\delta$.
\end{proof}

%% file: sections/proofsec67.tex
\section{Proofs for \Cref{sec:basicineq} and \Cref{sec:nonlinear}}

\subsection{Proofs for Sparse Identification}

\label{subsec:proofsforbasic}

\paragraph{Proof of \Cref{lem:sparseisomorph}}

We invoke \Cref{thm:anticonc} and apply it to every fixed subspace $S$ and then union bound over these subspaces. In fact, notice that the proof of \Cref{thm:anticonc} actually yields the stronger statement that \eqref{eq:loosetwosidedsparse} holds for any fixed $S$ with probability
\begin{equation}
    1-C_{\mathsf{sys}}^{2s} \exp \left( \frac{-T}{576K^2 k}\right)
\end{equation}
The lower bound is just the statement of \Cref{thm:anticonc} whereas the upper bound is easily seen to be part of the proof of \Cref{thm:anticonc} by combining \eqref{hwusedinthm} with the ensuing union bound and \Cref{lem:opnormdisc}.

We now union bound over $p \choose 2s$ subspaces $S$. Hence with probability at least 
\begin{equation*}
    1- {p \choose 2s} C_{\mathsf{sys}}^{2s} \exp \left( \frac{-T}{576K^2 k}\right)
\end{equation*}
the desired result holds. Use now the approximation ${p \choose 2s} \leq \left(\frac{e p}{2s} \right)^{2s}$ and solve for $\delta$ to obtain the result.\hfill $\blacksquare$

\paragraph{Proof of \Cref{prop:sparse}}
We work exclusively on the event of \Cref{lem:sparseisomorph}. The lower bound in \eqref{eq:loosetwosidedsparse} combined with \eqref{eq:sparsebasicineq} yields that
\begin{equation}
\label{eq:sparsebasicineqused}
\|(\widehat \theta -\theta^\star )\sqrt{\Sigma_k}\|_2^2 \leq 4T^{-1} \max_S  \left\|\left(\sum_{t=1}^{T} W_t (X_t)_S \right)\left(\sum_{t=1}^{T} (X_t)_S (X_t)_S^\T \right)^{-1/2 } \right\|^2_2
\end{equation}
on said event.

We now analyze the right hand side of \eqref{eq:sparsebasicineqused} for fixed $S$. On the event prescribed by \Cref{lem:sparseisomorph} With probability $1-\delta$ We have
\begin{multline}\label{eq:sparseprepforselfnorm}
     \left\|\left(\sum_{t=1}^{T} W_t (X_t)_S \right)\left(\sum_{t=1}^{T} (X_t)_S (X_t)_S^\T \right)^{-1/2 } \right\|^2_2
     \\
     \leq 2 \left\|\left(\sum_{t=1}^{T} W_t (X_t)_S \right)\left(\frac{T}{16k} \sum_{t=1}^{k}\E  \left[(X_t)_S (X_t)_S^\T\right]+\sum_{t=1}^{T} (X_t)_S (X_t)_S^\T \right)^{-1/2 } \right\|^2_2.
\end{multline}

For some universal positive constant $c\in \R$, still on the event of \Cref{lem:sparseisomorph}, \Cref{thm:self-normalized martingale} yields that the right hand side of \eqref{eq:sparseprepforselfnorm} is upper bound (pointwise in $S$) by
\begin{multline}\label{eq:sparseselfnormused}
    \sigma^2 \left[\log \left( \frac{\det \left(\frac{T}{16k} \sum_{t=1}^{k}\E  \left[(X_t)_S (X_t)_S^\T\right]+\sum_{t=1}^{T} (X_t)_S (X_t)_S^\T \right)}{\frac{T}{16k} \sum_{t=1}^{k}\E  \left[(X_t)_S (X_t)_S^\T\right]} \right) +\log (1/\delta) \right]
    \\
    \leq 
    \sigma^2 \left[\log \left( \frac{\det \left(\frac{T}{16k} \sum_{t=1}^{k}\E  \left[(X_t)_S (X_t)_S^\T\right]+ C'_{\mathsf{sys}} \sum_{t=1}^{T} \E \left[(X_t)_S (X_t)_S^\T\right] \right)}{\frac{T}{16k} \sum_{t=1}^{k}\E  \left[(X_t)_S (X_t)_S^\T\right]} \right) +\log (1/\delta) \right]
\end{multline}
where the first inequality is just \Cref{thm:self-normalized martingale} and the second follows from the fact that we have restricted to the event of \Cref{lem:sparseisomorph}. The inequality in \eqref{eq:sparseselfnormused} only upper bounds our desired quantity pointwise in $S$. We now argue exactly as in \eqref{eq:sparseselfnormused} but invoke the self-normalized martingale theorem ${p \choose 2s}$ times, once for each $S$. Hence the desired result holds on an event of  probability at least
\begin{equation}
    1-\left(1+{p \choose 2s} \right)\delta  \geq 1-2 \left(\frac{ep}{2s} \right)^{2s}\delta. 
\end{equation}
By redefining $\delta' = 2 \left(\frac{ep}{2s} \right)^{2s}\delta$ the result follows. \hfill $\blacksquare$

\subsection{Proofs for Nonlinear Identification}

\paragraph{The Lower Tail Revisited}

Under A1-A3 above in \Cref{sec:nonlinear}, the following exponential inequality takes the role of \Cref{thm:expineq}. 

\begin{lemma} \label{lem:selfnormoffexpineq}
    Impose A2-A3.    For every $\lambda \in \R_+$ and every $f\in \scrF_\star$ we have that:
    \begin{equation}
        \E \exp \left(-\lambda \sum_{t=1}^T \| f(X_t)\|^2_2   \right) \leq\exp\left(-\lambda \sum_{t=1}^T \E \| f(X_t)\|^2_2 + \frac{\lambda^2T }{2k}  \E \left( \sum_{t=1}^{k} \| f(X_t)\|^2_2\right)^2\right).
    \end{equation}
\end{lemma}

\begin{proof}
First, observe the elementary inequality $e^{-x} \leq 1 -x+ x^2/2$ valid for all $x\geq 0$. Using this and independence of the blocks, we obtain
\begin{equation}
\begin{aligned}
     \E \exp \left(-\lambda \sum_{t=1}^T \| f(X_t)\|^2_2   \right) &=  \prod_{j\in [T/k]} \E \exp \left(-\lambda \sum_{t=j k+1}^{(j+1)k } \| f(X_t)\|^2_2   \right)\\
     &\leq  \prod_{j\in [T/k]} \E \left(1-\lambda \sum_{t=j k+1}^{(j+1)k} \| f(X_t)\|^2_2 + \frac{\lambda^2}{2} \left( \sum_{t=jk+1}^{(j+1)k} \| f(X_t)\|^2_2\right)^2\right)\\
     &
     \leq\exp\left(-\lambda \sum_{t=1}^T \E \| f(X_t)\|^2_2 + \frac{\lambda^2}{2} \sum_{j\in [T/k]} \E \left( \sum_{t=j k+1}^{(j+1)k} \| f(X_t)\|^2_2\right)^2\right)
\end{aligned}
\end{equation}
by also using the inequality $1+x \leq e^x$, valid for all $x\in \R$. The result follows by invoking A2 (\iid\ blocks).
\end{proof}

\begin{lemma}
     Impose A2-A3.    For every $f \in\scrF_\star $ we have that:
    \begin{equation}
        \Pr \left( \sum_{t=1}^T\| f(X_t)\|^2_2 < \frac{1}{2}\sum_{t=1}^T \E \| f(X_t)\|^2_2  \right) \leq \exp\left( - \frac{T}{4k \times \mathrm{cond}_{\scrF}^2 } \right).
    \end{equation}

\end{lemma}

\begin{proof}
A Chernoff bound along with \Cref{lem:selfnormoffexpineq} yields the following estimate:
    \begin{equation}\label{eq:hypconstuff}
        \begin{aligned}
            &\Pr \left( \sum_{t=1}^T\| f(X_t)\|^2_2 < \frac{1}{2}\sum_{t=1}^T \E \| f(X_t)\|^2_2  \right)  \\
            &\leq\inf_{\lambda \geq 0} \exp\left(-\frac{\lambda}{2} \sum_{t=1}^T \E \| f(X_t)\|^2_2 + \frac{\lambda^2 T}{2 k}  \E \left( \sum_{t=1}^{k} \| f(X_t)\|^2_2\right)^2\right)\\
            &
            =\exp\left( - \frac{ T\left(\sum_{t=1}^k \E \| f(X_t)\|^2_2\right)^2}{4k\E \left( \sum_{t=1}^{k} \| f(X_t)\|^2_2\right)^2} \right).
        \end{aligned}
    \end{equation}
    Using  the  Cauchy-Schwarz inequality and \eqref{eq:hypcon}:
\begin{equation}\label{eq:hypconstuff2}
    \begin{aligned}
        \E \left( \sum_{t=1}^{k} \| f(X_t)\|^2_2\right)^2
        \leq \mathrm{cond}_{\scrF}^2 \:  \left(\sum_{t=1}^{k} \E \| f(X_t)\|^2_2\right)^2.
    \end{aligned}
\end{equation}
The result follows by combining \eqref{eq:hypconstuff} with \eqref{eq:hypconstuff2}.
\end{proof}

Consequently, invoking the union bound, we arrive at following lower uniform law.

\begin{proposition}\label{prop:loweruniformlaw}
    Impose A1-A3. We have that:
    \begin{equation}
        \Pr \left( \exists f \in \scrF_\star \ : \  \sum_{t=1}^T\| f(X_t)\|^2_2 < \frac{1}{2}\sum_{t=1}^T \E \| f(X_t)\|^2_2 \right) \leq |\scrF_\star|  \exp\left( - \frac{T}{4k \times \mathrm{cond}_{\scrF}^2 } \right).
    \end{equation}
\end{proposition}

\paragraph{Controlling the Supremum of the Empirical Process \eqref{eq:empprocess}}

We begin with pointwise control.

\begin{lemma}\label{lem:selfnormoffset}
Impose A4. For every $\lambda \in [0,1/8\sigma^2]$:
    \begin{equation}
        \E \exp \left(\lambda \left(\sum_{t=1}^{T} 4\langle V_t,  f(X_t)\rangle -  \sum_{t=1}^{T} \| f(X_t) \|^2_2\right) \right) \leq 1.
    \end{equation}
\end{lemma}

\begin{proof}
    The proof is mostly repeated application of the tower property of conditional expectation. Let us write $\E_{t}[\cdot]=\E[\cdot | X_{1:t}]$. Then:
    \begin{equation}
    \begin{aligned}
        &\E \exp \left(\lambda \left(\sum_{t=1}^{T} 4\langle V_t,  f(X_t)\rangle -  \sum_{t=1}^{T} \| f(X_t) \|^2_2\right) \right) \\
        &=\E \Bigg[\exp \left(\lambda \left(\sum_{t=1}^{T-1} 4\langle V_t,  f(X_t)\rangle -  \sum_{t=1}^{T-1} \| f(X_t) \|^2_2\right) \right)\\
        &\times \E_{T} \exp\left(\lambda \left( 4\langle V_T,  f(X_T)\rangle -   \| f(X_T) \|^2_2\right) \right)\Bigg]
        \\
        &\leq\E \Bigg[\exp \left(\lambda \left(\sum_{t=1}^{T-1} 4\langle V_t,  f(X_t)\rangle -  \sum_{t=1}^{T-1} \| f(X_t) \|^2_2\right) \right)\\
        &\times \E_{T} \exp\left( 8\lambda^2 \sigma^2 \| f(X_T) \|^2_2 -   \lambda\| f(X_T) \|^2_2\right)\Bigg]
        \\
        &\leq\E \Bigg[\exp \left(\lambda \left(\sum_{t=1}^{T-1} 4\langle V_t,  f(X_t)\rangle -  \sum_{t=1}^{T-1} \| f(X_t) \|^2_2\right) \right)\Bigg]\\
        &\leq\dots \leq 1
    \end{aligned}
    \end{equation}
    by repeated application of the tower property and using that $ 8\lambda^2 \sigma^2 \| f(X_t) \|^2_2 -   \lambda\| f(X_t) \|^2_2 \leq 0$  for every $t$ for the prescribed range of $\lambda$. 
\end{proof}

\begin{proposition}\label{prop:unionboundfiniteclass}
    Impose A1 and A4. For every $u \in \R_+$ we have that
    \begin{equation}
        \Pr \left( \max_{f \in \scrF_\star } \left\{\sum_{t=1}^{T} 4\langle V_t,  f(X_t)\rangle -  \sum_{t=1}^{T} \| f(X_t) \|^2_2\right\} >  u \right)\leq  |  \scrF_\star| \exp\left( \frac{-u}{8\sigma^2} \right).
    \end{equation}
\end{proposition}

\begin{proof}
Fix $u>0$. We write:
    \begin{equation}
        \begin{aligned}
           &\Pr \left( \max_{f \in \scrF_\star } \left\{\sum_{t=1}^{T} 4\langle V_t,  f(X_t)\rangle -  \sum_{t=1}^{T} \| f(X_t) \|_2^2\right\} >  u \right)\\
           &\leq | 
 \scrF_\star|\max_{f \in \scrF_\star } \Pr \left(  \left\{\sum_{t=1}^{T} 4\langle V_t,  f(X_t)\rangle -  \sum_{t=1}^{T} \| f(X_t) \|_2^2\right\} >  u \right)\\
 &
 \leq \inf_{\lambda \geq 0} |  \scrF_\star|\max_{f \in \scrF_\star }  \E\exp\left(-\lambda u +  \lambda \left\{\sum_{t=1}^{T} 4\langle V_t,  f(X_t)\rangle -  \sum_{t=1}^{T} \| f(X_t) \|_2^2\right\} \right)\\
 &
 \leq |  \scrF_\star| \exp\left( \frac{-u}{8\sigma^2} \right)
        \end{aligned}
    \end{equation}
    by choosing $\lambda = 1/8\sigma^2$ and invoking \Cref{lem:selfnormoffset}.
\end{proof}

\paragraph{Finishing the Proof of \Cref{thm:nonlinearthm}}

\Cref{thm:nonlinearthm} easily follows by combining \Cref{prop:loweruniformlaw} with \Cref{prop:unionboundfiniteclass}. Namely, on the event of \Cref{prop:loweruniformlaw}, which holds after the burn-in of \Cref{thm:nonlinearthm}, it suffices to control the supremum of the empirical process \eqref{eq:nonparamoffset}. This control is the content of \Cref{prop:unionboundfiniteclass}. The result follows by setting the respective failure probabilities to $\delta$, and rescaling $\delta$. \hfill $\blacksquare$